%% file: new_main.tex
\documentclass[conference,a4paper]{IEEEtran}
\usepackage{hyperref}
\hypersetup{
	colorlinks=true,
	linkcolor=black,
	citecolor=black
}
\pdfoutput=1
\input{defs}

\usepackage{graphics} 
\usepackage{epsfig} 
\usepackage{times} 
\usepackage{amsmath, amssymb, amsbsy}
\usepackage{amsthm}
\usepackage{cite}
\usepackage{soul}
\usepackage{pgfplots}
\usepackage{pgfplotstable}
\usepackage{pgf, tikz, pgfplots}
\usepackage{slashbox}
\usepgfplotslibrary{colormaps}
\usepackage[mathscr]{eucal}
\usepackage{listings}
\usepackage[super]{nth}
\usepackage[normalem]{ulem}
\usepackage{float}
\usepackage{chngcntr}
\usepackage{flushend}
\counterwithout{figure}{section}

\newtheorem{lem}{Lemma}
\newtheorem{prop}{Proposition}
\newtheorem{cor}{Corollary}

\newtheorem{const}{Construction}
\newtheorem{thmmystyle}{Theorem}
\newtheorem{rem}{Remark}
\newtheorem{examplex}{Example}
\newtheorem{innercustomthm}{Theorem}
\newenvironment{customthm}[1]
{\renewcommand\theinnercustomthm{#1}\innercustomthm}
{\endinnercustomthm}
\newtheorem{innercustomconst}{Construction}
\newenvironment{customconst}[1]
{\renewcommand\theinnercustomconst{#1}\innercustomconst}
{\endinnercustomthm}
\usetikzlibrary{shapes}
\usetikzlibrary{shapes.multipart,chains}
\makeatletter
\newcommand{\removelatexerror}{\let\@latex@error\@gobble}

\makeatother
\makeatletter
\newcommand*{\rom}[1]{\expandafter\@slowromancap\romannumeral #1@}
\makeatother

\usepackage{tabto}  
\newcommand*{\QEDA}{\hfill\ensuremath{\blacktriangleright}}

\newcommand{\pow}{\mu}
\newcommand{\vmes}{\ve{m}}
\newcommand{\mes}{m}
\newcommand{\mesp}{m(x)}

\definecolor{capri}{rgb}{0.0, 0.75, 1.0}

\IEEEoverridecommandlockouts
\pgfkeys{
	/tr/rowfilter/.style 2 args={
		/pgfplots/x filter/.append code={
			\edef\arga{\thisrow{#1}}
			\edef\argb{#2}
			\ifx\arga\argb
			\else
			
			\fi
		}
	}
}
\pgfplotscreateplotcyclelist{solidlinesbyhaider7}{
		color=red,solid,line width = 2pt
	\\
	color=cyan,line width = 2pt, dotted 
	\\
	color=green,line width = 2pt, loosely dash dot, dash pattern=on 2mm off 1mm on 0.25mm off 1mm 
	\\
	color=blue,line width = 2pt, dash dot, dash pattern=on 4mm off 1mm on .5mm off 1mm
	\\
	color=black,line width = 2pt , dash pattern=on 3mm off 1mm 
	\\
	color=red,solid,line width = 2pt
	\\
	color=purple,solid
	\\
	color=brown,solid,thick
	\\
	color=green!70!black,solid,thick
	\\
	color=yellow,solid, thick
	\\
	color=gray,solid,thick
	\\
	color=teal,solid,thick
	\\
	color=orange,solid,thick
	\\
}
\pgfplotscreateplotcyclelist{solidlinesbyhaider72}{
	color=violet,line width = 1pt, densely dashed 
	\\
	color=cyan,line width = 1pt, dotted 
	\\
	color=black,line width = 1pt , dash pattern=on 3mm off 1mm 
	\\
	color=green,line width = 1pt, loosely dash dot, dash pattern=on 1.5mm off 1mm on 0.25mm off 1mm 
	\\
	color=blue,line width = 1pt, dash dot, dash pattern=on 3mm off 1mm on .25mm off 1mm
	\\
	color=red,solid,line width = 1pt
	\\
	color=purple,solid
	\\
	color=brown,solid,thick
	\\
	color=green!70!black,solid,thick
	\\
	color=yellow,solid, thick
	\\
	color=gray,solid,thick
	\\
	color=teal,solid,thick
	\\
	color=orange,solid,thick
	\\
}
\pgfplotscreateplotcyclelist{solidlinesbyhaider21}{
	color=red,solid,thick
	\\
	color=purple,solid,thick
	\\
	color=gray,solid,thick
	\\
	color=green!70!black,solid,thick
	\\
	color=yellow,solid, thick
	\\
	color=brown,solid,thick
	\\
	color=teal,solid,thick
	\\
	color=orange,solid,thick
	\\
	color=violet,solid,thick
	\\
	color=cyan,solid,thick
	\\
	color=green,solid,thick
	\\
	color=blue,line width = 1pt, dash dot, dash pattern=on 3mm off 1mm on .5mm off 1.6mm
	\\
	color=red,dashed,thick
	\\
	color=purple,dashed,thick
	\\
	color=gray,dashed,,thick
	\\
	color=green!70!black,dashed,,thick
	\\
	color=yellow,dashed, ,thick
	\\
	color=brown,dashed,,thick
	\\
	color=teal,dashed,,thick
	\\
	color=orange,dashed,,thick
	\\
	color=violet,dashed,,thick
	\\
	color=cyan,dashed,,thick
	\\
	color=green,dashed,,thick
	\\
	color=red,dashdotted
	\\
	color=purple,dashdotted
	\\
	color=gray,dashdotted,semithick
	\\
	color=green!70!black,dashdotted,semithick
	\\
	color=yellow,dashdotted, semithick
	\\
	color=brown,dashdotted,semithick
	\\
	color=teal,dashdotted,semithick
	\\
	color=orange,dashdotted,semithick
	\\
	color=violet,dashdotted,semithick
	\\
	color=cyan,dashdotted,semithick
	\\
	color=green,dashdotted,semithick
	\\
}
\pgfplotscreateplotcyclelist{solidlinesbyhaider14}{
	color=blue,solid,line width = 1pt
	\\
	color=purple,solid,line width = 1pt
	\\
	color=orange,solid,line width = 1pt
	\\
	color=green!50!black,solid,line width = 1pt
	\\
	color=gray,solid,line width = 1pt
	\\
	color=orange,dotted,line width = 1pt
	\\
	color=green!50!black,dotted,line width = 1pt
	\\
	color=gray,dotted,line width = 1pt
	\\
	color=purple,dashed,line width = 1pt
	\\
	color=purple,line width = 1pt, dash dot,
	\\
}
\pgfplotscreateplotcyclelist{solidlinesbyhaider12}{
	color=orange,solid,line width = 1pt
	\\
	color=green!50!black,solid,line width = 1pt
	\\
	color=gray,solid,line width = 1pt
	\\
	color=orange,dotted,line width = 1pt
	\\
	color=green!50!black,dotted,line width = 1pt
	\\
	color=gray,dotted,line width = 1pt
	\\
	color=purple,dashed,line width = 1pt
	\\
	color=purple,line width = 1pt, dash dot,
	\\
}
\pgfplotscreateplotcyclelist{solidlinesbyhaider11}{
color=green,line width = 1pt, loosely dash dot, dash pattern=on 2mm off 1mm on 0.25mm off 1mm 
\\
color=red,dotted,thick
\\
color=purple,dotted,thick
\\
color=gray,dotted,thick
\\
color=green!70!black,dotted,thick
\\
color=yellow,dotted, thick
\\
color=brown,dotted,thick
\\
color=teal,dotted,thick
\\
color=orange,dotted,thick
\\
color=violet,dotted,thick
\\
color=cyan,dotted,semithick
\\
color=blue,line width = 1pt, dash dot, dash pattern=on 3mm off 1mm on .5mm off 1.6mm
\\
color=purple,solid
\\
color=gray,solid,semithick
\\
color=green!70!black,solid,semithick
\\
color=yellow,solid, semithick
\\
color=brown,solid,semithick
\\
color=teal,solid,semithick
\\
color=orange,solid,semithick
\\
color=violet,solid,semithick
\\
color=cyan,solid,semithick
\\
color=green,solid,semithick
\\
color=red,dashed
\\
color=purple,dashed
\\
color=gray,dashed,semithick
\\
color=green!70!black,dashed,semithick
\\
color=yellow,dashed, semithick
\\
color=brown,dashed,semithick
\\
color=teal,dashed,semithick
\\
color=orange,dashed,semithick
\\
color=violet,dashed,semithick
\\
color=cyan,dashed,semithick
\\
color=green,dashed,semithick
\\
}
\pgfplotscreateplotcyclelist{solidlinesbyhaider9}{
	color=black,line width = 1pt , dash pattern=on 3.4mm off 1mm
	\\
	color=red,dashed, line width = 1pt
	\\
	color=teal,dashed, line width = 1pt
	\\
	color=brown,dashed,line width = 1pt
	\\
	color=red,solid, line width = 1pt
	\\
	color=teal,solid, line width = 1pt
	\\
	color=brown,solid,line width = 1pt
	\\
}
\pgfplotscreateplotcyclelist{solidlinesbyhaider8}{
	color=red,solid
	\\
	color=teal,solid
	\\
	color=gray,solid,semithick
	\\
	color=blue!70!white,solid,semithick
	\\
	color=yellow,solid, semithick
	\\
	color=brown,solid,semithick
	\\
	color=teal,solid,semithick
	\\
	color=orange,solid,semithick
	\\
	color=red,dashed
	\\
	color=red,dashed
	\\
	color=gray,dashed,semithick
	\\
	color=blue!70!white,dashed,semithick
	\\
	color=yellow,dashed, semithick
	\\
	color=brown,dashed,semithick
	\\
	color=teal,dashed,semithick
	\\
	color=orange,dashed,semithick
	\\
}
\begin{document}

\title{Coding and Bounds for\\ Partially Defective Memory Cells \vspace{-0.4 cm}}
\author{\IEEEauthorblockN{Haider Al Kim$^{1,2}$\thanks{
This work has received funding from the German Research Foundation (Deutsche Forschungsgemeinschaft, DFG) under Grant No. WA3907/1-1, the German Academic Exchange Service (Deutscher Akademischer Austauschdienst, DAAD) under the support program ID 57381412, and the European Union's Horizon 2020 research and innovation program through the Marie Sklodowska-Curie under Grant No.~713683. 
This article was presented in part at the 2019 International Symposium Problems of Redundancy in Information and Control Systems (Redundancy), at the 17th International Workshop on Algebraic and Combinatorial Coding Theory (ACCT2020), and also as an extended abstract at the 2021 annual Non-Volatile Memories Workshop (NVMW).
}, Sven Puchinger$^{1,3}$, Ludo Tolhuizen$^{4}$, Antonia Wachter-Zeh$^{1}$}
\IEEEauthorblockA{
 $^1$Institute for Communications Engineering, Technical University of Munich (TUM), Germany\\ $^2$Electronic and Communications Engineering, University of Kufa (UoK), Iraq\\
 $^3$Department of Applied Mathematics and Computer Science, Technical University of Denmark (DTU), Denmark\\
 $^4$Philips Research, High Tech Campus 34, Netherlands\\
 Email: haider.alkim@tum.de, sven.puchinger@tum.de, ludo.tolhuizen@philips.com, antonia.wachter-zeh@tum.de}
 \vspace{-0.5cm}}
\maketitle
\begin{abstract}
		This paper considers coding for so-called \emph{partially stuck (defect)} memory cells.
	Such memory cells can only store partial information as some of their levels cannot be used fully due to, e.g., wearout.
	First, we present new constructions that are able to mask $u$ partially stuck cells while correcting at the same time $t$ random errors. 
	The process of ''masking'' determines a word whose entries coincide with writable levels at the (partially) stuck cells. 
For $u>1$ and alphabet size $q>2$, our new constructions improve upon the required redundancy of known constructions for $t=0$, and require less redundancy for masking partially stuck cells than former works required for masking fully stuck cells (which cannot store any information). 
Second, we show 
that treating some of the partially stuck cells as erroneous cells can decrease the required redundancy for some parameters. 
Lastly, we derive Singleton-like, sphere-packing-like, and Gilbert--Varshamov-like bounds. 
Numerical comparisons 
state that our constructions match the Gilbert--Varshamov-like bounds for several code parameters, e.g., BCH codes that contain all-one word by our first construction. 
\end{abstract}
\begin{IEEEkeywords}
	flash memories, phase change memories, non-volatile memories, defective memory, (partially) stuck cells, BCH code, cyclic code, sphere packing bound, Gilbert-Varshamov bound
 \vspace{-0.1cm}
\end{IEEEkeywords}
\section{Introduction}
\IEEEPARstart{T}{he} demand for reliable memory solutions and in particular for non-volatile memories such as flash memory and \emph{phase change memories} (PCMs) for different applications is steadily increasing. These multi-level devices provide permanent storage and a rapidly extendable capacity. Recently developed devices exploit an increased number of cell levels while at the same time the physical size of the cells was decreased.
Therefore, coding and signal processing solutions are essential to overcome reliability issues. The key characteristic of PCM cells is that they can switch between two main states: an amorphous state and a crystalline state. PCM cells may become \emph{defect} (also called \emph{stuck}) \cite{gleixner2009reliability,kim2005reliability,lee2009study,pirovano2004reliability} if they fail in switching their states.
This occasionally happens due to the cooling and heating processes of the cells. Therefore, cells can only hold a single phase \cite{gleixner2009reliability},\cite{pirovano2004reliability}. 
In multi-level PCM cells, failure may occur at a position in either of extreme states or in the partially programmable states of crystalline. 

The work~\cite{wachterzeh2016codes} investigates codes that mask so-called \emph{partially stuck} (partially defective) cells, i.e., cells which cannot use all levels.
For multi-level PCMs, the case in which the partially stuck level $s=1$ is particularly important since this means that a cell can reach all crystalline sub-states, but cannot reach the amorphous state.

Figure~\ref{Fig1} depicts the general idea of reliable and (partially) defective memory cells.
It shows two different cell level representations: Representation~1 forms the binary extension filed $\mathbb{F}_{2^4}$ and Representation~2 forms the set of integers modulo $q=4$, i.e., $\mathbb{Z}/4\mathbb{Z}$.
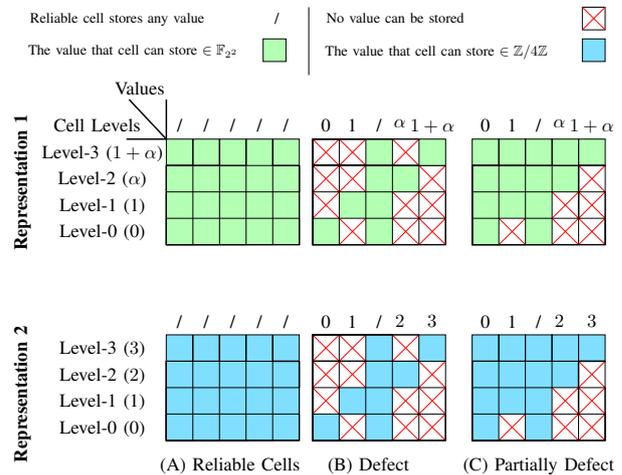
\begin{figure} [h]
	
	\begin{center}
	\scalebox{0.6}{	
		\begin{tikzpicture}
			\draw
			(9.1,4) node[anchor=north] {\small {No value can be stored}} 			
			;
			\draw(3.35,3.3) node[anchor=north] {\small {The value that cell can store $\in \mathbb{F}_{2^2}$}} 
			;
			\draw(10,3.3) node[anchor=north] {\small {The value that cell can store $\in \mathbb{Z}/4\mathbb{Z}$}} 
			;		
			\fill[capri!50!white] (13.2,2.8) rectangle (13.7,3.3);
			\fill[green!30!white] (6.2,2.8) rectangle (6.7,3.3);
			\draw (13.2,2.8) rectangle (13.7,3.3);
			\draw (13.2,3.5) rectangle (13.7,4);
			\draw (6.2,2.8) rectangle (6.7,3.3);
			\draw (6.5,4) node[anchor=north] {/}		
			(3,4) node[anchor=north] {\small {Reliable cell stores any value}}
			;
			\draw[red,-] (13.2,3.5) -- (13.7,4);
			\draw[red,-] (13.2,4) -- (13.7,3.5);
			\draw[black,-] (7.2,2.5) -- (7.2,4);	
		\end{tikzpicture}
	}	

	\scalebox{0.70}{		
		\begin{tikzpicture}[baseline={(0,-0.5)}]	
			\draw
			(-1.2,2) node[anchor=north] {Level-3 ($1+\alpha$)}
			(-1.2,1.5) node[anchor=north] {Level-2 ($\alpha$)}
			(-1.2,1) node[anchor=north] {Level-1 (1)}
			(-1.2,0.5) node[anchor=north] {Level-0 (0)}
			(0.25,2.5) node[anchor=north] {/}
			(0.75,2.5) node[anchor=north] {/}
			(1.25,2.5) node[anchor=north] {/}
			(1.75,2.5) node[anchor=north] {/}
			(2.25,2.5) node[anchor=north] {/}
			;
			\draw[thick,-] (0,1.5) -- (2.5,1.5);
			\draw (-2.7,2.6) node[anchor=north] {\rotatebox{90}{\textbf{Representation~1}}};
			\draw (-0.5,3.2) node[anchor=north] {Values};
			\draw (-1.3,2.5) node[anchor=north] {Cell Levels};
			\draw[thick,-] (0,2) -- (-0.7,2.75);
			\draw[thick,-] (0,2) -- (0,2.75);
			\draw[thick,-] (0,2) -- (-0.75,2);
			\fill[green!30!white] (0,0) rectangle (2.5,2);
			\draw (0,0) rectangle (2.5,2);
			\draw (0,0) rectangle (2.5,0.5);
			\draw (0,0) rectangle (2.5,1);
			\draw (0,0) rectangle (0.5,1.5);
			\draw (0,0) rectangle (1,1.5);
			\draw (0,0) rectangle (1.5,1.5);
			\draw (0,0) rectangle (2,1.5);
			\draw[thick,-] (2,1.5) -- (2.5,1.5);
			\draw[thick,-] (2.5,1) -- (2.5,1.5);
			\draw[thick,-] (0.5,1.5) -- (0.5,2);
			\draw[thick,-] (1,1.5) -- (1,2);
			\draw[thick,-] (1.5,1.5) -- (1.5,2);
			\draw[thick,-] (2,1.5) -- (2,2);
			\draw
			(3,2.5) node[anchor=north] {0}
			(3.5,2.5) node[anchor=north] {1}
			(4,2.5) node[anchor=north] {/}
			(4.39,2.5) node[anchor=north] {$\alpha$}
			(5,2.5) node[anchor=north] {$1+\alpha$};
			\draw (2.75,0) rectangle (5.25,1.5);
			\fill[green!30!white] (2.75,0) rectangle (3.25,0.5);
			\fill[green!30!white] (3.25,0.5) rectangle (3.75,1);
			\fill[green!30!white] (3.75,0) rectangle (4.25,2);
			\fill[green!30!white] (4.25,1) rectangle (4.75,1.5);
			\fill[green!30!white] (4.75,1.5) rectangle (5.25,2);
			\draw[red,-] (2.75,0.5) -- (3.25,1);
			\draw[red,-] (2.75,1) -- (3.25,0.5);
			\draw[red,-] (2.75,1) -- (3.25,1.5);
			\draw[red,-] (2.75,1.5) -- (3.25,1);
			\draw[red,-] (2.75,1.5) -- (3.25,2);
			\draw[red,-] (2.75,2) -- (3.25,1.5);
			\draw[red,-] (3.25,0) -- (3.75,0.5);
			\draw[red,-] (3.25,0.5) -- (3.75,0);
			\draw[red,-] (3.25,1) -- (3.75,1.5);
			\draw[red,-] (3.25,1.5) -- (3.75,1);
			\draw[red,-] (3.25,1.5) -- (3.75,2);
			\draw[red,-] (3.25,2) -- (3.75,1.5);
			\draw[red,-] (4.25,0) -- (4.75,0.5);
			\draw[red,-] (4.25,0.5) -- (4.75,0);
			\draw[red,-] (4.25,1.5) -- (4.75,2);
			\draw[red,-] (4.25,2) -- (4.75,1.5);
			\draw[red,-] (4.25,0.5) -- (4.75,1);
			\draw[red,-] (4.25,1) -- (4.75,0.5);
			\draw[red,-] (4.75,0) -- (5.25,0.5);
			\draw[red,-] (4.75,0.5) -- (5.25,0);
			\draw[red,-] (4.75,0.5) -- (5.25,1);
			\draw[red,-] (4.75,1) -- (5.25,0.5);
			\draw[red,-] (4.75,1) -- (5.25,1.5);
			\draw[red,-] (4.75,1.5) -- (5.25,1);
			\draw (2.75,1.5) rectangle (3.25,2);
			\draw (2.75,1.5) rectangle (3.75,2);
			\draw (2.75,1.5) rectangle (4.25,2);
			\draw (2.75,1.5) rectangle (4.75,2);
			\draw (2.75,1.5) rectangle (5.25,2);
			\draw (2.75,0) rectangle (5.25,0.5);
			\draw (2.75,0) rectangle (5.25,1);
			\draw (2.75,0) rectangle (3.25,1.5);
			\draw (2.75,0) rectangle (3.75,1.5);
			\draw (2.75,0) rectangle (4.25,1.5);
			\draw (2.75,0) rectangle (4.75,1.5);
			;	
			\draw	
			(6,2.5) node[anchor=north] {0}
			(6.5,2.5) node[anchor=north] {1}
			(7,2.5) node[anchor=north] {/}
			(7.39,2.5) node[anchor=north] {$\alpha$}
			(8,2.5) node[anchor=north] {$1+\alpha$};
			\draw[thick,-] (6,1.5) -- (6,1.5);
			\draw (5.75,0) rectangle (7.75,1.5);
			\fill[green!30!white] (5.75,0) rectangle (6.25,0.5);
			\fill[green!30!white] (6.25,0.5) rectangle (6.75,1);
			\fill[green!30!white] (6.75,0) rectangle (7.25,1.5);
			\fill[green!30!white] (5.75,1) rectangle (7.75,1.5);
			\fill[green!30!white] (5.75,0.5) rectangle (6.25,1);
			\draw[thick,-] (7.75,1) -- (7.75,1.5);
			\fill[green!30!white] (7.75,0.5) rectangle (7.75,1);
			\fill[green!30!white] (5.75,1.5) rectangle (8.25,2);
			\draw[red,-] (6.25,0) -- (6.75,0.5);
			\draw[red,-] (6.25,0.5) -- (6.75,0);
			\draw[red,-] (7.25,0) -- (7.75,0.5);
			\draw[red,-] (7.25,0.5) -- (7.75,0);
			\draw[red,-] (7.25,0.5) -- (7.75,1);
			\draw[red,-] (7.25,1) -- (7.75,0.5);
			\draw[red,-] (7.75,0) -- (8.25,0.5);
			\draw[red,-] (7.75,0.5) -- (8.25,0);	
			\draw[red,-] (7.75,0.5) -- (8.25,1);
			\draw[red,-] (7.75,1) -- (8.25,0.5);
			\draw[red,-] (7.75,1) -- (8.25,1.5);
			\draw[red,-] (7.75,1.5) -- (8.25,1);
			\draw (5.75,0) rectangle (8.25,0.5);
			\draw (5.75,0) rectangle (8.25,1);
			\draw (5.75,0) rectangle (6.25,1.5);
			\draw (5.75,0) rectangle (6.75,1.5);
			\draw (5.75,0) rectangle (7.25,1.5);
			\draw (5.75,0) rectangle (8.25,1.5);
			\draw[thick,-] (7.75,1.5) -- (8.25,1.5);
			\draw[thick,-] (7.75,.5) -- (7.75,1);
			\draw (5.75,1.5) rectangle (6.25,2);
			\draw (6.25,1.5) rectangle (6.75,2);
			\draw (6.75,1.5) rectangle (7.25,2);
			\draw (7.25,1.5) rectangle (7.75,2);
			\draw (7.75,1.5) rectangle (8.25,2);		
		\end{tikzpicture}
	}

	\scalebox{0.70}{	
		\begin{tikzpicture}[baseline={(0,-0.5)}]	
			\draw
			(-1.2,2) node[anchor=north] {Level-3 ($3$)}
			(-1.2,1.5) node[anchor=north] {Level-2 ($2$)}
			(-1.2,1) node[anchor=north] {Level-1 (1)}
			(-1.2,0.5) node[anchor=north] {Level-0 (0)}
			(0.25,2.5) node[anchor=north] {/}
			(0.75,2.5) node[anchor=north] {/}
			(1.25,2.5) node[anchor=north] {/}
			(1.75,2.5) node[anchor=north] {/}
			(2.25,2.5) node[anchor=north] {/}	
			;
			\draw[thick,-] (0,1.5) -- (2.5,1.5);
			\draw (-2.7,2.3) node[anchor=north] {\rotatebox{90}{\textbf{Representation~2}}};
			\draw (-0.5,3.2) node[anchor=north] {};
			\draw (-1.3,2.5) node[anchor=north] {};
			\fill[capri!50!white] (0,0) rectangle (2.5,2);
			\draw (0,0) rectangle (2.5,2);
			\draw (0,0) rectangle (2.5,0.5);
			\draw (0,0) rectangle (2.5,1);
			\draw (0,0) rectangle (0.5,1.5);
			\draw (0,0) rectangle (1,1.5);
			\draw (0,0) rectangle (1.5,1.5);
			\draw (0,0) rectangle (2,1.5);
			\draw[thick,-] (2,1.5) -- (2.5,1.5);
			\draw[thick,-] (2.5,1) -- (2.5,1.5);
			\draw[thick,-] (0.5,1.5) -- (0.5,2);
			\draw[thick,-] (1,1.5) -- (1,2);
			\draw[thick,-] (1.5,1.5) -- (1.5,2);
			\draw[thick,-] (2,1.5) -- (2,2);	
			\draw
			(3,2.5) node[anchor=north] {0}
			(3.5,2.5) node[anchor=north] {1}
			(4,2.5) node[anchor=north] {/}
			(4.39,2.5) node[anchor=north] {$2$}
			(5,2.5) node[anchor=north] {$3$};
			\draw (2.75,0) rectangle (5.25,1.5);
			\fill[capri!50!white] (2.75,0) rectangle (3.25,0.5);
			\fill[capri!50!white] (3.25,0.5) rectangle (3.75,1);
			\fill[capri!50!white] (3.75,0) rectangle (4.25,2);
			\fill[capri!50!white] (4.25,1) rectangle (4.75,1.5);
			\fill[capri!50!white] (4.75,1.5) rectangle (5.25,2);
			\draw[red,-] (2.75,0.5) -- (3.25,1);
			\draw[red,-] (2.75,1) -- (3.25,0.5);
			\draw[red,-] (2.75,1) -- (3.25,1.5);
			\draw[red,-] (2.75,1.5) -- (3.25,1);
			\draw[red,-] (2.75,1.5) -- (3.25,2);
			\draw[red,-] (2.75,2) -- (3.25,1.5);
			\draw[red,-] (3.25,0) -- (3.75,0.5);
			\draw[red,-] (3.25,0.5) -- (3.75,0);
			\draw[red,-] (3.25,1) -- (3.75,1.5);
			\draw[red,-] (3.25,1.5) -- (3.75,1);
			\draw[red,-] (3.25,1.5) -- (3.75,2);
			\draw[red,-] (3.25,2) -- (3.75,1.5);
			\draw[red,-] (4.25,0) -- (4.75,0.5);
			\draw[red,-] (4.25,0.5) -- (4.75,0);
			\draw[red,-] (4.25,1.5) -- (4.75,2);
			\draw[red,-] (4.25,2) -- (4.75,1.5);
			\draw[red,-] (4.25,0.5) -- (4.75,1);
			\draw[red,-] (4.25,1) -- (4.75,0.5);
			\draw[red,-] (4.75,0) -- (5.25,0.5);
			\draw[red,-] (4.75,0.5) -- (5.25,0);
			\draw[red,-] (4.75,0.5) -- (5.25,1);
			\draw[red,-] (4.75,1) -- (5.25,0.5);
			\draw[red,-] (4.75,1) -- (5.25,1.5);
			\draw[red,-] (4.75,1.5) -- (5.25,1);
			\draw (2.75,1.5) rectangle (3.25,2);
			\draw (2.75,1.5) rectangle (3.75,2);
			\draw (2.75,1.5) rectangle (4.25,2);
			\draw (2.75,1.5) rectangle (4.75,2);
			\draw (2.75,1.5) rectangle (5.25,2);
			\draw (2.75,0) rectangle (5.25,0.5);
			\draw (2.75,0) rectangle (5.25,1);
			\draw (2.75,0) rectangle (3.25,1.5);
			\draw (2.75,0) rectangle (3.75,1.5);
			\draw (2.75,0) rectangle (4.25,1.5);
			\draw (2.75,0) rectangle (4.75,1.5);
			;	
			\draw	
			(6,2.5) node[anchor=north] {0}
			(6.5,2.5) node[anchor=north] {1}
			(7,2.5) node[anchor=north] {/}
			(7.39,2.5) node[anchor=north] {$2$}
			(8,2.5) node[anchor=north] {$3$};
			\draw[thick,-] (6,1.5) -- (6,1.5);
			\draw (5.75,0) rectangle (7.75,1.5);
			\fill[capri!50!white] (5.75,0) rectangle (6.25,0.5);
			\fill[capri!50!white] (6.25,0.5) rectangle (6.75,1);
			\fill[capri!50!white] (6.75,0) rectangle (7.25,1.5);
			\fill[capri!50!white] (5.75,1) rectangle (7.75,1.5);
			\fill[capri!50!white] (5.75,0.5) rectangle (6.25,1);
			\draw[thick,-] (7.75,1) -- (7.75,1.5);
			\fill[capri!50!white] (7.75,0.5) rectangle (7.75,1);
			\fill[capri!50!white] (5.75,1.5) rectangle (8.25,2);
			\draw[red,-] (6.25,0) -- (6.75,0.5);
			\draw[red,-] (6.25,0.5) -- (6.75,0);
			\draw[red,-] (7.25,0) -- (7.75,0.5);
			\draw[red,-] (7.25,0.5) -- (7.75,0);
			\draw[red,-] (7.25,0.5) -- (7.75,1);
			\draw[red,-] (7.25,1) -- (7.75,0.5);
			\draw[red,-] (7.75,0) -- (8.25,0.5);
			\draw[red,-] (7.75,0.5) -- (8.25,0);
			\draw[red,-] (7.75,0.5) -- (8.25,1);
			\draw[red,-] (7.75,1) -- (8.25,0.5);
			\draw[red,-] (7.75,1) -- (8.25,1.5);
			\draw[red,-] (7.75,1.5) -- (8.25,1);
			\draw (5.75,0) rectangle (8.25,0.5);
			\draw (5.75,0) rectangle (8.25,1);
			\draw (5.75,0) rectangle (6.25,1.5);
			\draw (5.75,0) rectangle (6.75,1.5);
			\draw (5.75,0) rectangle (7.25,1.5);
			\draw (5.75,0) rectangle (8.25,1.5);
			\draw[thick,-] (7.75,1.5) -- (8.25,1.5);
			\draw[thick,-] (7.75,.5) -- (7.75,1);
			\draw (5.75,1.5) rectangle (6.25,2);
			\draw (6.25,1.5) rectangle (6.75,2);
			\draw (6.75,1.5) rectangle (7.25,2);
			\draw (7.25,1.5) rectangle (7.75,2);
			\draw (7.75,1.5) rectangle (8.25,2);
			\draw
			(1.2,-0.2) node[anchor=north] {(A) Reliable Cells} ;
			\draw (3.8,-0.2) node[anchor=north] {(B) Defect} 
			(7,-0.2) node[anchor=north] {(C) Partially Defect};
		\end{tikzpicture}
	}
	\caption{Illustration of reliable and (partially) defective memory cells. In this figure, there are $n=5$ cells with $q=4$ possible levels. The cell levels $\in \mathbb{F}_{2^2}$ are mapped to (0, 1, $\alpha$ or $1+\alpha$) shown in Representation~1 or $\in \mathbb{Z}/4\mathbb{Z}$ are mapped to ($0,1,2$ or $3$) shown in Representation~2. Case (A) illustrates fully reliable cells which can store any of the four values in both representations. In the stuck scenario as shown in case (B), the defective cells can store only the exact stuck level $s$. Case (C) is more flexible (partially defective scenario). Partially stuck cells at level $s \geq 1$ can store level $s$ or higher. }
	\label{Fig1}
	\end{center}
\end{figure}
\subsection{Related Work} 
Coding for memories with stuck cells, also known as \emph{defect-correcting codes} for \emph{memories with defects}, dates back to the 1970s, cf. the work by Kuznetsov and Tsybakov \cite{Kuznetsov1974Tsybakov}.
They proposed binary defect-correcting codes in finite and asymptotic regimes whose required redundancy is at least the number of defects. Later works \cite{Tsybakov1975etal,Tsybakov1975defect-correcting,Tsybakov1975defects-and-error-correction,Belov1977Shashin,Losev1978Konopel,Kuznetsov1978KasamiandYamamura,heegard1983partitioned,Kuznetsov1985,Chen1985,Borden1987Vinck,dumer1987aditivecoding,dumer1989asymptotically,dumer1990asymptotically} investigated the problem of defective cells under various aspects: binary and non-binary, only defect-correcting coding and error-and-defect-correcting coding, and finite and asymptotic length analysis. 

In binary defect-correcting coding models, e.g. \cite{Tsybakov1975defect-correcting,Losev1978Konopel,Kuznetsov1978KasamiandYamamura,Chen1985,Belov1977Shashin,Borden1987Vinck,dumer1987aditivecoding}, the authors dealt with masking stuck cells without considering additional substitution errors. In these studies, it is unclear if the proposed constructions are optimal in terms of their required redundancy. 
The works \cite{Tsybakov1975defects-and-error-correction,heegard1983partitioned,Kuznetsov1985} considered masking stuck memory cells while at the same time correcting potential random errors. In \cite{heegard1983partitioned}, so-called partitioned cyclic code and partitioned BCH codes were proposed for this task. 

The asymptotic model of stuck-cell-masking codes also received considerable attention in the previously mentioned papers. 
Moreover, there is work devoted to asymptotically optimal codes for a fixed number of defects \cite{dumer1989asymptotically} or for a number of defects proportional to the codeword length \cite{dumer1990asymptotically}. 
The proposed constructions, for example \cite [Section 4]{dumer1990asymptotically} and its extended version in \cite [Section 5]{dumer1990asymptotically} that can additionally correct substitution errors, show that $u$ check symbols are sufficient for masking $u$ defects. However, they use codes with a property that is not well studied in coding theory. Therefore, we do not dwell on \cite{dumer1989asymptotically} and \cite{dumer1990asymptotically}, and also our goal is to obtain code constructions for finite code length $n$.

The recent work \cite{wachterzeh2016codes} considers \emph{partially} stuck memory cells (see Figure~1. C), and improves upon the redundancy necessary for masking compared to all prior works for classical stuck cells. 
However, the paper does not consider error correction in addition to masking. 
\subsection{Our Contribution}
In this paper, 
we consider the problem of combined error correction and masking of partially stuck cells. Compared to the conventional stuck-cell case in \cite{heegard1983partitioned}, we reduce the redundancy necessary for masking, similar to the results in \cite{wachterzeh2016codes}, and even reduce further compared to \cite[Construction 5]{wachterzeh2016codes}. 

If cells are partially stuck at level $1$, we can simply use a $(q-1)$-ary error correcting code as mentioned in \cite[Section III]{wachterzeh2016codes}.
However, this approach could require too much redundancy if a cell is partially stuck at different levels rather than $1$.
For instance, using $(q-s)$-ary codes for $ 2 \leq s \leq q-1$ reduces the cardinality of the code because exempting $s$ out of the available $q$ levels is quite expensive. 
Further, for relatively few partially stuck-at-$1$ cells, even a $(q-1)$-ary error correcting code is not a competitor to 
our constructions 
(cf. Figure~\ref{compare_with_ordinary_GV}).
Therefore, considering sophisticated coding schemes is favorable.

We provide code constructions for any number of partially stuck cells; see Table~\ref{table1} for an overview of our constructions and their required redundancies. 
For the error-free case, where only masking is necessary, our redundancies coincide with those from \cite{wachterzeh2016codes} or are even smaller. 

Our paper also investigates a technique where the encoder, after a first masking step, introduces errors at some partially stuck positions of a codeword in order
to satisfy the stuck-at constraints. The decoder uses part of the error-correcting capability to correct these introduced errors.

We also derive bounds on our code constructions,
namely a Singleton-type, sphere-packing-type, and Gilbert-Varshamov-type bounds.
We provide a numerical analysis by comparing our code constructions and the derived bounds with other trivial codes and known limits.

Our focus is on long codes over small alphabets, i.e., the code length $n$ is larger than the field size $q$. Otherwise, one could instead mask by a code of length $n<q$ (by using, e.g., \cite{SOLOMON1974395}).

The remainder of the paper is arranged as follows. In Section~\ref{preliminaries}, we provide notations and define the models of joint errors and partially defective cells examined in this study. Our code constructions along with their encoding and decoding algorithms are presented in Section~\ref{u_less_than_q} and \ref{Constructions_u_leq_q_leq_n}. Section~\ref{arbitrary_s_levels} generalizes the previous constructions to mask partially stuck cells at any arbitrary level and correct errors additionally. Section~\ref{section_trading} investigates exchanging error correction capability toward more partially stuck cells masking possibility. Upper- and lower-like bounds on our constructions are derived in Section~\ref{upper_bounds_psmc} and \ref{lower_like_bounds}, respectively. In Section~\ref{comparisons}, we provide numerical and analytical comparisons. Finally, Section~\ref{Conclusion} concludes this work. 
\section{Preliminaries}\label{preliminaries}
\subsection{Notations }\label{ssec:notation}
For a prime power $q$, let $\mathbb{F}_q$ denote the finite field of order $q$ and ${\mathbb{F}}_q[x]$ be the set of all univariate polynomials with coefficients in $\mathbb{F}_q$. 
For $g,f \in \mathbb{Z}_{> 0}$, 
 denote $[f] = \{0,1, \dots, f-1\}$ and $[g,f]=\{g,g+1,\dots,f-1\}$. 

As usual, an $[n,k,d]_q$ code is a linear code over $\Fq$ of length $n$, dimension $k$ and minimum (Hamming) distance $d$. The (Hamming) weight $\wt(\ve{x})$ of a vector $\ve{x}\in\Fq^n$ equals its number of non-zero entries.

We fix throughout the paper a total ordering ``$\geq$'' of the elements of $\Fq$ such that $a \geq 1 \geq 0$ for all $a \in \Fq \setminus \{0\}$. 
	So $0$ is the smallest element in $\Fq$, and $1$ is the next smallest element in $\Fq$.
We extend the ordering on $\Fq$ to $\Fq^n$:
for $\vec{x}=(x_0,\ldots ,x_{n-1})\in\Fq^n$ and $\vec{y}=(y_0,\ldots, y_{n-1})\in\Fq^n$, we say that $\vec{x}\geq \vec{y}$ if and only if 
$x_i\geq y_i$ for all $i\in [n]$. 

In order to simplify notation, we sometimes identify $x\in\Fq$ with the number of field elements not larger than $x$, that is,
with the integer $q-|\{ y\in\Fq\mid x\geq y\}|$. The meaning of $x$ will be clear from the context.
Figure~\ref{Fig1} depicts the two representations that are equivalent in this sense.
Finally, we denote the $q$-ary entropy function by $h_q$, that is
\begin{align*}
	& h_q(0)=0 \mbox{, } h_q(1)=\log_q(q-1), \mbox{ and } h_q(x)= -x\log_q(x)\\
	& -(1-x)\log_q(1-x) + x\log_q(q-1) \mbox{ for } 0 < x < 1.
\end{align*}
\subsection{Definitions} \label{ssec:definitions}
\subsubsection{Defect and Partially Defect Cells}
A cell is called defect (\emph{stuck at level $s$}), if it can only store the value $s$.
A cell is called partially defect (\emph{partially stuck at level $s$}), if it can only store values which are at least $s$. Note that a cell that is partially defect at level 0 is a non-defect cell which can store any of the $q$ levels and a cell that is partially defect at level $q-1$ is a (fully) defect cell.
\subsubsection{($\Sigma$, $t$)-PSMC} For $\Sigma \subset \F^n_q$ and non-negative integer $t$, a $q$-ary ($\Sigma$, $t$)-\emph{partially-stuck-at-masking code} $\mycode{C}$ of length $n$ and size $M$ is a coding scheme consisting of a message set $\mathcal{M}$ of size $M$,
an encoder $\mathcal{E}$ and a decoder $\mathcal{D}$.

The encoder $\mathcal{E}$ is a mapping from $\mathcal{M}\times\Sigma$ to $\Fq^n$ such that 
\[ \mbox{for each } (\vmes,\vec{s})\in \mathcal{M} \times \Sigma, \quad \mathcal{E}(\vmes,\vec{s}) \geq \vec{s}, \]
For each $(\vmes,\vec{s})\in \mathcal{M}\times \Sigma$ and each $\vec{e}\in\Fq^n$ such that
\[ \wt({\ve{e}})\leq t \mbox{ and } \mathcal{E}(\vmes,\vec{s})+\vec{e}\geq \vec{s}, \] 
it holds that
\[ \mathcal{D}(\mathcal{E}(\vmes,\vec{s}) + \ve{e}) = \vmes . \]
\subsubsection{($u,1,t)$ PSMC}
A $q$-ary $(u,1,t)$ PSMC of length $n$ and cardinality $\mathcal{M}$ is a $q$-ary $(\Sigma,t)$ PSMC of length $n$ and size $\mathcal{M}$ where
\[ \Sigma= \{ \ve{s}\in\ \{0,1\}^n \mid \wt(\ve{s}) \leq u \} . \]
In this special case, the partially stuck-at condition means that the output of the encoder is non-zero at each position of the support $\ve{\phi}$ of $\ve{s}$.
\section{Code Construction for Masking up to $q-1$ Partially-Stuck-at-$1$ Cells
	\label{u_less_than_q}
}
\subsection{Code Construction}
In this section, we present a coding scheme over $\Fq$ that can mask up to $q-1$ partially stuck cells and additionally can correct errors.
We adapt the construction from \cite{wachterzeh2016codes}, which allows to mask up to $q-1$ partially-stuck-at-$1$ ($s_i =1$ for all $i$) cells with only a single redundancy symbol, but cannot correct any substitution errors. 
\begin{const}\label{cons:matrix_construction:u<q}
	Assume that there is an $[n, k, d]_q$ code $\mycode{C}$ with a $k \times n$ generator matrix of the form
	\begin{align*}
	\ve{G} = 
	\begin{bmatrix}
	\ve{G}_1\\
	\ve{G}_0
	\end{bmatrix}
	=\begin{bmatrix} \ve{0}_{k-1 \times 1} & \ve{I}_{k-1} & \ve{P}_{(k-1) \times (n-k)} \\ {1} & {\ve{1}_{k-1}} & {\ve{1}_{n-k}} \end{bmatrix},
	\end{align*}
where $\ve{I}_{k-1}$ is the $(k-1) \times (k-1)$ identity matrix, $\ve{P} \in \mathbb{F}^{(k-1)\times (n-k)}_q$, and $\ve{1}_{\ell}$ is the all-one vector of length $\ell$.
	Encoder and decoder are shown in Algorithm~\ref{tab:a1} and Algorithm~\ref{tab:a2}. 
\end{const} 
\begin{thmmystyle}
	\label{thm1}
	The coding scheme in Construction~\ref{cons:matrix_construction:u<q} is a ($q-1,1,\lfloor\frac{d-1}{2}\rfloor$) PSMC of length $n$ and cardinality
	$q^{k-1}$.
\end{thmmystyle}
\begin{algorithm}
	\caption{Encoding}
	\label{tab:a1}	
	\KwIn{
		\begin{itemize}
			\item Message: 
			$\vmes = (\mes_0, \mes_1,\dots, \mes_{k-2}) \in {\mathbb{F}_q^{k-1}}$
			\item Positions of partially stuck-at-$1$ cells:
			$\ve{\phi}$
	\end{itemize}}
	
	Compute $\ve{w} = \vmes \cdot \ve{G}_1$
	
	Find $v \in \mathbb{F}_q\setminus \{ w_i \mid i \in \ve{\phi}\}$
		
	Compute $\ve{c} = \ve{w} -v \cdot \ve{G}_0 $
	
	\KwOut{Codeword $\ve{c} \in \mathbb{F}^n_q$ }	
\end{algorithm}
\begin{algorithm}
	\label{tab:a2}
	\caption{Decoding}	
	\KwIn{
		\begin{itemize}
			\item Retrieve $\ve{y} = \ve{c} + \ve{e}$ , $\ve{y} \in \mathbb{F}^n_q $
	\end{itemize} }
	
	$\hat{\c} \gets$ decode $\ve{y}$ in $\mathcal{C}$
	
	$\hat{v} \gets$ first entry of $\hat{\c}$
	
	$\hat {\ve{w}} = (\hat{w}_0, \hat{w}_1,\cdots, \hat{w}_{n-1}) \leftarrow (\hat{\ve{c}} - \hat{v} \cdot \ve{G}_0)$
		
	$\hat{\vmes} \leftarrow (\hat{w}_1, \dots, \hat{w}_{k-1})$
		
	\KwOut{
		Message vector $\hat{\vmes} \in \mathbb{F}^{k-1}_q$} 	
\end{algorithm}	
\begin{proof} 
To mask the partially-stuck-at-$1$ positions, the codeword has to fulfill:
	\begin{equation} \label{equ}
		c_i\geq 1 , \mbox{ for all } i \in \ve{\phi}.
	\end{equation}
Since $\mid \ve{\phi}\mid \;< q$, there is at least one value $v \in \mathbb{F}_q$ such that $w_i \not = v$, $\mbox{ for all } i \in \ve{\phi}$. Thus, ${c}_i=(w_i - v)\not = 0$ and \eqref{equ} is satisfied.
	
The decoder (Algorithm~\ref{tab:a2}) gets $\ve{y}$, which is $\c$ corrupted by at most $\lfloor\frac{d-1}{2}\rfloor$ substitution errors. The decoder of $\mathcal{C}$ can correct these errors and obtain $\c$.

Due to the structure of $\G$, the first position of $\c$ equals $-v$. Hence, we can compute $\hat{\ve{w}} = \ve{w}$ (cf.~Algorithm~\ref{tab:a2}) and $\hat\vmes=\vmes$. 
\end{proof}
\begin{cor}\label{col0}
If there is an $[n,k,d]_q$ code containing a word of weight $n$, then there is a $q$-ary $(q-1,1,\lfloor\frac{d-1}{2}\rfloor)$ PSMC of 
length $n$ and size $q^{k-1}$.
\end{cor}
To obtain a cyclic code, similar to \cite[Theorem~2]{heegard1983partitioned}, we can adapt Algorithms~\ref{tab:a1} and \ref{tab:a2} of Construction~\ref{cons:matrix_construction:u<q} to directly operate on the generator polynomial of the code, which may be beneficial in practice. We present this variant in Appendix~\ref{app:variant_u<q_construction_cyclic_codes}.
For instance, any cyclic code whose generator polynomial is a divisor of $g_0(x) =1+x+x^2+ \dots +x^{n-1}$ contains the all-one codeword. For BCH codes, this is the case if the defining set of the code does not contain $0$. This gives an explicit family of codes whose parameters, for a specific choice of cyclotomic cosets, can be bounded by standard bounds on the minimum distance of cyclic codes such as the BCH bound.
\subsection{Comparison to the Conventional Stuck-Cell Scenario}
Theorem~\ref{thm1} combines \cite[Theorem~1]{heegard1983partitioned} and \cite[Theorem~4]{wachterzeh2016codes} to provide a code construction that can mask partially stuck cells and correct errors. The required redundancy is a single symbol for masking plus the redundancy for the code generated by the upper part of $\G$, needed for the error correction. In comparison, \cite[Theorem~1]{heegard1983partitioned} requires at least 
\begin{align*}
\min\{ n-k \, : \, \exists \, [n,k,d]_q \text{ code with } d > u\} \geq u
\end{align*}
redundancy symbols to mask $u$ stuck cells, where the inequality follows directly from the Singleton bound.

In the following, we present Tables~\ref{table2} and~\ref{table3} to compare ternary 
cyclic codes
	of length $n= 8$ for masking partially stuck cells to masking stuck cells \cite{heegard1983partitioned}, both with error correction. 
	
The tables show that masking partially stuck cells requires less redundancy than masking stuck cells, both with and without additional error correction. The reason is that there is only one forbidden value in each partially stuck-at-$1$ cell, while there
are $q-1$ forbidden values in each stuck at cell.
\subsection{ Remarks on Construction~\ref{cons:matrix_construction:u<q}
}
\begin{rem}\label{rem00}
The special case of Theorem~\ref{thm1} with $n<q$ was used in \cite{SOLOMON1974395} for constructing
a $(q-1)$-ary error-correcting code from a $q$-ary Reed-Solomon code, which can be of interest if $q-1$ is not the power of a prime.
\end{rem}
\begin{rem}\label{rem02}
	The code constructions in Theorem~\ref{thm1} and~\ref{thm:cyclic_construction} also work over the ring of integers modulo $q$ ($\mathbb{Z}/q\mathbb{Z}$) in which $q$ is not necessarily a prime power, similar to the construction for $u< q$ in \cite{wachterzeh2016codes}.
\end{rem}
\begin{rem}\label{rem_further_decrease_redundancy}
	According to \cite[Construction~3]{wachterzeh2016codes}, it is possible to further decrease the required redundancy for masking $u$ partially-stuck-at-$1$ cells to $1-\log_q\lfloor\frac{q}{u+1}\rfloor$. 
	We can use the same strategy here.
	Let $z= \lfloor (\frac{q}{u+1}) \rfloor$. We choose disjoint sets $A_1,A_2,..,A_z$ of size $u+1$ in $\mathbb{F}_q$. As additional information, the encoder picks $j \in \{1,2,\dots,z\}$. In Step~2 of Algorithm 1, it selects $v$ from $A_j$. As the decoder acquires $v$, it can obtain $j$ as well.
\end{rem}
\begin{table*}[h]
	\caption{ Ternary Codes for \textbf{Partially} Stuck-at-$1$ Memory Cell for $n=8$
	}
	\label{table2}
	\begin{center}
		\scalebox{1}{
			\begin{tabular}{ | l |l |l| l | l |}
				\hline
				Cardinality &Overall redundancy&$u$ &$t$& Defining set $D_c$ by \eqref{defining_set} 
				\\ \hline
				$3^7$ & 1 & 2 &0 &$\{4\}$
				\\ \hline
				$3^6$ & 2 & 2 &0 &$\{4\}$
				\\ \hline
				$3^5$ & 3 & 2 &0 &$\{5,7\}$
				\\ \hline
				$3^4$ & 4 & 2 &1 &$\{4,5,7\}$
				\\ \hline	
				$3^3$& 5 &2& 1&$\{1,2,3,6\}$
				\\ \hline
				$3^2$ &6 &2 &1 &$\{1,2,3,4,6\}$
				\\ \hline
				$3^2$&6& 2&1&$\{1,3,4,5,7\}$
				\\ \hline	
				3& 7 & 2&1 &$\{1,2,3,5,6,7\}$
				\\ \hline		
			\end{tabular}
		}
	\end{center}
\end{table*}
\begin{table*}[h]
	\caption{Ternary Codes for Stuck-at Memory \cite{heegard1983partitioned} for $n=8$
	}
	\label{table3}
	\begin{center}
		\scalebox{1}{
			\begin{tabular}{ | l | l|l |l| l |
			}
				\hline
			
				Cardinality &Overall Redundancy&$u$ &$t$&The defining set $D_c$ by \eqref{defining_set} 
				\\ \hline
				$3^7$ & 1 & 1&0&$\{0\}$
				\\ \hline
				$3^6$ & 2 & 1&0&$\{0\}$
				\\ \hline
				$3^5$& 3 & 1 &0 &$\{5,7\}$
				\\ \hline
				$3^4$ &4& 1 &1 &$\{0,1,3\}$
				\\ \hline
				$3^3$& 5 & 1& 1&$\{1,2,3,6\}$
				\\ \hline
				$3^2$ & 6 & 1 &2 &$\{0,1,2,3,6\}$
				\\ \hline
				$3^2$ & 6 &2&1 &$\{0,1,3\}$
				\\ \hline
				3&7&2&1 &$\{1,2,3,6\}$
				\\ \hline
			\end{tabular}
		}	
	\end{center}
\end{table*}
\section{Code Constructions for Masking more than $q- 1$ Partially-Stuck-at-$1$ Cells
}\label{Constructions_u_leq_q_leq_n}
The masking technique in the previous section only guarantees successful masking up to a number of $q-1$ partially stuck-at-$1$ cells. In this section, we present techniques to mask more than $q-1$ cells.
	
Depending on the values of the stored information in the partially stuck positions, Construction~\ref{cons:matrix_construction:u<q} may be able to mask more than $q-1$ cells.
In Section~\ref{ssec:probabilistic_masking}, we determine the probability that masking is possible for fixed partially stuck cell positions and randomly chosen information vectors.

Next, we propose two code constructions for simultaneous masking and error correction when $q \leq u < n$.
One is based on the masking-only construction in \cite[Construction~4]{wachterzeh2016codes} and the other is based on \cite[Section~\rom{6}]{wachterzeh2016codes}, which are able to mask $u\geq q$ partially stuck positions, but cannot correct any errors.
We generalize these constructions to be able to cope with errors.
The latter construction may lead to larger code dimensions for a given pair ($u$, $t$), 
in a similar fashion as \cite[Construction~5]{wachterzeh2016codes} improves upon \cite[Construction~4]{wachterzeh2016codes}.
Further, taking $t=0$ it achieves larger codes sizes than \cite[Construction~5]{wachterzeh2016codes} if the all-one word is in the code.
\subsection{Probabilistic Masking}\label{ssec:probabilistic_masking}
We determine the probability that masking is possible for $u\geq q$ partially stuck-at $1$ cells stuck positions with the code constructions in Theorem~\ref{thm1} and Theorem~\ref{thm2} are used.
This \emph{probabilistic masking} approach enables us to use the memory cells with a certain probability even if there are more than $q-1$ partially stuck cells.
\begin{thmmystyle}
	\label{thm3}
	Let $\ve{G}$ be as in Construction~\ref{cons:matrix_construction:u<q},
and let $\ve{\phi}\subset [n]$ have size $u$.
If the columns of $\ve{G}$ indexed by the elements in $\ve{\phi}$ are linearly independent, a uniformly drawn message from
$\Fq^{k-1}$ results in a word $\ve{c}$ with $c_i\neq 0$ for all $i\in \ve{\phi}$ with probability
	\begin{equation}
		\label{eq4}
		\mathrm{P}(q,u) = 1 - 
		{\dfrac{ \sum_{i=0}^{q-1} (-1)^i {q \choose i} (q-i)^{u}}{q^{u}}}.
	\end{equation}
\end{thmmystyle} 
 \begin{proof}
 An appropriate value for $v$ in Step 2 in Algorithm~\ref{tab:a1} cannot be found if and only if $\{w_i\mid i\in \ve{\phi}\}= \Fq$ which is true if and only $f: \ve{\phi}\mapsto \Fq$ defined by $f_w(i)=w_i$ is a surjection. As is well-known (see e.g \cite[Example 10.2]{VanLintWilson}, the number of surjections from a set of size $u$ to a set of size $q$ equals
\begin{equation}\label{eq:surj} \sum_{i=0}^{q-1} (-1)^i {q\choose i} (q-i)^u .
 \end{equation}
 As the columns of $\G$ are independent, the vector $\ve{w}$ restricted to $\ve{\phi}$ is distributed uniformly on $\Fq^u$, and hence a word is not masked with probability equal to the expression from (\ref{eq:surj}) divided by $q^u$.
\end{proof}
The following example illustrates that the probability that masking is successful can be quite large.
\begin{examplex}\label{ex2}
Let $q =3$, $n=8$, $n-k =0$.
	The probability to mask $u=n-1$ partially stuck-at-$1$ memory cells is
$\mathrm{P}(3,7) = 0.17$. 
This ratio is $0.77$ if $u=q$ and clearly it is $1$ if $u<q$. 
\end{examplex} 
\begin{rem}\label{rem2}
	The assumption in Theorem~\ref{thm3} that the columns of $\ve{G}$ indexed by the partially stuck positions are linearly independent is fulfilled for most
codes with high probability if $u\leq k-1$, especially if $u \ll k-1$. For dependent columns, it becomes harder to count the number of 
intermediate codewords $\ve{w}$ that do not cover the entire alphabet since
	$ w_i$ for all $i \in \ve{\phi}$ is not uniformly distributed over $\Fq^u$.
\end{rem}
\subsection{Code Construction over $\F_q$ for Masking Up to $ q+d_0-3$ Partially Stuck Cells}
We recall that \cite[Construction~4]{wachterzeh2016codes} can mask more than $q-1$ partially stuck-at-$1$ cells and it is a generalization of the all-one vector construction 
\cite[Theorem~4]{wachterzeh2016codes}.
	Hence, replacing the $\ve{1}_n$ vector in Theorem~\ref{thm1} by a parity-check matrix as in \cite[Construction~4]{wachterzeh2016codes} allows masking of 
$q$ or more partially stuck-at $1$ cells, 
 and correct $t$ errors.
\begin{const}\label{cons-ext}
	Suppose that there is an $[n, k, d ]_q$ code $\mycode{C}$ with a $k \times n$ generator matrix of the following form:
		\begin{center}			
		\scalebox{1}{	
			$ \ve{G} = \begin{bmatrix} & \ve{G}_{1} & \\ &\ve{H}_0\end{bmatrix}$}
	\end{center}
 where $\H_0 \in \F_q^{l\times n}$ is a parity-check matrix of an $[n,n-l,d_0]$ code $\mycode{C}_0$.
	Encoder and decoder are shown in Algorithm~\ref{a5} and Algorithm~\ref{a6}. 
\end{const}
\begin{thmmystyle}
	\label{thm-ext}
	The coding scheme in Construction~\ref{cons-ext} is a ($d_0+q-3,1,\lfloor\frac{d-1}{2}\rfloor$) PSMC of length $n$ and cardinality $q^{k-l}$.
\end{thmmystyle}
\begin{algorithm}
	\caption{Encoding}
	\label{a5}
	\KwIn{
		\begin{itemize}
			\item Message: $\vmes\in\F_q^{k-l}$
			\item Positions of partially stuck-at-$1$ cells: $\ve{\phi}$
		\end{itemize}
	}	
	Compute $\ve{w} = \vmes \cdot \ve{G}_1$
	
	Find $\ve{z}\in\F_q^l$ as explained in the proof of \cite[Theorem~7]{wachterzeh2016codes}\label{Step2_a5}	
	
	Compute $\ve{c} = \ve{w} + \ve{z} \cdot \ve{H}_0$	
	
	\KwOut{
		Codeword $\ve{c} \in \mathbb{F}^n_q$} 	
\end{algorithm}
\begin{algorithm}
	\label{a6}
	\caption{Decoding}	
	\KwIn{$\ve{y} = \ve{c}+ \ve{e} \in \mathbb{F}^n_q $}
	
	$\hat{\c} \gets$ decode $\y$ in the code $\mathcal{C}$
	
	Determine $\hat{\vmes}\in \Fq^{k-l}$ and $\hat{\ve{z}}\in\Fq^{l}$ such that $\hat{\c}=\hat{\vmes}\ve{G}_1 + \hat{\ve{z}}\ve{H}_0$.
		
	\KwOut{
		Message vector $\hat{\vmes} \in \mathbb{F}^{k-l}_q$ }	
\end{algorithm}	
\begin{proof}
	 Let $\ve{\phi}\subset [n]$ have size $u\leq q+d_0-3$.
	Algorithm~\ref{a5} finds $\ve{z} =\{z_0,z_1,\dots,z_{l-1}\}$ similar to \cite[Theorem~7]{wachterzeh2016codes} instead of only finding $v$ value as demonstrated in Algorithm~\ref{tab:a1}. 
	Then the proof is exactly the same as in \cite[Theorem~8]{wachterzeh2016codes} for the masking part. In short, the authors in the proof of \cite[Theorem~7]{wachterzeh2016codes} subdivides the code length $n$ into $l$ block lengths of sizes at most $q-1$. Hence, as each block contains at most $q-1$ constraints, then in a corresponding block there is at least $z_i \in \F_q $ such that $z_i \cdot \H_{0_{0,i}} \neq -w_i \mbox{ for } i \in [l]$. Then \cite[Theorem~8]{wachterzeh2016codes} reduces $l$ such that an appropriate $\ve{z} \in \F^l_q$ still exists as any $u-q+2 \leq d_0-1$ of $\H_0$ labeled by $u$ are linearly independent. 
The error correction part of the proof follows the proof of Theorem~\ref{thm1}.
\end{proof}
The gain of Theorem~\ref{thm-ext} in the number of partially stuck cells that can be masked comes at the cost of larger redundancy. 
However, the redundancy is still smaller than the redundancy of the construction for masking stuck-at cells and error correction in \cite{heegard1983partitioned}.
In particular, let $\mycode{C}$ be an $[n,k,d\geq 2t+1]$ code containing an $[n,l]_q$ subcode $\mycode{C}_0$ for which $\mycode{C}_0^{\perp}$ has minimum distance $d_0$.
With Theorem~\ref{thm-ext}, we obtain a $(d_0+q-3,1,\lfloor\frac{d-1}{2}\rfloor)$ PSMC of length $n$ and cardinality $q^{k-l}$.
The construction in \cite{heegard1983partitioned} yields a coding scheme with equal cardinality, allowing for masking up to $d_0-1$ fully stuck cells and correcting $\lfloor\frac{d-1}{2}\rfloor$ errors.
Hence, exactly $q-2$ more cells that are partially stuck at levels $1$
than classical stuck cells can be masked.
\begin{examplex}\label{exc4}
We apply Construction~\ref{cons-ext} to masking up to $u = 4$ partially stuck cells over $\F_4$ and $\vmes\in \mathbb{F}_{4}^9$.
Let $\alpha$ be a primitive element in $\F_{16}$ and let $\mycode{C}$ be the $[15,12,3]_4$ code 
with zeros $\alpha^0$ and $\alpha^1$. Let $\mycode{C}_0$ be the $[15,3]$ subcode of $\mycode{C}$ be the BCH code with zeros
$\{ \alpha^i\mid 0\leq i\leq 14\}\setminus\{ \alpha^5,\alpha^6,\alpha^9\}$. As $\mycode{C}_0^{\perp}$ is equivalent to the 
$[15,12,3]_4$ code with zeros $\alpha^5,\alpha^6,\alpha^9$, it has minimum distance $d_0=3$.
Hence, we obtain a $(4,1,1)$ PSMC code of cardinality $4^{9}$. 
\end{examplex}
\subsection{Code Construction over $\F_{2^\pow}$ for Masking Up to $2^{\pow-1}(d_0+1)-1$ Partially Stuck Cells}
We generalize \cite[Section~\rom{6}]{wachterzeh2016codes} to be able to cope with errors. 
Unlike \cite[Section~\rom{6}]{wachterzeh2016codes} that could be over any prime power $q$, the following code construction works over the finite field $\F_{q}$ where $q= 2^{\pow}$
in order to describe 
a $2^{\mu}$-ary partially stuck cells code construction. This is because binary subfiled subcodes that are required in this construction are not linear subspace for codes over any prime power $q$.
We denote by $\beta_0=1,\beta_1,\dots,\beta_{\pow-1}$ a basis of $\F_{2^\pow}$ over $\F_{2}$. That is, any element $a \in \F_{2^\pow}$ can be uniquely represented as 
$a=\sum_{i=0}^{\pow-1} a_i \beta_i$ where $a_i \in \F_2$ for all $i$. 
In particular, $a \in \F_2$ if and only if $a_1=\dots=a_{\pow-1}=0$. This is a crucial property of $\F_{2^\pow}$ that we will use in Construction~\ref{consbinary}.
\begin{const}\label{consbinary}
Let $\pow >1$.
Suppose $\G$ is a $k\times n$ generator matrix of an $[n,k,d]$$_{2^\pow}$ code $\mycode{C}$ of the form
\begin{equation}
	\G =\left[ \begin{matrix} \H_0 \cr \G_1 \cr \ve{x} \end{matrix} \right] 
\end{equation}
where
\begin{enumerate}
	\item $\H_0\in\F_2^{l\times n}$
 is a parity check matrix of an $[n,n-l,d_0]$$_{2}$ code $\mycode{C}_0$,
	\item $\G_1 \in \mathbb{F}_{2^\pow}^{k-l-1\times n}$,
	\item $\ve{x}\in \mathbb{F}_{2^\pow}$ has Hamming weight $n$.
\end{enumerate}
Encoder and decoder are shown in Algorithm~\ref{a13} and Algorithm~\ref{a14}.
\end{const}
\begin{algorithm}[ht]	
	\caption{Encoding ($\vmes; \ve{m^\prime}; \ve{\phi}$)}
	\label{a13}
	\KwIn{
		\begin{itemize}
			\item Message: \\
			$(\vmes',\vmes)\in \mathcal{F}^l \times \F_{2^{\pow}}^{k-l-1} $, where \\ $\mathcal{F}=\{\sum_{i=1}^{\pow-1} x_i\beta_i \mid (x_1,\ldots ,x_{\pow-1})\in \mathbb{F}_2^{\pow-1}\}$.
			\item Positions of partially stuck-at-$1$ cells: $\ve{\phi}$
			\item Notions introduced in Construction~\ref{consbinary}.
		\end{itemize}
	}
	$\ve{w} \gets \vmes' \cdot \H_0 + \vmes\cdot \G_1 + z \cdot \ve{x}$
	where $z\in\mathbb{F}_{2^{\pow}}$ is chosen such that $|\{i\in\ve{\phi}\mid w_i\in\F_2\}| \leq d_0-1$.\label{Step1}
	\\
	Choose $\ve{\gamma}\in\mathbb{F}_2^{l}$ such that $(\ve{\gamma}\H_0)_i = 1-{\ve{w}}_i$ for all $i\in \ve{\phi}$ for which ${\ve{w}}_i\in\mathbb{F}_2$.\label{Step2}

	\KwOut{$\ve{c}= \ve{w}+\ve{\gamma} \cdot \H_0\in \mycode{C}$
	}
\end{algorithm}
\begin{algorithm}[ht]
	\label{a14}
	\caption{Decoding}	
	\KwIn{
		\begin{itemize}
			\item $\ve{y} = \ve{c}+\ve{e} \in \mathbb{F}^n_{2^\pow}$, where $\ve{c}$ is a valid output of Algorithm~\ref{a13} and $\ve{e}$ is an error of Hamming weight at most $t$.
			\item Notions introduced in Construction~\ref{consbinary}.
		\end{itemize}
	}
	$\hat{\c} \gets$ decode $\ve{y}$ in the code $\mathcal{C}$ \\
      Obtain $\ve{a}\in\F_{2^{\pow}}^{l}, \hat{\vmes}\in\F_{2^{\pow}}^{k-l-1}, \hat{z}\in\F_{2^{\pow}}$
      such that $\hat{\c}=\ve{a}\H_0 + \hat{\vmes}\G_1 + \hat{z}\ve{x}$. \\
   Obtain $\hat{\vmes'}\in \mathcal{F}^{k-l-1}$ and $\hat{\ve{\gamma}}\in\F_2^{k-l-1}$ such that $\ve{a}=\hat{\vmes'} + \hat{\ve{\gamma}}$. \\
	\KwOut{($\hat{\vmes}, \hat{\vmes}')$}
\end{algorithm}	
\begin{thmmystyle}\label{thmbinary}
The coding scheme in Construction~\ref{consbinary} is a $2^{\pow}$-ary $(2^{\pow-1}d_0-1,1,\lfloor\frac{d-1}{2}\rfloor)$ PSMC of length $n$ and cardinality $2^{\pow(k-l-1)}2^{l(\pow-1)}$.
\end{thmmystyle}
\begin{proof}	
 Let $\ve{\phi}\subset [n]$ have size $u\leq 2^{\pow-1}d_0-1$.
\\ We first show the existence of $z$ from Step 1.
For each $i\in\ve{\phi}$, we have that 
$\ve{x}_i\neq 0$, so there are exactly two elements $z\in\F_{2^{\pow}}$ such that $(\vmes' \cdot \H_0 + \vmes\cdot \G_1)_i + z\ve{x}_i\in\mathbb{F}_2$.
As a result,
	\begin{align}
\nonumber 2u= 2|\ve{\phi}| = &\mid \{(i,z)\in \ve{\phi}\times 
\mathbb{F}_{2^{\pow}} \mid (\vmes' \cdot \H_0 + \vmes\cdot \G_1)_i \\ \nonumber &+ z\ve{x}_i \in \mathbb{F}_2\} \mid.
\end{align} 
As $u < 2^{\pow-1}d_0$, there is a $z\in\mathbb{F}_{2^{\pow}}$ such that the condition in Step 1 is satisfied.
	
As $\H_0$ is the parity check matrix of a code with minimum distance $d_0$, any $d_0-1$ columns of $\H_0$ are independent, so an appropriate 
$\ve{\gamma}$ 
exists. Now we show that $c_i\neq 0$ for all $i\in \vec{\phi}$. 
	 Indeed, if $\ve{w}_i\not\in\mathbb{F}_2$, then $\ve{c}_i= \ve{w}_i + ({\ve{\gamma}}\H_0)_i \in\{\ve{w}_i,\ve{w}_i+1\}$, 
	so $\ve{c}_i\not\in\mathbb{F}_2$.
	By Step~\ref{Step2} in Algorithm~\ref{a13}, for $\ve{w}_i\in \mathbb{F}_2$, we have that $\ve{c}_i=1$. 	
	Hence, for all $i\in\vec{\phi}$, $\ve{c}_i$ is either $1$ or is in $\F_{2^\pow}\setminus \F_2$, i.e., $c_i\neq 0$.
	
	\textbf{Decoding:}
	As $\ve{c}\in\mycode{C}$, $\hat{\ve{c}}=\ve{c}$.
	As $\G$ has full rank, and 
\[ \ve{c}= (\vmes'+\ve{\gamma}) \H_0 + \vmes \G_1 + z\ve{x}, \]
it holds that $\ve{a}=\hat{\vmes'} + \hat{\ve{\gamma}}$, $\hat{\vmes}=\vmes$ and $\hat{z}=z$.
As $\hat{\vmes'}\in \mathcal{F}^{l}$ and $\hat{\ve{\gamma}}\in\mathbb{F}_2^{l}$, we can retrieve $\hat{\vmes'}=\vmes'$ 
from $\ve{a} =\hat{\vmes'} + \hat{\ve{\gamma}}$. 
\end{proof}
We show next two minor extensions on Theorem~\ref{thmbinary} for the special case that $\ve{x}$ is the all-one vector.
	\begin{prop}\label{proposition_binary}
		If $\ve{x}$ is the all-one vector in Theorem~\ref{thmbinary}, then the coding scheme in Construction~\ref{consbinary} can be modified to produce a $2^{\pow}$-ary $(2^{\pow-1}d_0-1,1,\lfloor\frac{d-1}{2}\rfloor)$ PSMC of length $n$ and cardinality $2\times 2^{\pow(k-l-1)}2^{l(\pow-1)}$.
	\end{prop}
\begin{proof}
	For $\ve{x}=\ve{1}$,
	if $\vmes' \H_0+\vmes\G_1+z\ve{1}$ has at most $d_0-1$ binary entries, then so has $\vmes' \H_0+\vmes\G_1
	+(z+1)\ve{1}$.
	Hence, there is a $z_0\in\mathcal{F}$ such that $\ve{w}+z_0\ve{1}$ has at most $d_0-1$ binary entries, and we
	can encode
	\[ \ve{w}=\vmes' \H_0 + \vmes\G_1 + (z_0+\zeta)\ve{1}, \] 
	where $\zeta\in\{0,1\}$ is an additional message bit so that the cardinality from Theorem~\ref{thmbinary} is doubled. 
As $z_0\in\mathcal{F}$ and $\zeta\in\{0,1\}$, the pair $(z_0,\zeta)$ can be retrieved from $z_0+\zeta$.
\end{proof}
\begin{customconst}{3.A}[Extension of Construction~\ref{consbinary}]\label{Extension_Const3_more_u}
Let $\G$ be a $k\times n$ generator matrix of an $[n,k,d]$$_{2^\pow}$ code $\mycode{C}$ of the form
\begin{equation*}
	\G =\left[ \begin{matrix} \H_0 \cr \G_1 \cr \ve{1} \end{matrix} \right] \mbox{ , where }
\end{equation*}
1) $\ve{1}$ is the all-one vector of length $n$ \\
2) $\G_1\in \F_q^{k-l-1\times n}$ \\
3) $ \begin{bmatrix} \ve{H}_0 \\ \ve{1} \end{bmatrix}$ is the parity-check matrix of an 
$[n,n-l-1,d_e]_2$ code.
\end{customconst}
\begin{customthm}{4.A}\label{Extension_Th4_more_u}
If the conditions of Construction~\ref{Extension_Const3_more_u} hold, then Construction~\ref{consbinary} can be modified to produce a 
	 $2^{\pow}$-ary $(2^{\mu-1}d_e,1,\lfloor\frac{d-1}{2}\rfloor)$ PSMC of length $n$ and cardinality $2^{\pow(k-l-1)}2^{l(\pow-1)}$.
\end{customthm}
\begin{proof}
In Step 1 of Algorithm~\ref{a13}, the encoder determines $z$ such that
$\mid \{i\in \ve{\phi} \mid w_i\in \mathbb{F}_2\}\mid \;\;\leq d_e-1$; the existence of such a $z$
is proved as in the proof of Theorem~\ref{consbinary}.
Next, the encoder determines $\ve{\gamma}\in \{0,1\}^{l}$ and $\gamma_0\in\{0,1\}$ such that
\[ \ve{v}=
( \ve{\gamma},\gamma_0) \cdot \begin{bmatrix} \ve{H}_0 \\ \ve{1} \end{bmatrix}
\]
is such that $v_i=1-w_i$ for all $i\in\vec{\phi}$ for which $w_i\in\{0,1\}$.
The encoding output $\ve{c}=\ve{v}+\ve{w}=
(\vmes'+\ve{\gamma})\H_0 + \vmes\G_1 + (z_0+\gamma_0)\ve{1}$ thus is in $\mycode{C}$ and 
has no zeros in the positions of $\vec{\phi}$.
	
In decoding, from $\ve{c}$ both $(\vmes'+\ve{\gamma})$ and $\vmes$ can be retrieved, and so, as $\vmes'\in \mathcal{F}^{l}$ and $\ve{\gamma}\in\{0,1\}^{l}$,
 $\vmes'$ can be retrieved as well.
\end{proof}
Proposition~\ref{proposition_binary} doubles the size of the PSMC
as compared
to Theorem~\ref{thmbinary} (by using $\zeta$ as additional message bit), while masking the same number of partially stuck-at-errors and correcting the same number of substitution errors.
Theorem~\ref{Extension_Th4_more_u}, as compared to Theorem~\ref{thmbinary}, results in a PSMC of the same size and error correction capabilities,
but increases the number of cells that can be masked from $2^{\pow-1}d_0-1$ to $2^{\pow-1}d_e-1$. If $d_0$ is odd, then this increment is at least $2^{\pow-1}$.

Now, we show an example using nested BCH codes, allowing to store more symbols compared to Theorem~\ref{thmbinary} for the same code parameters. 

\begin{examplex}\label{ex5}
Let $\alpha$ be a primitive $15^{th}$ root of unity in $\mathbb{F}_{16}$, and let $\mycode{C}$ be the $[15,12,3]_4$
BCH code with zeros $\alpha^5$, $\alpha^6$ and $\alpha^9$. 
Let the $[15,4]_2$ subcode $\mycode{C}^\perp_0$ of $\mycode{C}$ be defined as
\[ \mycode{C}^\perp_0 = \{(x_0,\ldots ,x_{14})\in \mathbb{F}_2^{15} \mid 
\sum_{i=0}^{14} x_i\alpha^{ij}=0 \] 
\[\mbox{ for } j\in\{0,\ldots ,14\} \setminus \{7,11,13,14\}\} . \]
As $\ve{1}\in \mycode{C}\backslash \mycode{C}^\perp_0$, the code $\mycode{C}$ has a generator matrix
of the form given in Construction~\ref{consbinary}, namely
	$$\G^\prime =
\begin{bmatrix} \textcolor{red}{\ve{H}_0} \\\ve{G}_1\\ \textcolor{blue}{\ve{x}} \end{bmatrix},$$
 where 
 $\H_0$ is a generator matrix for $\mycode{C}^\perp_0$ and $\G_1$ has $12-4-1=7$ rows. 
The code $\mycode{C}_0=(\mycode{C}^{\perp}_0)^{\perp}$ is equivalent to the $[15,11]_2$ BCH code with
zeros $\alpha^7,\alpha^{11},\alpha^{13}$ and $\alpha^{14}$. As this BCH code has two consecutive zeros,
its minimum distance (and hence the minimum distance of $\mycode{C}_0$) is at least 3. 

We stipulate that $\alpha^4=\alpha+1$ to obtain explicit $\ve{G}^\prime$ as below,

\scalebox{0.75}{$
	\G^\prime=\left(\begin{array}{rrrrrrrrrrrrrrr}
		\textcolor{red}1 & \textcolor{red}0 & \textcolor{red}0 & \textcolor{red}1 & \textcolor{red}1 & \textcolor{red}0 & \textcolor{red}1 & \textcolor{red}0 & \textcolor{red}1 & \textcolor{red}1 & \textcolor{red}1 & \textcolor{red}1 & \textcolor{red}0 & \textcolor{red}0 & \textcolor{red}0 \\
		\textcolor{red}0 & \textcolor{red}1 & \textcolor{red}0 & \textcolor{red}0 & \textcolor{red}1 & \textcolor{red}1 & \textcolor{red}0 & \textcolor{red}1 & \textcolor{red}0 & \textcolor{red}1 & \textcolor{red}1 & \textcolor{red}1 & \textcolor{red}1 & \textcolor{red}0 & \textcolor{red}0 \\
		\textcolor{red}0 & \textcolor{red}0 & \textcolor{red}1 & \textcolor{red}0 & \textcolor{red}0 & \textcolor{red}1 & \textcolor{red}1 & \textcolor{red}0 & \textcolor{red}1 & \textcolor{red}0 & \textcolor{red}1 & \textcolor{red}1 & \textcolor{red}1 & \textcolor{red}1 & \textcolor{red}0 \\
		\textcolor{red}0 & \textcolor{red}0 & \textcolor{red}0 &\textcolor{red}1 & \textcolor{red}0 & \textcolor{red}0 & \textcolor{red}1 & \textcolor{red}1 & \textcolor{red}0 & \textcolor{red}1 & \textcolor{red}0 & \textcolor{red}1 & \textcolor{red}1 & \textcolor{red}1 & \textcolor{red}1 \\
      	 \omega & \omega & 0 & 1 & 0 & 0 & 0 & 0 & 0 & 0 & 0 & 0 & 0 & 0 & 0 \\
      	0 & \omega & \omega & 0 & 1 & 0 & 0 & 0 & 0 & 0 & 0 & 0 & 0 & 0 & 0 \\ 
      	0 & 0 & \omega & \omega & 0 & 1 & 0 & 0 & 0 & 0 & 0 & 0 & 0 & 0 & 0 \\
      	0 & 0 & 0 & \omega & \omega & 0 & 1 & 0 & 0 & 0 & 0 & 0 & 0 & 0 & 0 \\
      	0 & 0 & 0 & 0 & \omega & \omega & 0 & 1 & 0 & 0 & 0 & 0 & 0 & 0 & 0 \\
      	0 & 0 & 0 & 0 & 0 & \omega & \omega & 0 & 1 & 0 & 0 & 0 & 0 & 0 & 0 \\
      	0 & 0 & 0 & 0 & 0 & 0 & \omega & \omega & 0 & 1 & 0 & 0 & 0 & 0 & 0 \\
		\textcolor{blue}1 & \textcolor{blue}1 & \textcolor{blue}1 & \textcolor{blue}1 & \textcolor{blue}1 & \textcolor{blue}1 & \textcolor{blue}1 & \textcolor{blue}1 & \textcolor{blue}1 & \textcolor{blue}1 & \textcolor{blue}1 & \textcolor{blue}1 & \textcolor{blue}1 & \textcolor{blue}1 & \textcolor{blue}1
	\end{array}\right) 	
	$}, \\
where $\F_4$ has elements $\{0,1,\omega,\omega^2\}$ with $\omega=\alpha^5$.
Note that the top row of $\H_0$ corresponds to the generator polynomial for $\mycode{C}_0^{\perp}$, and the top row of
$\G_1$ corresponds to the coefficients of $(x+\alpha^5)(x+\alpha^6)(x+\alpha^9)$ which is the
generator polynomial of $\mycode{C}$. 
Application of 
	Proposition~\ref{proposition_binary} 
yields a $(5,1,1)$ PSMC over $\mathbb{F}_{2^2}$ of length $15$ and size $2 \times 4^{7}2^{4}= 2^{19}$, whereas application of Construction~\ref{Extension_Const3_more_u} gives a $(7,1,1)$ PSMC over $\mathbb{F}_{2^2}$ with cardinality $2^{2(7+4)-4}=2^{18}$.
Note that application of Construction~\ref{consbinary} yields a $(5,1,1)$ PSMC over $\mathbb{F}_{2^\pow}^2$ of length $15$ and size $4^72^4=2^{18}$.

Finally, we note that application of Theorem~\ref{thm-ext} to $\mycode{C}$ yields a $(4,1,1)$ PSMCs of size $4^{8}$, which has worse parameters than the three PSMC mentioned before.
\qquad \qquad\QEDA
\end{examplex}
Example~\ref{ex5} clearly shows that for the same code parameters, Construction~\ref{consbinary}, Proposition~\ref{proposition_binary} and Construction~\ref{Extension_Const3_more_u} significantly improve upon Construction~\ref{cons-ext}.
\begin{rem}\label{rem3}
	For masking only, choose $n-k =0$ in Construction~\ref{consbinary} and therefore,
	\begin{align*}
		\ve{G}_1= \begin{bmatrix}\ve{0}_{(n-l-1) \times (l+1)} & \ve{I}_{(n-l-1)}
			&\ve{0}_{(n-l-1)\times1} \end{bmatrix},
	\end{align*}
	and we can store $n-l-1$ information symbols.
	Thus, Proposition~\ref{proposition_binary} for masking only improves upon \cite[Construction~5]{wachterzeh2016codes}.
	For example if $l=4$, then $n-l-1= 10$ in \cite[Example~7]{wachterzeh2016codes} and the size of the
	code is $2^{2(n-l-1)}\cdot 2^l= 2^{24} $, while $n-l-1= 10$ in Proposition~\ref{proposition_binary} 
	for $\ve{x}=\ve{1}$ and the cardinality is $2\cdot 2^{2n-l-1} = 2^{25}$. 
\end{rem}
We summarize in Table~\ref{table1} our constructions and compare them with some of the previous works, namely with the construction for masking classical stuck cells in~\cite{heegard1983partitioned} and constructions for partially stuck cells without errors in~\cite{wachterzeh2016codes}. 
	\begin{table*}[h]
	\caption{Comparison between \cite {heegard1983partitioned}, \cite {wachterzeh2016codes}, and this work. We denote by $d$ the minimum distance required to correct errors and $d_0$ to mask (partially) stuck cells. A positive integer $\pow > 1$ is defined in Construction~\ref{consbinary}. Other Notation: See Section~\ref{ssec:notation}.
	}
	\begin{center}
			\def\arraystretch{1.5}
			\begin{tabular}{ |p{3.5cm}| p{3cm} |p{2.5 cm} | p{1.8 cm} | p{3 cm} |p{2.3 cm} |}
				\hline
				&\textbf{(Partially) Stuck Cells $u$} & \textbf{ Distance $d_0$} & \textbf{Errors $\lfloor \frac{d-1}{2} \rfloor$} & \textbf{Redundancy}& \textbf{Cardinality} \\ \hline
				Construction~\ref{cons:matrix_construction:u<q}& $\leq q-1$ & irrelevant & Yes & $n-k+1$ & $q^{k-1}$ \\ \hline
				
				Construction~\ref{cons-ext}& 
				$\leq n$ & $\geq u-q+3$ & Yes
				& 
				$n-k+l$ & $q^{k-l}$
				\\ \hline
				Construction~\ref{consbinary}& $\leq n$ & $\geq \lfloor\frac{2u}{2^\pow}\rfloor+1$ & Yes & 
				$n-k+1+\frac{l}{\pow}$ & $q^{k-1-\frac{l}{\pow}}$
				\\ \hline
				Proposition~\ref{proposition_binary}& $\leq n$ & $ \geq \lfloor\frac{2u}{2^\pow}\rfloor+1$ & Yes & 
				$n-k+1+\frac{l-1}{\pow}$& $q^{k-1+\frac{1-l}{\pow}}$
				\\ \hline
				Construction~\ref{Extension_Const3_more_u}& 
				$\leq n$ & $ \geq \lfloor\frac{2u}{2^\pow}\rfloor$ if $d_0$ is odd
				& Yes
				& 
				 $n-k+1+\frac{l}{\pow}$& $2^{2(k-1-\frac{l}{\pow})}$
				\\ \hline
				\cite[Construction~2]{wachterzeh2016codes}
				 & $\leq q-1$& irrelevant & No 
				& $1$ (since $n-k =0$) & $q^{n-1}$
				\\ \hline
				 \cite[Construction~4]{wachterzeh2016codes}& $\leq n$ & $\geq u-q+3$& 
				 No	
				& $l$ (since $n-k =0$) & $q^{n-l}$
				\\ \hline	
				 \cite[Construction~5]{wachterzeh2016codes}
				 & $\leq n$ & $\geq \lfloor\frac{2u}{q}\rfloor+1$ &
				 No
				& $1+l(1-\log_{q} \lfloor\frac{q}{2}\rfloor)$ (since $n-k =0$), and
			$1+\frac{l}{\pow}$ (for $q=2^\pow$) & $q^{n-1-l(1-\log_{q} \lfloor\frac{q}{2}\rfloor)}$, and $2^{2(n-1-\frac{l}{\pow})}$ (for $q=2^\pow$)
								\\ \hline	
				Proposition~\ref{proposition_binary} (masking only)& $\leq n$ & $\geq \lfloor\frac{2u}{2^\pow}\rfloor+1$ & 
				No
				& 
				$1+\frac{l-1}{\pow}$ (since $n-k =0$) & $2^{2(n-1-\frac{l-1}{\pow})}$
				\\ \hline
				\cite[Theorem~1]{heegard1983partitioned}&$\leq n$ & $\geq u+1$& 
				Yes
				& $n-k+l$ & $q^{k-l}$
				\\ \hline	
			\end{tabular}\label{table1}
	\end{center}
\end{table*}
\section{Generalization of the Constructions to Arbitrary Partially Defective Levels}\label{arbitrary_s_levels}
So far, we have considered the \emph{important} case for $s_i =1 \mbox{ for all } i \in \ve{\phi}$.
In this section, we present error correction and masking codes constructions that can mask partially stuck cells at any level $\ve{s}_i$ and correct errors additionally.
\subsection{Generalization of Theorem~\ref{thm1}}
Here, we give only the main theorem without adding the exact encoding and decoding processes because it follows directly as a consequence of Construction~\ref{cons:matrix_construction:u<q}.
\begin{thmmystyle}[Generalization of Theorem~\ref{thm1}]\label{arbitrary_s_for_u_leq_q}
	Let $\Sigma= \{ \ve{s}\in \Fq^n \mid \sum_{i=0}^{n-1} s_i\leq q-1\}$. 
	Assume there is an $[n, k, d]_q$ code $\mycode{C}$ of a generator matrix as specified in Theorem~\ref{thm1}.
Then there exists a $(\Sigma,\lfloor\frac{d-1}{2}\rfloor)$ PSMC over $\mathbb{F}_q$ of length $n$ and cardinality $q^{k-1}$.
\end{thmmystyle}
\begin{proof}
	We follow the generalization for the masking partially stuck at any arbitrary levels in \cite[Theorem~10]{wachterzeh2016codes}. Hence, for $\ve{s} \in \Sigma$, we modify Step~2 in Algorithm~\ref{tab:a1} such that
$w_i-v\geq s_i$ for all $i\in [n]$. 
Such a $v$ exists as each cell partially stuck at level $s_i$ excludes $s_i$ values for $v$, and $\sum_{i=0}^{n-1} s_i < q$.
The rest of the encoding steps and the decoding process are analogous to Algorithms~\ref{tab:a1} and \ref{tab:a2}. As the output from the encoding process is a codeword, we can correct $\lfloor\frac{d-1}{2}\rfloor$ errors.
\end{proof}
\subsection{Generalization of Construction~\ref{cons-ext}}
In the following, we generalize Construction~\ref{cons-ext} to arbitrary $\ve{s}$ stuck levels.
\begin{prop}\label{modified_version_of_Theorem3}
Let	
\begin{align*}
	\Sigma = & \left\{ \ve{s}\in \F_q^n \; \Big| \; \min \bigg\{ \sum_{i\in \Psi} s_i \; \big| \; \Psi\subseteq[n] , \mid \! \Psi \! \mid =n-d_0+2\bigg\} \right. \\
			& \quad \leq q-1 \Bigg\}
\end{align*}
	then the coding scheme in Construction~\ref{cons-ext} can be modified to produce a
	$(\Sigma,t)$ PSMC of length $n$ and size $q^{k-l}$.
\end{prop}
\begin{proof}
 To avoid cumbersome notation, we assume without loss of generality that 
	$s_0\geq s_1 \geq \cdots \geq s_{n-1}$.
	As the $d_0-2$ leftmost columns of $\H_0$ are independent, there
	is an invertible $\ve{T}\in\F_q^{l\times l}$ such that the matrix $\Y=\ve{T}\H_0$ has the form
	\[ \Y = \begin{bmatrix} I_{d_0-2} & \A \\ \0 & \B 
	\end{bmatrix}, \]
	where $I_{d_0-2}$ is the identity matrix of size $d_0-2$, $\ve{0}$ denotes the $(l-d_0+2)\times (d_0-2)$ all-zero matrix,
	$\ve{A}\in \F_q^{(d_0-2)\times (n-d_0+2)}$ and $\ve{B}\in\F_q^{(l-d_0+2)\times (n-d_0+2)}$.
	As $\ve{T}$ is invertible, and any $d_0-1$ columns of $\H_0$ are independent, any $d_0-1$ columns of $\Y$ are 
	independent as well.
	
	For $0\leq i\leq l-1$, we define 
	\begin{equation}\label{partial_stuck_portion_of_phi}
		L_i = \{ j \in [n] \mid Y_{i,j}\neq 0 \mbox{ and } Y_{m,j}=0 \mbox{ for } m>i\} .
	\end{equation}
	Clearly, $L_0,\ldots , L_{l-1}$ are pairwise disjoint. 
	Moreover, for each $j\in \{d_0-2,d_0-1,\ldots ,n-1\}$, 
	 column $j$ of $\Y$ is independent from the $(d_0-2)$ 
	leftmost columns
	of $\Y$, and so there is an $i\geq d_0-2$ 
	such that $Y_{i,j}\neq 0$. Consequently,
	\begin{equation}\label{eq:Liexhausts}
		\bigcup_{i=d_0-2}^{l-1} L_i = \{d_0-2,\ldots , n-1\}.
	\end{equation}
	By combining (\ref{eq:Liexhausts}) and the form of $\Y$, we infer that
   \begin{equation}\label{eq:small_k}
       L_k=\{k\} \mbox{ for all } k\in [d_0-2] .
\end{equation}
Let $\ve{w}\in\Fq^n$ be the vector to be masked, i.e. the vector after Step~1 in Algorithm~\ref{a5}. 
The encoder successively determines the coefficients $\zeta_0,\ldots ,\zeta_{l-1}$ of $\ve{\zeta}\in\F_q^l$ such that
$\ve{w}+\ve{\zeta}\Y \geq \ve{s}$, as follows. \\
	For $j\in [d_0-2]$, the encoder sets $\zeta_j=s_j-w_j$. \\
	Now let $d_0-2\leq i \leq l-1$ and assume that $\zeta_0, \ldots, \zeta_{i-1}$ have been obtained such that
	\begin{equation}\label{eq:inductionstep}
		w_j + \sum_{k=0}^{i-1} \zeta_k Y_{k,j} \geq s_j \mbox{ for all } 
		j\in \bigcup_{k=d_0-2}^{i-1} L_k.
	\end{equation}
	It follows from combination of (\ref{eq:small_k}) and the choice of $\zeta_0,\ldots \zeta_{d_0-3}$ that (\ref{eq:inductionstep}) is satisfied for $i=d_0-2$.

	For each $j\in L_i$, we define $F_j$ as
	\[ F_j = \bigg\{ x\in \F_q \; \Big| \; w_j + \sum_{k=0}^{i-1} \zeta_k Y_{k,j} + xY_{i,j} < s_j \bigg\} .\]
	Clearly, $|F_j|=s_j$ as $Y_{i,j}\neq 0$, and so
	\[ \Big|\bigcup_{j\in L_i} F_j\Big| \leq \sum_{j\in L_i} \big|F_j\big| = \sum_{j\in L_i} s_j \leq 
	\sum_{j=d_0-2}^{n-1} s_j \leq q-1, \]
	where the last inequality follows from the assumption of $\Sigma$ in the proposition statement
	and the ordering of the components of $\ve{s}$.
	Hence, $\bigcup_{j\in L_i} F_j \neq \F_q$.
	The encoder chooses $\zeta_i\in \F_q\setminus \bigcup_{j\in L_i} F_j$.
	We claim that
	\begin{equation}\label{eq:Fj}
		w_j + \sum_{k=0}^i \zeta_k Y_{k,j} \geq s_j \mbox{ for all } j\in \bigcup_{k=0}^{i} L_k.
	\end{equation}
	For $j\in L_i$, (\ref{eq:Fj}) follows from the definition of $F_j$.
	For $j\in \bigcup_{k=0}^{i-1} L_k$, (\ref{eq:Fj}) follows from (\ref{eq:inductionstep}) and the fact that $Y_{i,j}=0$.
	
	By using induction on $i$, we infer that
	\begin{equation}\label{eq:all_cell_masked}
	w_j + \sum_{k=0}^{l-1} \zeta_k Y_{k,j} \geq s_j \mbox{ for all } j\in \cup_{k=0}^{l-1}L_k = [n].
\end{equation}
	That is, with
	$\ve{\zeta}=(\zeta_0,\ldots ,\zeta_{l-1})$, we have that $\ve{w}+\ve{\zeta}\Y \geq \ve{s}$.
	As $\Y=\ve{T}\H_0$, it follows that $\ve{z}:=\ve{\zeta}\ve{T}$ is such that
	\[ \ve{w} + \ve{z}\H_0 \geq \ve{s}. \]	
The decoding process remains as in Algorithm~\ref{a6}.
\end{proof}
We give an alternative \emph{non-constructive} proof for Proposition~\ref{modified_version_of_Theorem3} in Appendix~\ref{alternative_Proof_of_Proposition_2}. 
 \begin{rem}\label{remark_two_u_groups}
The proof of Proposition~\ref{modified_version_of_Theorem3} shows that $(d_0-2)$ cells can be set to any desired value,
while the remaining $(n-d_0+2)$ cells can be made to satisfy the partial stuck-at conditions, provided that the sum of the 
stuck-at levels in these ($n-d_0+2$) cells is less than $q$.
\end{rem}
\begin{cor} [Generalization of Theorem~\ref{thm-ext}]
	\label{arbitrary_s_for_u_geq_q}
   	Let $s\in\F_q$ and let
$\Sigma=\{ \ve{s}\in\F_q^n\mid \wt(\ve{s})\leq d_0+\lceil \frac{q}{s}\rceil - 3 \mbox{ and}
  \max \{s_i\mid i\in [n]\leq s\}.$ 
	The coding scheme in Construction~\ref{cons-ext} is a $(\Sigma,t)$ PSMC
	scheme of length $n$ and size $q^{k-l}$.
\end{cor}
\begin{proof}
     Let $\ve{s}\in\Sigma$ have weight $u\leq d_0+q-3$. Let $\Psi\subseteq [n]$ of size
$n-d_0+2$ be such that the number of non-zero components of $\ve{s}$ in $[n]\setminus \Psi$ equals $\min(d_0-2,u)$.
Then $\ve{s}$ has $u-\min(d_0-2,u)$ non-zero components in $\Psi$. As a consequence, if $u\leq d_0-2$, then
$\sum_{i\in\Psi} s_i=0$, and if $u>d_0-2$, then
	\[ \sum_{i\in \Psi} s_i \leq  s(u-d_0+2) \leq s(\lceil \frac{q}{s}\rceil -1) < s(\frac{q}{s}+1 - 1) = q. \]
	Hence in both cases, $\sum_{i\in \Psi} s_i \leq q-1$. The corollary thus follows from Proposition~\ref{modified_version_of_Theorem3}.
In particular, if $s=1$, the corollary agrees with Theorem~\ref{thm-ext}.
\end{proof}
We do not generalize Construction~\ref{consbinary} as it is tailored to the \emph{special} case where $s_i =1 \mbox{ for all } i \in \ve{\phi}$. 
\section{Trading Partially Stuck Cells with Errors \label{section_trading}}
In the constructions shown so far, the encoder output $\ve{c}$ is a word from an error correcting code $\mycode{C}$. 
If $\ve{c}$ does not satisfy the partial stuck-at conditions in $j$ positions, the encoder could modify it in these $j$ positions
to obtain a word $\ve{c}'=\ve{c}+\ve{e}'$ satisfying the partial-stuck at constrains, while $\wt(\ve{e}')=j$. 
If $\mycode{C}$ can correct $t$ errors, then it still is possible to correct $t-j$ errors in $\ve{c}'$. This observation was also made in
 \cite[Theorem~1]{heegard1983partitioned}.
The above reasoning shows that the following proposition holds.
\begin{prop}\label{proposition3}
If there is an $(n,M)_q (u,1,t)$ PSMC, then for any $j$ with $0\leq j\leq t$, 
there is an $(n,M)_q (u+j,1,t-j)$ PSMC.
\end{prop}
In the remainder of this section, we generalize the above proposition to general $\Sigma$ (Theorem~\ref{tradingu_t}).
We also provide variations on the idea of the encoder introducing errors to the result of a first encoding step in order that
the final encoder output satisfies the partially stuck-at conditions.	
\begin{thmmystyle} [Partial Masking PSMC]\label{tradingu_t} Let $\Sigma\subset \F_q^n$, and assume that there exists an $(n,M)_q (\Sigma,t)$ PSMC $\mycode{C}$.
	For any $j\in[t]$, there exists an $(n,M)_q (\Sigma^{(j)},t-j)$ PSMC $\mycode{C}_j$,
	where
	\[ \Sigma^{(j)}= \{\ve{s'}\in \F_q^n\mid \exists{\ve{s}\in\Sigma} \left[ d(\ve{s},\ve{s'})\leq j \mbox{ and } \ve{s'}\geq \ve{s} \right] \} . \]
\end{thmmystyle}	
\begin{algorithm}	
	\caption{Encoding}
	\label{alg_Encoder01}
	\KwIn{ $(\vmes,\vec{s'})\in \mathcal{M}\times \Sigma^{(j)}$.
	}	
	Determine $\ve{s}\in\Sigma$ such that $d(\ve{s},\ve{s'})\leq j$ and $\ve{s}'\geq \ve{s}$. 
	
	Let $\ve{c}=$ $\mathcal{E}$($\vmes,\ve{s})$.	
	
	Define $\ve{c'}=\mathcal{E}_j'$($\vmes,\ve{s'})$ as $c'_i=\max(c_i,s_i')$ for $i \in [n]$.
	
	\KwOut{
		Codeword $\ve{c'}$.
} 		
\end{algorithm}
\begin{algorithm}
	\label{alg_decoder02}
	\caption{Decoding}	
	\KwIn{Recived $\ve{y}=\ve{c'}+\ve{e}$ where $\wt(\ve{e})\leq t-j$ and $\ve{y}\geq \ve{s'}$}	
	
	Message $\vmes =\mathcal{D}(\ve{y})$

	\KwOut{
		Message vector $\vmes$ }
\end{algorithm}
\begin{proof}
	Let the encoder $\mathcal{E}_j$ and the decoder $\mathcal{D}_j$ for $\mycode{C}_j$ be Algorithm~\ref{alg_Encoder01} and Algorithm~\ref{alg_decoder02}, respectively.
	By definition, $\ve{c'}\geq \ve{s'}$. Moreover, if $s_i=s'_i$, then $c_i\geq s_i=s'_i$, so $c_i=c'_i$.
	As a result, $d(\ve{c},\ve{c'})\leq j$.
	
	In Algorithm~\ref{alg_decoder02}, the decoder $\mathcal{D}$ of $\mycode{C}$ is directly used for decoding $\mycode{C}_j$. 
	As $\ve{y}\geq \ve{s'}$, surely $\ve{y}\geq \ve{s}$. Moreover, we can write $\ve{y}=\ve{c} + (\ve{c'}-\ve{c}+\ve{e})$.
	As shown above, wt$(\ve{c'}-\ve{c})\leq j$, and so 
	$\wt(\ve{c}-\ve{c}'+\ve{e})\leq t.$ As a consequence, $\mathcal{D}(\ve{y})=\vmes$.
\end{proof}
We can improve on Theorem~\ref{tradingu_t} for Construction~\ref{consbinary} giving Lemma~\ref{lem03}.
	\begin{lem}\label{lem03}
				Given an $[n,k,d]_q$ code as defined in Construction~\ref{consbinary}, then 
	for any $j$ such that $0\leq j\leq \lfloor\frac{d-1}{2}\rfloor$, there is a $2^{\mu}$-ary $(2^{\mu-1}(d_0+j)-1,1,\lfloor\frac{d-1}{2}\rfloor-j)$ PSMC of length $n$ and size $q^{k-l-1}$.
		\end{lem}
	\begin{proof} Let $\ve{\phi}\subset [n]$ has size $u\leq 2^{\mu-1}(d_0+j)-1$. We use the notation from Algorithm~\ref{a13}. After Step 1, 
$\ve{w}$ has at most $u_0=\lfloor \frac{2u}{2^{\mu}}\rfloor \leq d_0+j-1$ 
binary entries in the positions from $\ve{\phi}$.
 After Step 2, at least $d_0-1$ of these entries in $\ve{c}$ differ from 0. By setting the at most $j$ other binary entries in the positions from
$\ve{\phi}$ equal to 1, the encoder introduces at most $j$ errors, and guarantees that the partial stuck-at conditions are satisfied.
	\end{proof}
In Lemma~\ref{lem01}, we use another approach for introducing errors in order to satisfy the stuck-at conditions.
	\begin{lem}\label{lem01}
		Given an $[n,k,d]_q$ code containing a word of weight $n$, for any $j$ with $0\leq j\leq \lfloor \frac{d-1}{2}\rfloor$, there is a $q$-ary $(q-1+qj,1,\lfloor\frac{d-1}{2}\rfloor-j)$ PSMC of length $n$ and size $q^{k-1}$.
	\end{lem}
	\begin{proof} We use the notation from Construction~\ref{cons:matrix_construction:u<q}.
	
 Let $\ve{\phi}\subset [n]$ have size $u \leq q-1+qj$.
Let $\ve{x}$
be a codeword of weight $n$. For each $i\in \ve{\phi}$, there is exactly one $v\in\F_q$ such that
$w_i+vx_i=0$, and so
\[ \sum_{v\in\F_q} \mid \{ i\in \ve{\phi}\mid w_i+vx_i=0\}\mid\; = u. \]
As a consequence, there is $v\in \F_q$ such that $\ve{c}=\ve{w}+v\ve{x}$ has most $\lfloor\frac{u}{q}\rfloor \leq j$ entries in $\ve{\phi}$ equal to zero.
By setting these entries of $\ve{c}$ to a non-zero value, the encoder introduces at most $j$ errors.
As $\mycode{C}$ can correct up to $\lfloor\frac{d-1}{2}\rfloor$ errors, it can correct these $j$ errors and additionally up to $\lfloor\frac{d-1}{2}\rfloor-j$ substitution errors.
	\end{proof}
		\begin{examplex}\label{example_trading_initial}
	Consider a $[15,9,5]_4$ code $\mycode{C}$ containing the all-one word, e.g. the BCH code with zeroes 
	$\alpha,\alpha^2,\alpha^3$, where $\alpha$ is a primitive element in $\F_{16}$. 
	Let $u\leq 7$ and $t=1$. We use the all-one word for partial masking, ensuring that $0$ occurs in at most 
	$\lfloor \frac{u}{4}\rfloor\leq 1$ position indexed by $\ve{\phi}$. We set the codeword value in this position to $1$, introducing one error. We can correct this introduced error and 
	one additional random error as $\mycode{C}$ has minimum distance $5$.
	Hence, we have obtained a $4$-ary $(7,1)$ PSMC of length $15$ and cardinality $4^8$.
		\qquad \qquad\QEDA
\end{examplex}
We show in Example~\ref{example_binary_trading} how applying Lemma~\ref{lem01} for Construction~\ref{consbinary} outperforms Lemma~\ref{lem03}.
\begin{examplex}\label{example_binary_trading}
	Given $d_0=3$, $u=15$ and $q=2^2$ and let $\alpha$ be a primitive element in $\F_{4}$ and take $\ve{x}= \ve{1}$. Assume we have
		\begin{align*}
		 \ve{w}^{(\ve{\phi})} & =(\vmes' \cdot \H_0 + \vmes\cdot \G_1) + z \cdot \ve{1}\\
		& = (\textcolor{blue}{0},\textcolor{red}{1},\alpha,1+\alpha,\textcolor{blue}{0},\textcolor{red}{1},\alpha,1+\alpha,\textcolor{blue}{0},\textcolor{red}{1},\alpha,1+\alpha,\textcolor{blue}{0},\textcolor{red}{1},\alpha)\\
		&+z \cdot \ve{1},
	\end{align*}
	 then choosing $z=1+\alpha$ minimizes the number of binary values in $\ve{w}^{(\ve{\phi})}$, we get: 
	\[\ve{w}^{(\ve{\phi}) }= (1+\alpha,\alpha,\textcolor{red}{1},\textcolor{blue}{0},1+\alpha,\alpha,\textcolor{red}{1},\textcolor{blue}{0},1+\alpha,\alpha,\textcolor{red}{1},\textcolor{blue}{0},1+\alpha,\alpha,\textcolor{red}{1}).\] 
	Following Step~\ref{Step2} in Algorithm~\ref{a13} and since $d_0 = 3$, we can mask at most $d_0-1$ binary values highlighted in the vector $\ve{w}^{(\ve{\phi})}$ that leaves us, in this example, with at most $\lfloor \frac{2u}{2^2}\rfloor-d_0+1 = 5$ zeros that remain unmasked.
	
	However, applying Lemma~\ref{lem01} instead for Construction~\ref{consbinary} 
	gives a better result. 
	Choosing $\ve{\gamma} = \ve{0}$ in Step~\ref{Step2} of Algorithm~\ref{a13}, we obtain $\ve{c}^{(\ve{\phi})} = \ve{w}^{(\ve{\phi})}$ with $(\lfloor\frac{u}{q}\rfloor =3)$ zeros highlighted in blue above that we can directly trade.	
	\qquad \qquad\QEDA
\end{examplex}
	\begin{rem}\label{rem11}
   As the code from Construction~\ref{consbinary} has a word of weight $n$, Lemma~\ref{lem01} implies the existence of an $(u,1,t)$ PSMC of cardinality $q^{k-1}$ under the condition that
$2(t+\lfloor \frac{u}{q}\rfloor) < d$. Lemma~\ref{lem03} shows the existence of an $(u,1,t)$ PSMC of smaller cardinality, {\em viz.} $q^{k-l}$, under the condition that
$2(t+\max(0,\lfloor\frac{2u}{2^{\mu}}\rfloor-d_0+1) < d$. As a consequence, Lemma~\ref{lem03} can only improve on Lemma~\ref{lem01} if $d_0-1> \lfloor \frac{2u}{2^{\mu}}\rfloor - \lfloor \frac{u}{2^{\mu}}\rfloor$.
\end{rem}
We can generalize Lemma~\ref{lem01} as follows.
	\begin{lem}\label{lem04}
	Given an $[n,k,d]_q$ code containing a word of weight $n$.
Let $0\leq j\leq \lfloor\frac{d-1}{2}\rfloor$, and let 
\[ \Sigma = \bigg\{ \ve{s}\in \F_q^n \; \Big| \; \sum_i s_i \leq q-1+qj\bigg\}. \]
There is a $q$-ary ($\Sigma,\lfloor\frac{d-1}{2}\rfloor-j)$ PSMC of length $n$ and size $q^{k-1}$.
	\end{lem}
\begin{proof}
We use the notation from Theorem~\ref{arbitrary_s_for_u_leq_q}. For simplicity, we assume
that the code contains the all-one word.
	We wish to choose the multiplier $v\in\F_q$ such that $\ve{c}=\ve{w}-v\cdot \ve{1}$ satisfies
	$c_i\geq s_i$ for as many indices $i$ as possible. For each index $i$, there are $q-s_i$ values for $v$ such that this 
	inequality is met. Hence, there is a $v\in\F_q$ such that $c_i\geq s_i$ for at least
	$\lceil \frac{1}{q} \sum_i(q-s_i)\rceil = n-\lfloor \frac{1}{q}\sum_i s_i\rfloor$ indices $i$.
	The encoder thus needs to introduce errors only in the at most $\lfloor \frac{1}{q} \sum_i s_i \rfloor$ positions for which the inequality is not satisfied.
\end{proof}
\begin{lem}\label{only_trading_thm:super-general0}
	Assume there exists a matrix as in Proposition~\ref{modified_version_of_Theorem3}.
	Let $0\leq j\leq \lfloor\frac{d-1}{2}\rfloor$, and let
\begin{align*}
	\Sigma=&\bigg\{ \ve{s}\in\F_q^n \; \Big| \; 
	\exists{\Psi\subset [n]: \; \big| \;\Psi\;\big|=n-d_0+2} \\
		&\Big[ \sum_{i\in \Psi} s_i \leq q-1+qj \Big]\bigg\}.
\end{align*}
	Then exists a $q$-ary $(\Sigma,\lfloor\frac{d-1}{2}\rfloor-j)$ PSMC of length $n$ and size $q^{k-l}$.
\end{lem}
\begin{proof}
	Let $\ve{s}\in\Sigma$.
	 In order to simplify notation, we assume without loss of generality that 
	$\sum_{i=d_0-2}^{n-1} s_i \leq q-1+qj$.
	We use the same argument as in the alternative proof of Proposition~\ref{modified_version_of_Theorem3} (see Appendix~\ref{alternative_Proof_of_Proposition_2}).
Clearly,
	\[ n-d_0+2- \left\lfloor\frac{1}{q} \sum_{i=d_0-2}^{n-1}s_i\right\rfloor
	\geq n-d_0+2-j . \]
	So we infer that for at least $n-j$ indices $i\in[n]$,
	\[ w_i+ \Big( (\ve{\zeta},\ve{\eta}) \ve{T}\H_0\Big)_i \geq s_i. \tag*{\qedhere} \] 	
\end{proof}
\begin{rem}\label{remark9}
The proof of Lemma~\ref{only_trading_thm:super-general0} shows that the encoder output in fact can be made equal to $\ve{s}$ in the
$d_0-2$ largest entries of $\ve{s}$. In fact, it shows that the scheme allows for masking $d_0-2$ stuck-at errors, masking
partial stuck errors in the remaining cells, and correcting $\lfloor\frac{d-1}{2}\rfloor-j$ substitution errors, provided that the sum
of the stuck-at levels in the $n-d_0+2$ remaining cells is less than $(j+1)q$. 
\end{rem}
\section{Upper bounds on PSMC codes }\label{upper_bounds_psmc}
The output of an encoder has restrictions on the values in the partially stuck-at cells; in the other cells, it can attain all values.
So the set of all encoder outputs is a poly-alphabetic code \cite{VSidorenkoPolyalphabetic}. To be more precise, the following proposition holds.
\begin{prop}\label{prop:poly} Let $\mycode{C}$ be an $(n,M)_q (\Sigma,t)$ partially stuck-at-masking code with encoder $\mathcal{E}$.
For any $\ve{s}\in\Sigma$, let
\[ \mycode{C}_{\ve{s}}=\{ \mathcal{E}(\vmes,\ve{s}) \mid \vmes\in \mathcal{M}\}. \]
Then $\mycode{C}_{\ve{s}}$ is a code with minimum distance at least $2t+1$ and $|\mathcal{M}|$ words, and 
\[ \mycode{C}_{\ve{s}} \subset Q_0\times Q_1 \times \cdots \times Q_{n-1}, \mbox{ where} \]
$Q_i = \{ x\in\mathbb{F}_q\mid x\geq s_i\}$.
\end{prop}
\begin{proof}
By our error model, errors in stuck-at cells result in values still satisfying the stuck-at constraints. Therefore, $t$ errors can be corrected 
if and only if $\mycode{C}_{\ve{s}}$ has minimum Hamming distance at least $2t+1$. The rest of the proposition is obvious.
\end{proof}
As a result of Proposition~\ref{prop:poly}, upper bounds on the size of poly-alphabetic codes \cite{VSidorenkoPolyalphabetic} are also upper bounds on
the size of partially-stuck-at codes.
\begin{thmmystyle} (Singleton type bound)
Let $\mycode{C}$ be a $q$-ary $(\Sigma,t)$ PSMC of length $n$ and size $M$.
Then for any $\ve{s}=(s_0,\ldots ,s_{n-1})\in \Sigma$,
\[ M \leq \min\bigg\{ \prod_{j\in J}(q-s_j) \; \Big| \; J\subset[n], |J|=n-2t\bigg\}. \]
\end{thmmystyle}
\begin{proof}
Combination of Proposition~\ref{prop:poly} and \cite[Theorem~2]{VSidorenkoPolyalphabetic}.
\end{proof}
\begin{thmmystyle} (Sphere-packing type bound)
Let $\mycode{C}$ be a $q$-ary $(\Sigma,t)$ PSMC of length $n$ and size $M$.
Then for any $\ve{s}=(s_0,\ldots ,s_{n-1})\in\Sigma$
\[ M \leq \frac{\prod_{i=0}^{n-1}(q-s_i)}{V^{(b)}_t}, \]
where $V^{(b)}_t$, the volume of a ball of radius $t$, satisfies
\[ V^{(b)}_t=\sum_{r=0}^t V^{(s)}_r, \]
where the volume $V^{(s)}_r$ of the sphere with radius $r$ is given by 
\[ V^{(s)}_0 =1 , \]
\[ V^{(s)}_r= \sum_{1\leq i_1 < \ldots < i_r\leq n} (q-1-s_{i_1})\cdots(q-1-s_{i_r}). \]
\end{thmmystyle}
\begin{proof}
Combination of Proposition~\ref{prop:poly} and \cite[Theorem~3]{VSidorenkoPolyalphabetic}.
\end{proof}
\begin{rem}
 The difference between poly-alphabetic codes and partially-stuck-at-masking codes is that in the former,
the positions of stuck-at cells and the corresponding levels are known to both encoder and decoder, whereas in the latter,
this information is known to the encoder only.	
\end{rem}
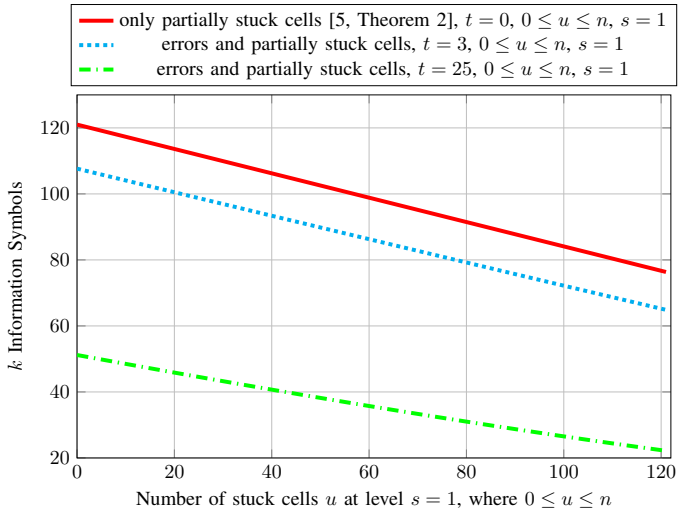
\begin{figure}[h]
		\begin{center}
	\scalebox{0.75}{
		\begin{tikzpicture} 
			\begin{axis}[
				filter discard warning=false,
				height=8cm,
				width=12cm,
				ylabel= {$k$ Information Symbols},
				xmin = 0,
				xmax = 122,
				ymin = 20,
				ymax = 130,
				xlabel= {Number of stuck cells $u$ at level $s=1$, where $0\leq u\leq n$},
				grid=major,
				legend style={at={(0.50,1.25)},anchor=north},
				cycle list name=solidlinesbyhaider7
				]
				\input{figures/SPq3n121t3/k_cases_Only_PSMC}
				\addlegendentry{only partially stuck cells \cite[Theorem~2]{wachterzeh2016codes}, $t=0$, $0\leq u \leq n$, $s=1$};
				\addlegendentry{errors and partially stuck cells, $t=3$, $0\leq u \leq n$, $s=1$};
				\input{figures/SPq3n121t3/k_cases_overlapping}
				\addlegendentry{errors and partially stuck cells, $t=25$, $0\leq u \leq n$, $s=1$};
				\input{figures/SPq3n121t25/k_cases_overlapping}
			\end{axis}
			
	\end{tikzpicture} }
	\caption{\textbf{Sphere-packing bounds}: Comparison for $k$ information symbols for ("only partially stuck cells \cite[Theorem 2]{wachterzeh2016codes}") and our sphere-packing-like ("errors and partially stuck cells") bounds. The classical sphere-packing bound ("only errors") can read at $u=0$ in our sphere-packing-like bounds curves. 
	The chosen parameters are $\pow =5$ and $ q =3$, and $n = ((q^\pow -1)/(q-1))$.}
	\label{fig2}
	\end{center}
\end{figure}
Figure~\ref{fig2} compares our derived sphere-packing-like bound to the amount of storable information symbols for a completely reliable memory (i.e., no stuck cells, no errors that can be seen at $u=0$ in the solid line) and the upper bound on the cardinality of an only-masking PSMC (only stuck cells, no errors) derived in \cite{wachterzeh2016codes} as depicted in the solid curve.
At $u=0$, the derived sphere-packing-type bound (dotted and dashed-dotted plots) matches the usual sphere-packing bound ("only errors") case.
The more $u$ partially stuck at cells, the less amount of storable information, i.e. only $q-1$ levels can be utilized. Hence, the dotted and dashed-dotted lines are declining while $u$ is growing. 
On the other hand, the more errors (e.g., $t=25$ in the dashed-dotted plot), the higher overall required redundancy and the lower storable information for the aforementioned curve.
\section{Gilbert--Varshamov (GV) Bound}\label{lower_like_bounds}
\subsection{ Gilbert--Varshamov (GV) Bound: finite length}
We have provided various constructions of $(u,1,t)$ PSMCs based on $q$-ary $t$-error correcting codes with additional properties.
In this section, we first employ GV-like techniques to show the existence of
$(u,1,t)$ PSMCs. Next, we study the asymptotic of the resulting GV bounds.

We start with a somewhat refined version of the Gilbert bound, that should be well-known, but for which we did not 
find an explicit reference.
\begin{lem}\label{lem:GVrefined}
Let $q$ be a prime power, and assume there is an $[n,s]_q$ code $\mycode{C}_s$ with minimum distance at least $d$.
If $k\geq s$ is such that
\[ \sum_{i=0}^{d-1} {n\choose i} (q-1)^i < q^{n-k+1}, \]
then there is an $[n,k]_q$ code $\mycode{C}_k$ with minimum distance at least $d$ that has $\mycode{C}_s$ as a subcode.
\end{lem}
\begin{proof} 
By induction on $k$.
For $k=s$, the statement is obvious. Now let $\kappa \geq s$ and let $\mycode{C}_{\kappa}$ be an $[n,\kappa]_q$ code with minimum
distance at least $d$ that has $\mycode{C}_s$ as a subcode. 
If $q^{\kappa}\sum_{i=0}^{d-1} {n\choose i}(q-1)^i < q^n$, then the balls with radius $d-1$ centered at the words of $\mycode{C}_{\kappa}$
do not cover $\F_q^n$, so there is a word $\ve{x}$ at distance at least $d$ from all words in $\mycode{C}_{\kappa}$.
As shown in the proof of \cite[Theorem. 5.1.8]{Lint}, the $[n,\kappa+1]_q$ code $\mycode{C}_{\kappa+1}$ spanned by $\mycode{C}_{\kappa}$ and $\ve{x}$ has 
minimum distance at least $d$.
\end{proof}
\subsubsection{Finite GV bound based on Lemma~\ref{lem01}}\label{Finite_GV_Bin}
\begin{thmmystyle}\label{th:upartial}
Let $q$ be a prime power. Let $n,k,t,u$ be non-negative integers such that
\[ \sum_{i=0}^{2(t+\lfloor \frac{u}{q}\rfloor)} {n\choose i}(q-1)^i < q^{n-k+1}. \]
There exists a $q$-ary $(u,1,t)$ PSMC of length $n$ and size $q^{k-1}$.
\end{thmmystyle}
\begin{proof}
	Let $\mycode{C}_1$ be the $[n,1,n]_q$ code generated by the all-one word.
	Lemma~\ref{lem:GVrefined} implies that
	there is an $[n,k]_q$ code with minimum distance at least $2t+1$ that contains the all-one word.
	Lemma~\ref{lem01}
shows that $\mycode{C}_k$ can be used to construct a PSMC with the claimed parameters.
\end{proof}
\begin{rem}\label{special_case_of_th:upartial}
 GV bound from Theorem~\ref{thm1} is a special case of Theorem~\ref{th:upartial} for $u \leq q-1$.
\end{rem}
\subsubsection{Finite GV bound based on Construction~\ref{cons-ext}}\label{qgequ}
\begin{lem}\label{Lem:Averaging}
Let $q$ be a prime power, and let $1\leq k <n$. 
Let $E\subset \F_q^n\setminus\{0\}$. The fraction of $[n,k]_q$ codes with non-empty intersection with $E$ is less than
$|E|/q^{n-k}$.
\end{lem}
\begin{proof}
Let $\mycode{C}$ be the set of all $[n,k]_q$ codes. Obviously,
\[ \Big|\big\{C\in \mycode{C}\mid C\cap E \neq \emptyset \big\}\Big| \;\; \hspace{-1ex}\leq \hspace{-2ex}\sum_{C\in \mycode{C}: C\cap E\neq \emptyset}\big| \; C\cap E \; \big| =
\sum_{C\in \mycode{C}} \big| \; C\cap E\;\big| . \]
It follows from \cite[Lemma 3]{LoeIT} that 
\[ \frac{1}{\mid \mycode{C}\mid} \sum_{C\in \mycode{C}} \big| \; C\cap E \;\big| = \frac{q^k-1}{q^n-1}|E| < \frac{|E|}{q^{n-k}} . \tag*{\qedhere}\]
\end{proof}
\begin{rem}\label{rem:divide-by-(q-1)}
If $E$ has the additional property that $\lambda \ve{e}\in E$ for all $\ve{e}\in E$ and $\lambda\in\F_q\setminus\{0\}$,
then the upper bound in Lemma~\ref{Lem:Averaging} can be reduced to 
$|E|/(q-1)q^{n-k}$. 
\end{rem}
\begin{lem}\label{Th:GV-q}
	Let $k,n,d$ and $d^{\perp}$ be integers such that
	\begin{align*}
	& \sum_{i=0}^{d-1} {n\choose i} (q-1)^i < \frac{1}{2} q^{n-k} \mbox{ and }\\
	& \sum_{i=0}^{d^{\perp}-1} {n \choose i} (q-1)^i < \frac{1}{2} q^{k}.
\end{align*}
	There exists a $q$-ary $[n,k]$ code $\mycode{C}$ with minimum distance at least $d$
	 such that $\mycode{C}^{\perp}$ has minimum distance at least $d^{\perp}$.
\end{lem}
\begin{proof}
Let $\mycode{C}$ denote the set of all $[n,k]_q$ codes. By applying Lemma~\ref{Lem:Averaging} with $E=\{\ve{e}\in \F_q^n \mid 1\leq \wt(\ve{e}) \leq d-1\}$ and using
the first condition of the lemma, we see that more than half of the codes in $\mycode{C}$ have empty intersection with $E$, that is, have minimum distance at least $d$.
Similarly, more than half of all $q$-ary $[n,n-k]$ codes have minimum distance at least $d^{\perp}$, and so more than half of the codes in $\mycode{C}$ have a dual with minimum distance at least $d^{\perp}$. We conclude that $\mycode{C}$ contains a code with both desired properties.
\end{proof}
\begin{thmmystyle} [Gilbert-Varshamov-like bound by Construction~\ref{cons-ext}]
	\label{th12_new} 
	Let $q$ be a prime power. Suppose the positive integers $u,t,n,k,l$ with $u,t\leq n$ and $l<k$ satisfy
	\begin{equation}\label{eq:GV-like-bound_dual_before_extenstion}
		\sum_{i=0}^{2t} \binom{n}{i} (q-1)^i
		< \frac{1}{2}q^{n-l} ,
	\end{equation}	
	\begin{equation}\label{eq:GV-like-bound_dual_new}
		 \sum_{i=0}^{u-q+2} \binom{n}{i} (q-1)^i
		< \frac{1}{2}q^{l} ,
	\end{equation}
	\begin{equation}\label{eq:GV-like-bound_k_new}
		 \sum_{i=0}^{2t} \binom{n}{i} (q-1)^i < q^{n-k+1} .
	\end{equation}
	Then there is a $q$-ary $(u,1,t)$ PSMC of length $n$ and cardinality $q^{k-l}$.
\end{thmmystyle}
\begin{proof} 
	According to Lemma~\ref{Th:GV-q}, (\ref{eq:GV-like-bound_dual_before_extenstion}) and (\ref{eq:GV-like-bound_dual_new})
imply the existence of an $[n,l]_q$ code $\mycode{C}_0$ with minimum distance at least $2t+1$
for which the dual code has minimum distance at least $u-q+3$. 
Lemma~\ref{lem:GVrefined} shows that $\mycode{C}_0$ can be extended to an $[n,k]_q$ code $\mycode{C}$ with minimum distance at least $d$.
 As $\mycode{C}$ has a generator matrix of the form required by Construction~\ref{cons-ext}, the theorem follows.
\end{proof}
\subsubsection{Finite GV bound based on Proposition~\ref{proposition_binary}}
In this section, we give sufficient conditions for the existence of matrices satisfying the conditions of Proposition~\ref{proposition_binary}.
We start with Lemma~\ref{lm:bind} and Lemma~\ref{Th:GV-s}, then we prove the main theorem (Theorem~\ref{Th:GV-s2}).
\begin{lem}\label{lm:bind} Let $\G$ be a $k\times n$ matrix over $\mathbb{F}_q$. For $s\geq 1$, let 
	\[ d_s = \min \{ \wt(\vmes\G) \mid \vmes\in \mathbb{F}_{q^s}^k \setminus \{\ve{0}\}\}. \]
	Then $d_s=d_1$. 
\end{lem}
\begin{proof}
	The proof of Lemma~\ref{lm:bind} can be found in the appendix. 
\end{proof}
Now we introduce Lemma~\ref{Th:GV-s} which is the binary version of Lemma~\ref{Th:GV-q} but with an extra restriction on the weight of the words.
\begin{lem}\label{Th:GV-s}
	Let $k,n,d$ and $d^{\perp}$ be integers such that
	\[ \sum_{i=0}^{d-1} {n\choose i} < \frac{1}{4} \cdot 2^{n-k} \mbox{ and }
	\sum_{i=0}^{d^{\perp}-1} {n \choose i} <\frac{1}{2} \cdot 2^k. \]
	There exists a binary $[n,k]$ code $\mycode{C}$ with minimum distance at least $d$ without a word of weight more than $n-d+1$ such that $\mycode{C}^{\perp}$ has minimum distance at least $d^{\perp}$.
\end{lem}
\begin{proof}
Similar to the proof of Lemma~\ref{Th:GV-q}.
Let $\mycode{C}$ denote the set of all binary $[n,k]$ codes.
By applying Lemma~\ref{Lem:Averaging} with $E=\{ \ve{e}\in F_2^n\mid 1\leq \wt(\ve{e}) \leq d-1 \text{ or }\wt(\ve{e})\geq n-d+1\}$,
we infer that the first inequality implies that more than half of the codes in $\mycode{C}$ contain no element from $E$. 
	Similarly, the second inequality implies that more than half of the binary $[n,n-k]$ codes have minimum weight at least $d^{\perp}$, and so more than half of the codes in $\mycode{C}$ have a dual with minimum distance at least $d^{\perp}$. 
We conclude that there is a code 
in $\mycode{C}$ enjoying both desired properties.
\end{proof}
\begin{thmmystyle} [Gilbert-Varshamov-like bound by Construction~\ref{consbinary}]\label{Th:GV-s2}
	Let $n,k,l,u,t,\pow$ be positive integers with $u\leq n, 2t< n$ and $l<k$ be such that
	\begin{equation}\label{eq:GV-like-bound_bin}
	 \sum_{i=0}^{2t} {n\choose i} < \frac{1}{4}\cdot 2^{n-l},
\end{equation}
\begin{equation}\label{eq:GV-like-bound_dual_bin}
	\sum_{i=0}^{\lfloor \frac{u}{2^{\pow-1}}\rfloor} {n\choose i}<\frac{1}{2}\cdot 2^{l} ,
\end{equation}
\begin{equation}\label{eq:GV-like-bound_k_bin} 
	\sum_{i=0}^{2t} {n\choose i} (2^{\pow}-1)^i < 2^{\pow(n-k+1)},
\end{equation}
Then there exists a $(u,1,t)$ PSMC of length $n$ over $\F_{2^{\pow}}$ with cardinality
$2\cdot 2^{\pow(k-l-1)}2^{l(\pow-1)}$.
\end{thmmystyle}
\begin{proof} 
	By Lemma~\ref{Th:GV-s}, 
	there exists a binary $[n,l]$ code $\mycode{C}_0$ with minimum distance at least $2t+1$ for which
$\mycode{C}_0^{\perp}$ has minimum distance at least $\lfloor \frac{u}{2^{\pow-1}}\rfloor+1$ with the following additional property: if $\H_0\in \F_2^{l\times n}$ is a generator matrix for $\mycode{C}_0$, then the binary code $\mycode{C}_{\pow}$ with generator matrix 
	\scalebox{0.8}{$\begin{bmatrix} \ve{H}_0 \\ \ve{1} \end{bmatrix}$}
	has minimum distance at least $2t+1$.
	According to Lemma~\ref{lm:bind}, the code $\mycode{C}_{\pow}$ over $\mathbb{F}_{2^{\pow}}$ with this generator matrix has
	minimum distance at least $2t+1$ as well.
 Lemma~\ref{lem:GVrefined} implies that $\mycode{C}_{\pow}$ can be extended to an $[n,k]_{2^{\pow}}$ code
with minimum distance at least $2t+1$. The $[n,k]$ code has a generator matrix of the form
\scalebox{0.8}{$\ve{G} = \begin{bmatrix} \ve{H}_0 \\\ve{G}_1\\ \ve{1} \end{bmatrix}$}.
 Application of Proposition~\ref{proposition_binary} yields the claim.
\end{proof}
\subsubsection{Finite GV bound from trivial construction}
Clearly, a $(q-1)$-ary code of length $n$ with minimum distance at least $2t+1$ is a $q$-ary $(u,1,t)$ PSMC of length $n$.
Combining this observation with the Gilbert bound for a $q-1$-ary alphabet, we obtain the following.
\begin{thmmystyle}\label{th:GVtrivial}
Let $q\geq 3$, and let 
\begin{equation*}
	M= \left\lceil \dfrac{(q-1)^n}{\sum_{i=0}^{2t} {n\choose i} (q-2)^i} \right\rceil.
\end{equation*}
There is a $q$-ary $(n,1,t)$ PSMC of length $n$ and cardinality~$M$.
\end{thmmystyle}
So far we have covered the GV-like bounds for our code constructions for finite length $n$. 
\subsection{Asymptotic GV Bound on PSMCs}
In this section, we present the asymptotic version of the GV bounds from the previous section. That is,
we provide lower bounds on the achievable rates of a $q$-ary $(u,1,t)$ PSMCs in the regime that the code length $n$ tends to infinity, and the number $u$ of partial stuck-at cells and the number $t$ of random errors both grow linearly in $n$.

We recall the well-known following lemma that estimates the volume of a Hamming ball using the $q$-ary entropy function.
\begin{lem}\label{es:Hamming_ball} For positive integers $n$, $q\geq 2$ and real $\delta$, $0\leq \delta \leq 1-\frac{1}{q}$,
	\[
	\Vol_q(n,\delta n) \leq q^{h_q(\delta)n},
	\]
where $\Vol_q(n,r)= \sum_{j=0}^r {n\choose j} (q-1)^j$	denotes the volume of a Hamming ball with radius $r$.
\end{lem}
\begin{proof}
	The proof of Lemma~\ref{es:Hamming_ball} has been stated in many references including \cite[p.105]{rotmroth} and
\cite[Proposition 3.3.1]{GuRuSuBook}.
\end{proof}
 \subsubsection{Asymptotic bound for Theorem~\ref{th:upartial}}
\begin{thmmystyle}\label{upartial_thm}
Let $q$ be a prime power. Let $0\leq \tau,\upsilon<1$ be such that
\[ 2(\tau+\frac{\upsilon}{q}) < 1-\frac{1}{q}. \]
For sufficiently large $n$, there exists an $(\lfloor \upsilon n\rfloor, 1, \lfloor \tau n\rfloor)$ PSMC of length $n$ and rate
at least 
$$ 1-h_q(2(\tau+\frac{\upsilon}{q}))- \frac{2}{n}. $$
\end{thmmystyle}
\begin{proof}
Let $n$ be a positive integer such that $\lceil n h_q(2(\tau + \frac{\upsilon}{q})\rceil < n$.
 Let $t=\lfloor \tau n\rfloor$ and $u=\lfloor \upsilon n \rfloor$.
Take $k=n-\lceil nh_q(2\tau + 2\frac{\upsilon}{q})\rceil $.
Lemma~\ref{es:Hamming_ball} implies that Vol$_q(n,2t+2\lfloor \frac{u}{q}\rfloor) \leq q^{n-k}$, and so, 
according to Theorem~\ref{th:upartial}, there is a $q$-ary $(u,1,t)$ PSMC of length $n$ with rate $\frac{k-1}{n}\geq
1-h_q(2(\tau + \frac{\upsilon}{q})) - \frac{2}{n}$.
\end{proof}
\subsubsection{Asymptotic GV bound from Construction~\ref{cons-ext}}
\begin{thmmystyle} [Asymptotic Gilbert-Varshamov-like bound from Theorem~\ref{th12_new}]\label{Th:aGVnew}
   Let $q$ be a prime power. Let $\upsilon,\tau$ be such that
 \[ 0< \upsilon,2\tau < 1- \frac{1}{q} \mbox{ and } h_q(\upsilon)+h_q(2\tau) < 1. \]
For sufficiently large $n$, there exists a $q$-ary $(\lfloor \upsilon n\rfloor,1, \lfloor\tau n \rfloor)$ PSMC of length $n$ and rate
at least 
$$1 - h_q(2\tau) -h_q(\upsilon)-\frac{4\log_q(2)+2}{n}. $$
\end{thmmystyle}
\begin{proof}
Let $n$ be a positive integer. Write $u=\lfloor \upsilon n\rfloor$ and $t=\lfloor \tau n\rfloor$.
Then Vol$_q(n,u-q+2)\leq \mbox{Vol}_q(n,u)$. Hence, by setting
\[ l = \lceil nh_q(\upsilon) +2 \log_q(2)\rceil, \]
 Lemma~\ref{es:Hamming_ball} implies that~(\ref{eq:GV-like-bound_dual_new}) is satisfied. 
 
Similarly, by setting 
\[ k = n- \lceil nh_q(2\tau) + 2\log_q(2)\rceil, \]
it is ensured that~(\ref{eq:GV-like-bound_k_new}) is satisfied.
\\ According to Theorem~\ref{th12_new}, there is a $q$-ary $(u,1,t)$ PSMC of length $n$ and size $q^{k-l}$,
so with rate $k-l$. The choices for $k$ and $l$ show that the theorem is true.
\end{proof}
\begin{rem}
Theorem~\ref{Th:aGVnew} in fact holds for classical stuck-at cells instead of stuck-at-$1$ errors, as follows from considering
the generalization of Theorem~\ref{th12_new} in Proposition~\ref{modified_version_of_Theorem3}, i.e., Heegard's construction~\cite{heegard1983partitioned}.
\end{rem}
\subsubsection{Asymptotic GV bound from Construction~\ref{consbinary}}
\begin{thmmystyle}[Asymptotic Gilbert-Varshamov-like bound from Theorem~\ref{Th:GV-s2}]\label{th:bGV-s}
Let $\pow$ be a positive integer, and let $\upsilon$ and $\tau$ be such that
\[ 0\leq \frac{\upsilon}{2^{\pow-1}} <\frac{1}{2}, 0 < 2 \tau <\frac{1}{2}, \mbox{ and }h_2(\frac{\upsilon}{2^{\pow-1}}) + h_2(2\tau) < 1. \]
For sufficiently large $n$ there is a $2^{\pow}$-ary
 $(\lfloor \upsilon n\rfloor,1,\lfloor \tau n\rfloor)$ PSMC of length $n$ and
rate at least 
$$1-h_{2^{\pow}}(2\tau) -\frac{1}{\mu}h_2(\frac{\upsilon}{2^{\pow-1}}) -\frac{2}{n} - \frac{3}{\mu n}. $$
\end{thmmystyle}
\begin{proof}
For notational convenience, we set $\upsilon_0=\frac{\upsilon}{2^{\pow-1}}$ 
and $\eta=1-h_2(2\tau) - h_2(\upsilon_0)$. Note that $\eta >0$. \\
Let $n$ be a positive integer satisfying $n\geq \frac{7}{\eta}$,
and let $u=\lfloor \upsilon n\rfloor$, $u_0=\lfloor \frac{u}{2^{\pow-1}}\rfloor$ and $t=\lfloor \tau n\rfloor$.
We set 
\[ l = \lceil nh_2(\upsilon_0) \rceil + 3 . \]
 Lemma~\ref{es:Hamming_ball} implies that (\ref{eq:GV-like-bound_dual_bin}) is satisfied.
Moreover, as
\[ n-l-3 \geq n-nh_2(\upsilon_0) - 7 = nh_2(2\tau) + n\eta -7 \geq nh_2(2\tau), \]
Lemma~\ref{es:Hamming_ball} implies that
 (\ref{eq:GV-like-bound_bin}) is satisfied.

We set 
\[ k= n - \lceil n h_{2^{\pow}}(2\tau) \rceil. \]
Lemma~\ref{es:Hamming_ball} implies that (\ref{eq:GV-like-bound_k_bin}) is satisfied.

According to \cite[Corollary 3.3.4]{GuRuSuBook}, we have that $h_{2^{\pow}}(2\tau)\leq h_2(2\tau)$, and so
\[ k-l \geq n- nh_2(2\tau) -1 - nh_2(\upsilon_0) - 4 = n\eta - 5 \geq 2. \]
Theorem~\ref{Th:GV-s2} implies the existence of a $2^{\pow}$-ary $(u,1,t)$ PSMC of length $n$ 
with size $2\cdot 2^{\pow(k-1)}2^{-l}$, i.e., its rate is
\[ \frac{k-1}{n} - \frac{l-1}{\mu n} \geq 1-h_{2^{\pow}}(2\tau)-\frac{1}{\mu} h_2(\upsilon_0) - \frac{2}{n}-\frac{3}{\mu n}.\tag*{\qedhere} \]
\end{proof}
\begin{thmmystyle}[Asymptotic Gilbert-Varshamov bound from Theorem~\ref{th:GVtrivial}]
Let $q\geq 3$. For each positive integer $n$ and each $\tau$ with $0\leq 2\tau < 1-\frac{1}{q-1}$,
there exists a $q$-ary $(n,1,\lfloor \tau n\rfloor)$ PSMC of length $n$ and rate at least
\[ (1-h_{q-1}(2\tau))\cdot \log_q(q-1). \]
\end{thmmystyle}
\begin{proof}
Let $t=\lfloor \tau n\rfloor$. Theorem~\ref{th:GVtrivial} implies the existence of a $q$-ary $(n,1,t)$ PSMC of length $n$
and cardinality $M$ satisfying 
\begin{align*}
	M \geq \dfrac{(q-1)^n}{V_{q-1}(n,2t)} \geq (q-1)^{n(1-h_{q-1}(2\tau))},
\end{align*}
where the last inequality holds by Lemma~\ref{es:Hamming_ball}.
\end{proof}
\section{Comparisons}\label{comparisons}
We provide different comparisons between our code constructions and the existence of the code based on Theorem~\ref{th:upartial}, Theorem~\ref{th12_new} and Theorem~\ref{Th:GV-s2}. 
Next, we also compare to the known limits and investigate the trade-off between 
masking and error-correction as described in Section~\ref{section_trading}. 
\subsection{Comparison of Theorem~\ref{th:upartial} for $u \leq q-1$ to other Bounds
}
Figure~\ref{const1_with_other_bounds} illustrates the rates of a $(q-1,1,t)$ PSMC obtained from Theorem~\ref{cons:matrix_construction:u<q} (applying Theorem~\ref{th:upartial} for the special case $u\leq q-1$) for $n=114, q=7$ and $0\leq t\leq 56$.
We show how close explicit BCH codes that contains the all-one word of certain rates $R$ and that can correct designed distances $d \geq 2t+1$ to the achieved rates from Theorem~\ref{cons:matrix_construction:u<q}. We note that the solid red graph matches the dashed-dotted green plot for a few code parameters and overpasses it for $t=39$. We also compare to the classical $q$-ary GV bound (in dashed black) as well as to reduced alphabet $(q-1)$-ary GV bound (in dashed-dotted blue).
To this end, we show upper bounds on the rates that can be obtained from Theorem~\ref{cons:matrix_construction:u<q} using the Griesmer bound \cite{griesmer1960bound}, and the Ball--Blokhuis bound \cite{ball2013bound} on the size of codes containing the all-one word.
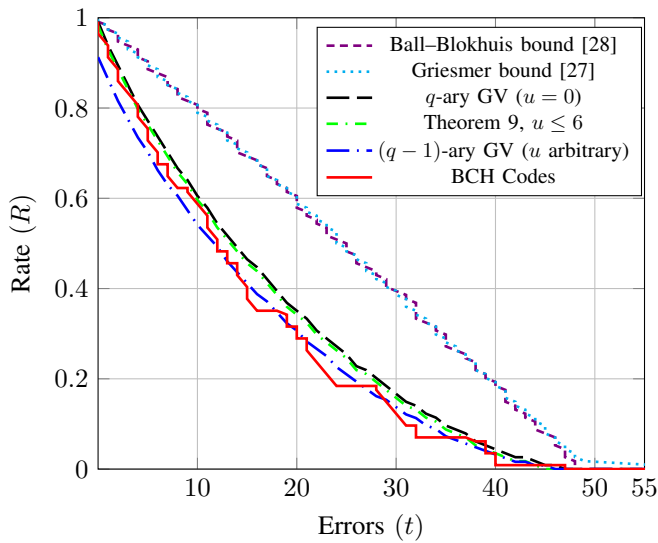
\begin{figure} [htp] 
	\begin{center}	
		\scalebox{1}{	
			\begin{tikzpicture}
				\begin{axis}[
					width = \columnwidth,
					xlabel = {Errors $(t)$},
					ylabel = {Rate $(R)$},
					xmin = 0,
					xmax = 55.0,
					xtick={10,20,30,40,50,55},
					ytick={0,0.2,0.4,0.6,0.8,1},
					ymin = 0,
					ymax = 1.0,
					grid=major,
					legend style={ at={(0.40,0.98)},nodes={scale=0.8, transform shape}, anchor= north west},
					cycle list name=solidlinesbyhaider72
					]		
					\def\mymark{x}
						
						\input{figures/comparisons/BCH_match_const1_GV_bound/plot_list_ball_blokhuis_bound}
						\input{figures/comparisons/BCH_match_const1_GV_bound/plot_list_griesmer_bound}
						\input{figures/comparisons/BCH_match_const1_GV_bound/q_ary_usual_GV_n114}
						\input{figures/comparisons/BCH_match_const1_GV_bound/plot_list_construction_1_n114}
						\input{figures/comparisons/BCH_match_const1_GV_bound/ordGVq7_n114}
						\input{figures/comparisons/BCH_match_const1_GV_bound/plot_different_k_for_BCH_Codes}
					\end{axis}
			\end{tikzpicture}}
			\caption{Comparison of other upper and lower limits to our derived GV-like bound in Theorem~\ref{th:upartial} taking $n=114$, $q=7$, $0\leq t\leq 56$ and $u \leq q-1$. The dashed-dotted green curve shows the rates for Theorem~\ref{cons:matrix_construction:u<q} by Theorem~\ref{th:upartial} for $u \leq q-1$ in which codes that have the all-one words are considered. This curve for several code parameters matches the red line that shows how the rates of BCH codes that contain all-one word with regard to the designed distances $d \geq 2t+1$. 
			} \label{const1_with_other_bounds}
		\end{center}
	\end{figure}
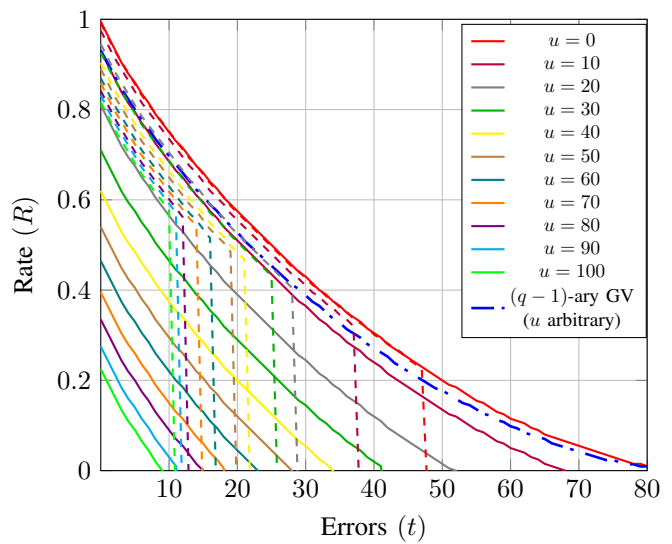
\begin{figure} [htp] 
	\begin{center}	
		\scalebox{1}{	
			\begin{tikzpicture}
				\begin{axis}[
					width = \columnwidth,
					xlabel = {Errors $(t)$},
					ylabel = {Rate $(R)$},
					xmin = 0,
					xmax = 80.0,
					xtick={10,20,30,40,50,60,70,80},
					ytick={0,0.2,0.4,0.6,0.8,1},
					ymin = 0,
					ymax = 1.0,
					grid=major,
					legend style={cells={align=center, scale=0.73, transform shape}, anchor= north west,  at={(0.66,0.99)}},
					cycle list name=solidlinesbyhaider21
					]		
					\def\mymark{x}
						\input{figures/comparisons/GV_const_2_3_q8/plot_list_const_2_q8}
						\input{figures/comparisons/Const1_to_ord_GV/ordGvq7}
						\input{figures/comparisons/GV_const_2_3_q8/plot_list_constbin_q8}
				\end{axis}
		\end{tikzpicture}}
		\caption{The achievable rates $R=\frac{1}{n}\log_{2^3} \mathcal{M}$ 
			of GV bounds for different $u,t$ for $n=200$ and $q=2^3$ in Theorem~\ref{th12_new} and Theorem~\ref{Th:GV-s2}, where $\mathcal{M}$ is the code cardinality. They are also compared to the rates from an ordinary $7$-ary GV bound for different $t$ as illustrated in the dashed-dotted blue plot. The solid and the dashed lines represent the derived GV like bounds from Theorem~\ref{th12_new} and Theorem~\ref{Th:GV-s2}, respectively. 
		}\label{fig4}
	\end{center}	
\end{figure}
\begin{figure} [htp] 
	\begin{center}	
		\scalebox{1}{
			\begin{tikzpicture}
				\begin{axis}[
					width = \columnwidth,
					xlabel = {Errors $(t)$},
					ylabel = {Rate $(R)$},
					xmin = 0,
					xmax = 80.0,
					xtick={10,20,30,40,50,60,70,80},
					ytick={0,0.2,0.4,0.6,0.8,1},
					ymin = 0,
					ymax = 1.0,
					grid=major,
					legend style={cells={align=center, scale=0.73, transform shape}, anchor= north west,  at={(0.66,0.99)}},
					cycle list name=solidlinesbyhaider21
					]	
					\def\mymark{x}
						\input{figures/comparisons/GV_const2_q4/plot_list_const_2_q4}
						\input{figures/comparisons/GV_const2_q4/ordGVq3}
						\input{figures/comparisons/GV_const2_q4/plot_list_constbin_q4}		
				\end{axis}
				q		\end{tikzpicture}}
		\caption{
			The achievable rates $R=\frac{1}{n}\log_{2^2} \mathcal{M}$ of GV bounds for different $u,t$ for $n=200$ and $q=2^2$ in Theorem~\ref{th12_new} and Theorem~\ref{Th:GV-s2}. They are also compared to the rates from an ordinary $3$-ary GV bound for different $t$ as illustrated in the dashed-dotted blue plot. The solid and the dashed lines correspond to the derived GV like bounds by Theorem~\ref{th12_new} and by Theorem~\ref{Th:GV-s2}, respectively.} \label{figGVcomparewithq4}. 
	\end{center}
\end{figure}
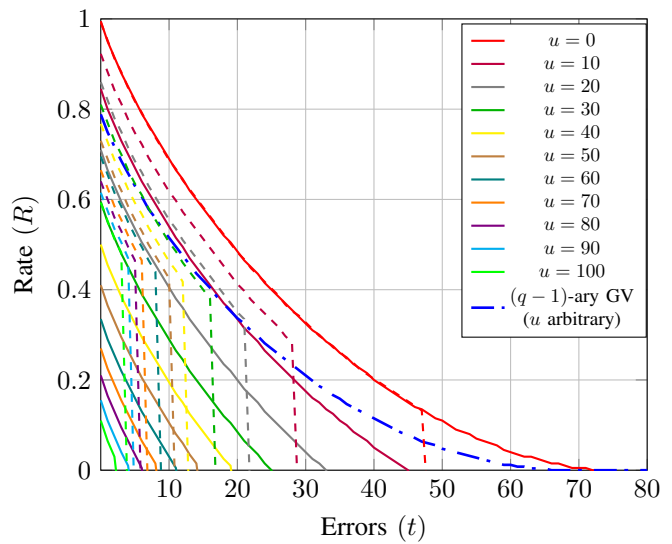
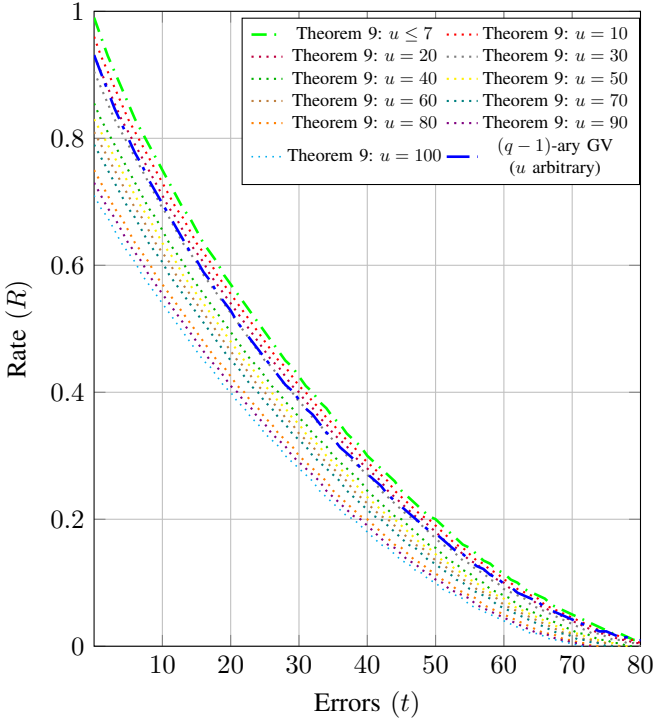
\begin{figure} [htp] 
	\begin{center}	
		\scalebox{1}{
			\begin{tikzpicture}
				\begin{axis}[
					width = \columnwidth,
					height=10cm,
					xlabel = {Errors $(t)$},
					ylabel = {Rate $(R)$},
					xmin = 0,
					xmax = 80.0,
					xtick={10,20,30,40,50,60,70,80},
					ytick={0,0.2,0.4,0.6,0.8,1},
					ymin = 0,
					ymax = 1.0,
					grid=major,
					legend style={cells={align=center, scale=0.68, transform shape}, anchor= north west ,  at={(0.27,0.99)}, legend columns=2},
					cycle list name=solidlinesbyhaider11
					]	
					\def\mymark{x}
						
						\input{figures/comparisons/Const1_to_ord_GV/plot_list_construction_1}
						\input{figures/comparisons/Const1_to_ord_GV/GV_treat_u_as_t}	
						\input{figures/comparisons/Const1_to_ord_GV/ordGVq7}
				\end{axis}
		\end{tikzpicture}}
		\caption{The achievable rates $R= \frac{1}{n} \log_{2^3} \mathcal{M}$ of GV bounds for different $u,t$ for $n=200$ and $q=2^3$ in Theorem~\ref{th:upartial} that are compared to the reduced alphabet conventional $(q-1)$-ary GV bound for different $t$. The dashed-dotted green curve represents the rates from Theorem~\ref{th:upartial} when $u \leq q-1$.
		}\label{compare_with_ordinary_GV}
	\end{center}
\end{figure}
\subsection{ Comparison among Theorem~\ref{th12_new}, Theorem~\ref{Th:GV-s2} and $(q-1)$-ary Gilbert-Varshamov bound}
We plot the achievable rates $(R=log_q(M)/n)$ as a function of $t$ for different fixed values of $u$. Figure~\ref{fig4} is the resulting plot for $n=200$ and $q=2^3$. 
It can be seen that the GV-like bound in different ranges of $u$ and $t$ based on Construction~\ref{cons-ext} improves upon the ($q-1$)-ary GV bound for $u\leq 5$ as depicted in the solid red curve, and improves further (up to $u\leq 20$) based on Construction~\ref{consbinary} as shown in the dashed gray line (3rd one from above). The dashed dotted blue curve is used to see what if we map our $2^3$ levels such that we avoid the subscript $0$ to compare with $7$ levels. 
It is obvious that for $\pow=3$, the rate loss\footnote{For $t=0$, the loss is $1 - \log_8(7)= 0.0642$.} resulting from using $q-1$ instead of $q$ symbols is already quite small. 
Note that for $u=0$ the solid red curve mostly achieves the exact rates obtained from the standard $2^3$-ary GV bound for $0 \leq t \leq 80$, and so as for the dashed red curve but for $0 \leq t \leq 47$.

For $\pow=2$, the improvements from Construction~\ref{cons-ext} ($u \leq 10$) and Construction~\ref{consbinary} ($u \leq 30$) upon a usual ($q-1$)-ary GV bound are more significant as shown in Figure~\ref{figGVcomparewithq4}. 
\subsection{Comparisons between Theorem~\ref{th:upartial} and $(q-1)$-ary Gilbert-Varshamov bound}
In Figure~\ref{compare_with_ordinary_GV}, we compare the GV like bound from Theorem~\ref{th:upartial} for $q=2^3$ with the conventional GV bound for $q-1=7$ shown in dashed black curve.
We see the dashed-dotted green curve by Theorem~\ref{th:upartial} for ($u\leq 7$ as stated in Remark~\ref{special_case_of_th:upartial}).
For $q=8$, we observe that the conventional $q-1$-ary GV bound is superior to the derived GV-like bound from Theorem~\ref{th:upartial} and many larger values of $u$. However, applying Theorem~\ref{th:upartial} where $u \leq 7$, the traditional $q-1$-ary GV bound is a bad choice. 
\subsection{Comparisons between Theorem~\ref{Th:GV-s2} and Theorem~\ref{th:upartial}}
In Figure~\ref{compare_trading_const1_and_const3}, we compare Theorem~\ref{Th:GV-s2} and Theorem~\ref{th:upartial}.
	Theorem~\ref{Th:GV-s2} is showing higher rates for larger $u$ values,
for example taking $u= 40$ and $t=1$, the rate is $R = 0.87$ from Theorem~\ref{Th:GV-s2} while $R = 0.83$ from Theorem~\ref{th:upartial}.
 It is interesting to note that for $u=30$ and $t> 10$ Theorem~\ref{th:upartial} is better, and for $u =10$ and $t > 18$ Theorem~\ref{th:upartial} is as good as Theorem~\ref{Th:GV-s2}.
	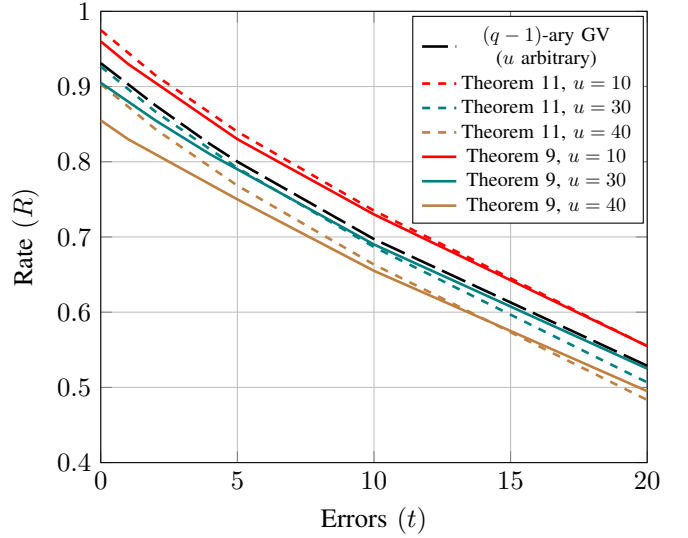
\begin{figure} [htp] 
		\begin{center}	
			\scalebox{1}{	
				\begin{tikzpicture}
					\begin{axis}[
						width = \columnwidth,
						xlabel = {Errors $(t)$},
						ylabel = {Rate $(R)$},
						xmin = 0,
						xmax = 20.0,
						xtick={0,5,10,15,20,25},
						ytick={0.4,0.5,0.6,0.7,0.8,0.9,1},
						ymin = 0.4,
						ymax = 1.0,
						grid=major,
						legend style={cells={align=center, scale=0.75, transform shape}, anchor= north west, at={(0.57, 0.99)}},
						cycle list name=solidlinesbyhaider9
						]		
						\def\mymark{x}
							
							\input{figures/comparisons/Trading_Com3_Th11/q_minus_1_ord_GVq7_few}
							\input{figures/comparisons/Trading_Com3_Th11/plot_list_constbin_few}
							\input{figures/comparisons/Trading_Com3_Th11/GV_treat_u_as_t}	
					\end{axis}
			\end{tikzpicture}}
			\caption{
				The achievable rates $R=\frac{1}{n}\log_{q} \mathcal{M}$ of GV bounds for different $u =\{10,30,40\},\mbox{ and } t =\{0,1,2,4, 5, 10,20\}$ for $n=200$ and $q=2^3$ in Theorem~\ref{Th:GV-s2} and Theorem~\ref{th:upartial}. They are also compared to the rates for $(q-1)$-ary GV bound as shown in dashed black graph.}\label{compare_trading_const1_and_const3}
		\end{center}
	\end{figure}
\subsection{Comparisons of application of Theorem~\ref{tradingu_t} vs direct application of Theorem~\ref{th12_new}}
For given $(u,t)$, we illustrate the trading $(u+1, t-1)$ in Figure~\ref{figtrade_const2}. For some of $t$ and a few of $u$ values, it is advantageous if the encoder introduces an error in a partially stuck at position in order to mask this position. The orange solid curve, for instance, represents the rates that have been determined by Theorem~\ref{th12_new} for $u=17$ and $0\leq t\leq 50$, while the orange dotted sketch highlights the rates for $u =16$ while $1\leq t\leq 51$. Due to the exchange such that $u+1 =17$ and $0 \leq t-1 \leq 50$, the orange dotted line slightly fluctuates up and down the rates shown in the orange solid curve for most $t$ values. 
	
Let us describe some points of Figure~\ref{figtrade_const2} in Table~\ref{table4}. Let $\mycode{C}_{u,t}$ be a code by Theorem~\ref{th12_new} whose rate is $R$ given in Table~\ref{table4} at $u$ row and $t$ column. Take $\mycode{C}_{21,15}$ so that its rate $R = \textcolor{red}{\textbf{0.470}}$. 
By applying Theorem~\ref{tradingu_t} on $\mycode{C}_{21,15}$, we obtain a code $\mycode{C}_{22,14}$ of $R= 0.470$. Direct application of Theorem~\ref{th12_new} yields a $\mycode{C}_{22,14}$
of rate $R= \textcolor{red}{\textbf{0.475}}$ as highlighted in Table~\ref{table4}. We conclude that in this case, the trade by Theorem~\ref{tradingu_t} gives lower rates than taking the same code directly by Theorem~\ref{th12_new} for given $(u=22,t=14)$. 
	
On contrary, for larger $t$ values, Table~\ref{table4} shows that the exchange is beneficial giving higher rates. For example, we start with $\mycode{C}_{21,41}$ whose $R=\textcolor{green!50!black}{\textbf{0.105}}$, then applying Theorem~\ref{tradingu_t} gives $\mycode{C}_{22,40}$ of $R=0.105$ which is greater than $R=\textcolor{green!50!black}{\textbf{0.100}}$ that has been obtained directly by Theorem~\ref{th12_new} as stated in Table~\ref{table4}.
	\begin{table} 
	\caption{ Table of Selected Points from Figure~\ref{figtrade_const2} with slightly lower and higher rates due to trading. All Points are from Theorem~\ref{th12_new}.
	}
	\label{table4}
	\begin{center}
		\scalebox{1}{
			\begin{tabular}{ | l ||l |l| l | l |l |l |l |}
				\hline
				\backslashbox{$u$}{$t$}&$13$&$14$&$15$&$\dots$ &$40$&$41$ & $42$
				\\ \hline 
				$16$ & 0.560 & 0.545& \textcolor{red}{\textbf{0.525}}& $\dots$ &0.170 & \textcolor{green!50!black}{\textbf{0.160}} & 0.150
				\\ \hline
				$17$ & 0.545 & \textcolor{red}{\textbf{0.530}} & 0.510 & $\dots$ &\textcolor{green!40!black}{\textbf{0.155}} &0.145 & 0.135 
				\\ \hline
				$\vdots$ & $\vdots$ & $\vdots$& $\vdots$ & $\vdots$ &$\vdots$ & $\vdots$& $\vdots$ 
				\\ \hline
				$21$ &0.505&0.490&\textcolor{red}{\textbf{0.470}}& $\dots$ &	0.115&\textcolor{green!50!black}{\textbf{0.105} } & 0.095
				\\ \hline
				$22$ & 0.490 & \textcolor{red}{\textbf{0.475}} & 0.455 & $\dots$ &\textcolor{green!50!black}{\textbf{0.100}} & 0.090 & 0.080
				\\ \hline
				$23$ & 0.480& 0.465& 0.445 & $\dots$ &0.090 &0.080 & 0.070 
				\\ \hline
			\end{tabular}
		}
	\end{center}
\end{table}
	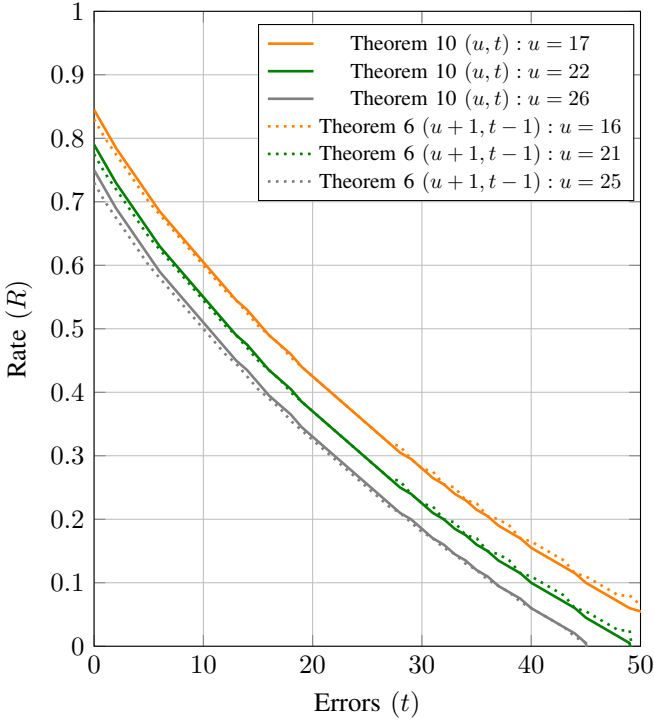
\begin{figure} [htp] 
	\begin{center}	
		\scalebox{1}{
			\begin{tikzpicture}
				\begin{axis}[
					width = \columnwidth,
					height = 10 cm,
					xlabel = {Errors $(t)$},
					ylabel = {Rate $(R)$},
					xmin = 0,
					xmax = 50.0,
					xtick={0,10,20,30,40,50},
					ytick={0, 0.1,0.2,0.3,0.4,0.5,0.6,0.7,0.8,0.9,1},
					ymin = 0,
					ymax = 1.0,
					grid=major,
					legend style={ at={(0.30,0.98)},nodes={scale=0.8, transform shape}, anchor= north west},
					cycle list name=solidlinesbyhaider12
					]		
					\def\mymark{x}
						
						\input{figures/comparisons/Trading_Theorem_12/plot_list_const_2}
						\input{figures/comparisons/Trading_Theorem_12/plot_list_after_const2_trading}	
				\end{axis}
		\end{tikzpicture}}
		\caption{
			The achievable rates $R=\frac{1}{n}\log_{q} \mathcal{M}$ of GV bounds for different $u,\mbox{ and } t$ for $n=200$ and $q=2^3$ in Theorem~\ref{th12_new}. The solid plots are the rates from the derived GV-like bound and the dotted graphs are the rates after trading $u+1, t-1$ by Theorem~\ref{tradingu_t}.
		}\label{figtrade_const2}
	\end{center}
\end{figure}
	\subsection{Comparisons of applications of Theorem~\ref{tradingu_t}, Lemma~\ref{lem03}, Lemma~\ref{lem01} vs direct application of Theorem~\ref{Th:GV-s2}}
		For the derived GV bound based on Construction~\ref{consbinary} obtained by Theorem~\ref{Th:GV-s2}, 
	we demonstrate the exchange of a one error correction ability with a single masking capability of a partially stuck cell following Theorem~\ref{tradingu_t} in Figure~\ref{figtradebinary}. 
	The solid and dotted lines represent the rates before and after trading, respectively. We also show the exchange by Lemma~\ref{lem03} and Lemma~\ref{lem01} in which the reduction of the correctable errors by one increases $u$ by $2^{\pow-1}$ and $2^\pow$, respectively.
	Let us discuss the following curves. For $u=19$, the orange solid curve shows the rates by Theorem~\ref{Th:GV-s2}. Exchanging $u+1$ and $t-1$ throughout Theorem~\ref{tradingu_t} obtains the orange dotted line for $u+1 = 20$ which lies slightly bellow the orange solid plot. 
	Hence, the exchange gives lower rates but 
	provides rate $R =0.380$ for $t=30$ while direct application of Theorem~\ref{Th:GV-s2} (compared to its corresponding graph which is the solid green curve at $u = 20,21,22,23$) does not.	
Now, we apply Lemma~\ref{lem03} rather Theorem~\ref{tradingu_t}. We observe the dashed red graph for $u+2^{3-1} = 23$ shows the exact rates from the orange dotted curve for $u+1 = 20$. Therefore, it is clear that Lemma~\ref{lem03} provides a gain of masking exactly $3$ more cells with regard to Theorem~\ref{tradingu_t}. 
	
However, if we take $u=23$ directly by Theorem~\ref{Th:GV-s2}, we achieve slightly higher rates. 
We conclude that Theorem~\ref{Th:GV-s2} can directly estimate the maximum possible masked $u$ cells that can also be achieved applying 
Lemma~\ref{lem03}, and can achieve slightly higher rates. However, Theorem~\ref{Th:GV-s2} does not give rates for larger $t$ values while Lemma~\ref{lem03} and Theorem~\ref{tradingu_t} do that. 

On the other hand, as Theorem~\ref{Th:GV-s2} is based on Construction~\ref{consbinary} that contains a word of weight $n$, Lemma~\ref{lem01} is applicable under the condition that $2(t+ \lfloor\frac{u}{q}\rfloor) <d$ (cf. Remark~\ref{rem11}). Hence, we can achieve higher rates while masking up to the same number of $u$ cells than employing Lemma~\ref{lem03} or Theorem~\ref{tradingu_t} as shown in the dashed-dotted curve.

For that we describe some points of Figure~\ref{figtradebinary} by Table~\ref{table5}. Let $\mycode{C}_{u,t}$ be 
a code by Theorem~\ref{Th:GV-s2} whose rate is $R$ given in Table~\ref{table5} at $u$ row and $t$ column.
Taking $\mycode{C}_{19, 27}$ gives 
$\mycode{C}_{20, 26}\mbox{ and } \mycode{C}_{23, 26}$ with $R = 0.435$ applying Theorem~\ref{tradingu_t} and Lemma~\ref{lem03}, respectively. In contrary, taking $\mycode{C}_{19, 31}$ is advantageous as there are codes ($\mycode{C}_{20, 30}$ by Theorem~\ref{tradingu_t} and $\mycode{C}_{23, 30}$ by Lemma~\ref{lem03}) with $R =0.380$ while direct application of Theorem~\ref{Th:GV-s2} cannot provide these codes 
as highlighted in green with "\textcolor{green!50!black}{\textbf{None}}". Now, we apply Lemma~\ref{lem01} on a code obtained by Theorem~\ref{th:upartial} for $(u=7, t=27)$ to obtain the code $\mycode{C}_{15, 26}$ of rate $R = 0.465$ that satisfies $2(26+\lfloor \frac{15}{8} \rfloor) < 55$. The achieved rate is higher compared to $\mycode{C}_{15, 26}$ of $R = 0.460$ that is directly obtained by Theorem~\ref{Th:GV-s2}, or applying Theorem~\ref{tradingu_t} on $\mycode{C}_{14, 27}$ to obtain $\mycode{C}_{15, 26}$ of $R = 0.445$, or using Lemma~\ref{lem03} on $\mycode{C}_{11, 27}$ to obtain $\mycode{C}_{15, 26}$ of $R = 0.456$. This result does not mean that application Lemma~\ref{lem01} on a code obtained by Theorem~\ref{th:upartial} always provides higher code rates for the same parameters $u,t$ (see Figure~\ref{compare_trading_const1_and_const3}).
	\begin{table}
	\caption{ Table of Selected Points from Figure~\ref{figtradebinary}. All Points are from Theorem~\ref{Th:GV-s2}.
	}
	\label{table5}
	\begin{center}
		\scalebox{1}{
			\begin{tabular}{ | l ||l |l| l | l |l |l |l |}
				\hline
				\backslashbox{$u$}{$t$}&$26$&$27$&$28$&$29$&$30$ &$31$&$32$
				\\ \hline 
				$11$ & 0.471& 0.456 & 0.441&0.431 &0.416 & 0.401 & 0.391
				\\ \hline
				$12$ & 0.460 & 0.445& 0.430& 0.420 & 0.405 & 0.390 & 0.380
				\\ \hline
				$\vdots$ & $\vdots$ & $\vdots$& $\vdots$ & $\vdots$ &$\vdots$ & $\vdots$& $\vdots$ 
				\\ \hline
				$14$ & 0.460 & 0.445& 0.430& 0.420 & 0.405 & 0.390 & 0.380
				\\ \hline
				$15$ & 0.460 & 0.445& 0.430& 0.420 & 0.405 & 0.390 & 0.380
				\\ \hline
				$\vdots$ & $\vdots$ & $\vdots$& $\vdots$ & $\vdots$ &$\vdots$ & $\vdots$& $\vdots$ 
				\\ \hline
				$19$ & 0.450 & \textcolor{red}{\textbf{0.435}}& 0.420& 0.410 & 0.395 & \textcolor{green!50!black}{\textbf{0.380}} & None
				\\ \hline
				$20$ &\textcolor{red}{\textbf{0.441}}& 0.426 &0.411 &0.401 & \textcolor{green!50!black}{\textbf{None}} & None & None
				\\ \hline
				$\vdots$ & $\vdots$ & $\vdots$& $\vdots$ & $\vdots$ &$\vdots$ & $\vdots$& $\vdots$ 
				\\ \hline
				
				$23$ & \textcolor{red}{\textbf{0.441}}& 0.426 &0.411 &0.401 & \textcolor{green!50!black}{\textbf{None}} & None & None
				\\ \hline
			\end{tabular}
		}
	\end{center}
\end{table}
	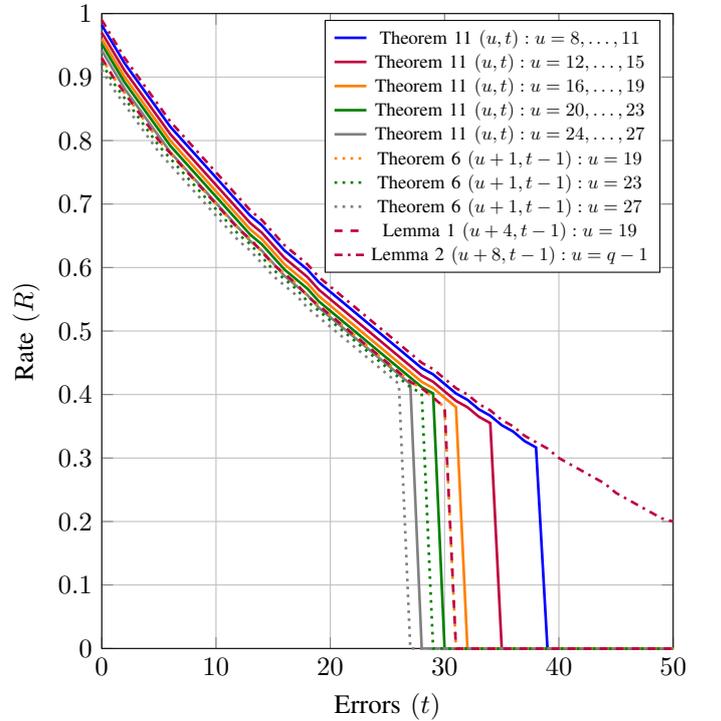
\begin{figure} [htp] 
	\begin{center}	
		\scalebox{1}{	
			\begin{tikzpicture}
				\begin{axis}[
					clip=false,
					width = 9.1 cm,
					height = 10 cm,
					xlabel = {Errors $(t)$},
					ylabel = {Rate $(R)$},
					xmin = 0,
					xmax = 50.0,
					xtick={0,10,20,30,40,50},
					ytick={0, 0.1,0.2,0.3,0.4,0.5,0.6,0.7,0.8,0.9,1},
					ymin = 0,
					ymax = 1.0,	
					grid=major,
					legend style={cells={align=center, scale=0.70, transform shape}, anchor= north west,  at={(0.39,0.99)}},
					cycle list name=solidlinesbyhaider14
					]	
					\def\mymark{x}	
						\input{figures/comparisons/Trading_Theorem_13/plot_list_constbin_rates}
						\input{figures/comparisons/Trading_Theorem_13/plot_list_after_constbin_trading_rates}	
						\input{figures/comparisons/Trading_Theorem_13/plot_list_after_constbin_trading_lemma_rates}
						\input{figures/comparisons/Trading_Theorem_13/plot_list_after_constbin_trading_lemma2_rates} 
				\end{axis}
		\end{tikzpicture}}
		\caption{
			The achievable rates $R=\frac{1}{n}\log_{q} \mathcal{M}$ of GV bounds for different $u,\mbox{ and } t$ for $n=200$ and $q=2^3$ in Theorem~\ref{Th:GV-s2}. The solid plots are the rates from the derived GV like bound and the dotted graphs are the rates after trading $u+1, t-1$ by Theorem~\ref{tradingu_t}. We also show the exchange by Lemma~\ref{lem03} and Lemma~\ref{lem01} in which the reduction of the correctable errors by one increases $u$ by $2^{\pow-1}$ and $2^\pow$, respectively.
		}\label{figtradebinary}
	\end{center}
\end{figure}
\subsection{Analytical comparison of asymptotic GV-like bounds}\label{subsec:ascomp}
In this section, we state the results of the analytical comparisons of the asymptotic GV bounds from Theorems \ref{upartial_thm}, \ref{Th:aGVnew} and \ref{th:bGV-s}, ignoring the 
terms that tend to zero for increasing $n$. The proofs can be found in Appendix~\ref{app:asGVcomp}.
\begin{prop}\label{upartialvsaGVnew}
	If $\upsilon,\tau$ and $q$ are such that the conditions of Theorem~\ref{upartial_thm} and Theorem~\ref{Th:aGVnew} are met,
	then the rate guaranteed by Theorem~\ref{upartial_thm} is at least equal to the code rate guaranteed by Theorem~\ref{Th:aGVnew}.
\end{prop}
\begin{prop}\label{aGVnewvsbGV-s}
	If $\upsilon,\tau$ and $q=2^{\pow}$ are such that the conditions of Theorem~\ref{Th:aGVnew} and Theorem~\ref{th:bGV-s} are met,
	then the code rate guaranteed by Theorem~\ref{th:bGV-s} is at least equal to the code rate guaranteed by Theorem~\ref{Th:aGVnew}.
\end{prop}
We note that the requirement $2\tau < \frac{1}{2}$ from Theorem~\ref{th:bGV-s} is stricter
	than the requirement $2\tau < 1-\frac{1}{2^{\pow}}$ from Theorem~\ref{Th:aGVnew}.
	That is, there are pairs $(\tau,\upsilon)$ for which Theorem~\ref{Th:aGVnew} is applicable,
	but Theorem~\ref{th:bGV-s} is not.
	
	As explained in Appendix~\ref{app:asGVcomp}, comparison of Theorem~\ref{upartial_thm} and Theorem~\ref{th:bGV-s} is more complicated.
	We have the following partial result.
	\begin{prop}\label{upartialvsbGV-s}
		Let $\upsilon,\tau>0$ and $q=2^{\pow}$ be such that the conditions of Theorem~\ref{upartial_thm} and Theorem~\ref{th:bGV-s} are met.
		If $\upsilon$ is sufficiently small, then the rate guaranteed by Theorem~\ref{upartial_thm} is larger than the rate guaranteed by Theorem~\ref{th:bGV-s}.
	\end{prop}

\section{Conclusion}\label{Conclusion}
In this paper, code constructions and bounds for non-volatile memories with partial defects have been proposed. 
Our constructions can handle both: partial defects (also called partially stuck cells) and random substitution errors, and require less redundancy symbols for $u> 1$ and $q>2$ than the known constructions for stuck cells. 
Compared to error-free masking of partially stuck cells, our achieved code sizes coincide with those in \cite{wachterzeh2016codes}, or are even larger as shown in Proposition~\ref{proposition_binary}. We summarize our constructions and the previous works on partially/fully stuck cells in Table~\ref{table1}. 

Further, we have shown that it can be advantageous to introduce errors in some partially stuck cells in order to satisfy the stuck-at constraints.
 For the general case that is applicable for all of our constructions, we have shown in Theorem~\ref{tradingu_t} how to replace any $0 \leq j \leq t$ errors by $j$ masked partially stuck cells. This theorem has been improved for Construction~\ref{consbinary} by Lemma~\ref{lem03}, and further enhanced by another method for introducing errors in the partially stuck locations through Lemma~\ref{lem01} (cf. Example~\ref{example_binary_trading}). We gain (e.g., for $j=1$) exactly $2^{\pow-1}$ and $2^{\pow}$ (under the condition that $2(t+ \lfloor\frac{u}{q}\rfloor) <d$) additional masked partially stuck cells applying Lemma~\ref{lem03} and Lemma~\ref{lem01}, respectively.
So far, determining if introducing errors in partially stuck cells is advantageous or not can only be done numerically.

We also derived upper and lower limits on the size of our constructions. Our sphere-packing-like bound for the size of $(\Sigma,t)$ PSMCs has been compared to the usual sphere-packing upper bound, and for the case of no errors ($t=0$) to \cite[Theorem 2]{wachterzeh2016codes}.

We have numerically compared our Gilbert--Varshamov-type bounds, for given $(u,t)$, to each other and to $(q-1)$-ary codes. For $u \leq q-1$, Theorem~\ref{th:upartial} states the existence of $(u,1,t)$ PSMCs with rates that almost match the ones from the usual $q$-ary GV bound (shown in Figure~\ref{const1_with_other_bounds}). Moreover, up to $u = 20$ for $q=8$, Figure~\ref{compare_with_ordinary_GV} shows that application of Theorem~\ref{th:upartial} is better than using $(q-1)$-ary code as mentioned in \cite[Section III]{wachterzeh2016codes}. 
On the other hand, for $q=4$ and $u =10$, Theorem~\ref{th12_new} and Theorem~\ref{Th:GV-s2} obviously require less redundancies than $(q-1)$-ary code as shown in Figure~\ref{figGVcomparewithq4}.

Figures~\ref{figtrade_const2} and~\ref{figtradebinary} demonstrate the application of Theorem~\ref{tradingu_t}, Lemma~\ref{lem03} and Lemma~\ref{lem01} on $(u,1,t)$ PSMCs of rates that have been obtained based on Theorem~\ref{th12_new} and Theorem~\ref{Th:GV-s2}. For some parameters (i.e. $u=16$, $t=41$ as shown in Table~\ref{table4} and $u=19$, $t=31$ as shown in Table~\ref{table5}), application of Theorem~\ref{tradingu_t} and Lemma~\ref{lem03} achieve higher code rates and more masked cells. Application Lemma~\ref{lem01} on a code obtained by Theorem~\ref{th:upartial} (i.e. $u=7$, $t=27$) provides higher code rate compared to the direct employment of Theorem~\ref{Th:GV-s2}, Theorem~\ref{tradingu_t} and Lemma~\ref{lem03}.

In the asymptotic regime of our GV-like bounds, Theorem~\ref{upartial_thm} and Theorem~\ref{th:bGV-s} are remarkable competitors to Theorem~\ref{Th:aGVnew}. However, the analytical comparison between Theorem~\ref{upartial_thm} and Theorem~\ref{th:bGV-s} is more complicated to decide which one is the better choice. This was also confirmed numerically via Figure~\ref{compare_trading_const1_and_const3}.

\appendix
\subsection{Variant of Construction~\ref{cons:matrix_construction:u<q} ($u<q$) for Cyclic Codes}\label{app:variant_u<q_construction_cyclic_codes}
This section provides an alternative of Construction~\ref{cons:matrix_construction:u<q} by generalizing the construction of \cite[Theorem~2]{heegard1983partitioned}. We use 
the so-called \emph{partitioned cyclic codes} from \cite{heegard1983partitioned} as basic idea, but we require only a single redundancy symbol $l=1$ for the masking operation similar to \cite[Theorem 4 and Algorithm 3]{wachterzeh2016codes}. 
Compared to Construction~\ref{cons:matrix_construction:u<q}, Construction~\ref{cons1} directly implies a constructive strategy on how to choose a cyclic code of a certain minimum distance. 
	In this appendix, we will use the following notation.
If $\mathcal{C}$ is a cyclic code, it has a generator polynomial $g(x)$ of degree $n-k$ with roots in $\mathbb{F}_{q^m}$, where $n$ divides $q^m - 1$. The defining set $D_c$ of $\mathcal{C}$ 
is the set containing the indices $b$ of the zeros $\alpha^b$ of the generator polynomial $g(x)$, i.e., 
\begin{equation}\label{defining_set}
D_c := \{ b: g(\alpha^b)= 0\}.
\end{equation} 
The minimum distance $d$ of $\mycode{C}$ can be bounded from below by the BCH bound $\delta$ or more involved bounds such as the Hartmann-Tzeng bound \cite{HartmannTzengBound1972} or the Roos bound \cite{Roos1979}.
	\begin{const}\label{cons1}
		Let $u\leq \min \{n,q-1\}$. Assume there is an $[n, k, \delta\geq 2t+1]_q$ cyclic code $\mycode{C}$ with a generator polynomial $g(x)$ of degree $< n-k$ that divides $g_0(x) :=1+x+x^2+ \dots +x^{n-1}$.
			Encoder and decoder are given in Algorithms~\ref{a3} and~\ref{a4}. 
	\end{const}
\begin{thmmystyle}\label{thm:cyclic_construction}
	If $u\leq \min \{n,q-1\}$, Construction~\ref{cons1} provides an $(n, M = q^{k-1})_q$ ($u,1,t$)-PSMC with redundancy of $n-k+1$ symbols.
	\label{thm2}
\end{thmmystyle}
\begin{algorithm} \label{a3}
	\caption{Encoding}
	\KwIn{
		\begin{itemize}
			\item Message:
			$\mesp \in \mathbb{F}_q[x]$ of degree $ < k-1$
			\item Positions of partially stuck cells:
			$\ve{\phi}$
	\end{itemize}}
	
	$w(x) = {w}_0 +\dots+ {w}_{n-2} x^{n-2} \leftarrow \mesp \cdot g(x)$
	
	Select $v \in \mathbb{F}_q\setminus \{ w_i \mid i \in \ve{\phi}\}$.
	
	$ c(x) = w(x) -v \cdot g_0(x) \mod (x^n-1)$
	
	\KwOut{
		Codeword $c(x) \in \mathbb{F}_q[x]$ of degree $\leq n-1$}	
\end{algorithm}
\begin{algorithm}\label{a4}
	\caption{Decoding}	
	\KwIn{Retrieve $y(x) = c(x) + e (x)$, where $e (x)\in \mathbb{F}_q[x]$ of degree $\leq n-1$ is the error polynomial }
	
	$\hat{c}(x) \gets$ Decode $y(x)$ in the code generated by $g(x)$
	
	$\hat{m}(x) \gets \hat{c}(x) \mod g_0(x)$
	
	\KwOut{Message $\hat{\mes}(x) \in \mathbb{F}_q[x]$ of degree $< k-1$ }
\end{algorithm}	
\begin{proof}
	A cyclic code of length $n$ contains the all-one word if and only if its generator polynomial $g(x)$ divides $g_0(x) = 1+x+\dots+x^{n-1}$. Thus, Construction~\ref{cons1} follows directly from Theorem~\ref{thm1}, but with different encoding and decoding algorithms.
	Algorithm~\ref{a3} shows the encoding process for the cyclic code construction. 
	Step~1 in Algorithm~\ref{a3} calculates $w(x)$ of degree $< n-1$.
	Since $u < q$, there is at least one $v \in \mathbb{F}_q$ such that 
	all coefficients of $w(x)$, $w_i \in \mathbb{F}_q$, are unequal to $v$. 
	Therefore, after Step~3, $c_{n-1} = -v$.
	The requirement for masking, see~\eqref{equ} is satisfied for $c(x)$ since ${c}_i = (w_i - v) \in \mathbb{F}_q \neq 0$. 
	
	Algorithm~\ref{a4} decodes the retrieved polynomial $y(x)$.
	First, decode $y(x)$ 
		in the code generated by $g(x)$. 
	Second, the algorithm performs the unmasking process to find $\hat{\mes}(x)$. 
	We obtain:
	\begin{align*}
		\hat c(x) &= \hat{\mes}(x)\cdot g(x) + z_0 \cdot g_0(x) \\
		\hat{\mes}(x) &= \dfrac{\hat{w}(x) \mod g_0(x)}{g(x)} = m(x).
		\tag*{\qedhere}
	\end{align*}
\end{proof}
Construction~\ref{cons1} provides an explicit cyclic construction that can mask $u < q$ cells and correct $t$ errors. 
If we use a BCH code in Construction~\ref{cons1}, we can bound the minimum distance of the code $\mycode{C}$ by the BCH bound. 
This is done in Tables~\ref{table2} and ~\ref{table3}.
\subsection{An alternative Proof of Proposition~\ref{modified_version_of_Theorem3}}\label{alternative_Proof_of_Proposition_2}
We start with Lemma~\ref{lem:basic0}.
\begin{lem}\label{lem:basic0}
	Let $\ve{M}\in\Fq^{m\times n}$ be such that each column of $\ve{M}$ has at least one non-zero entry. 
	Let $\ve{s}\in\Fq^n$.
	For each $\ve{w}\in \F_q^n$, there is a $\ve{v}\in\Fq^m$ such that 
	\[ \Big| \big\{i\in [n]\mid w_i+ (\ve{v}\ve{M})_i \geq s_i\big\}\Big| \;\; \geq n - \left\lfloor \frac{1}{q}\sum_{i=0}^{n-1} s_i \right\rfloor. \]
\end{lem}
\begin{proof}
		We define the set $S$ as
	\[ S = \big\{(i,\ve{v})\in [n]\times \Fq^m\mid w_i+(\ve{v}\ve{M})_i\geq s_i\big\} .\]
	Clearly, there is $\ve{v}\in\Fq^m$ such that
	\begin{equation}\label{eq:avg0}
		\Big| \big\{ i\in[n] \; \big| \; w_i+ (\ve{v}\ve{M})_i\geq s_i\big\}\Big| \;\; \geq \left\lceil\frac{|S|}{q^m}\right\rceil. 
	\end{equation}
	Let $i\in [n]$. As the $i$-th column of $\ve{M}$ has a non-zero entry, for each $y\in\Fq$ there
	are exactly $q^{m-1}$ vectors $\ve{x}\in\Fq^m$ such that $(\ve{x}\ve{M})_i=y$. As a consequence, 
	\[ \Big|\big\{ \ve{v}\in \Fq^m \mid w_i+(\ve{v}\ve{M})_i \geq s_i\big\} \Big| \;\; = (q - s_i)q^{m-1} , \]
	and so
	\begin{equation}\label{eq:sum_exact0}
		|S|= \sum_{i=0}^{n-1} (q-s_i)q^{m-1} = nq^m - q^{m-1}\sum_{i=0}^{n-1} s_i.
	\end{equation}
	The lemma follows from combining (\ref{eq:avg0}) and (\ref{eq:sum_exact0}).
\end{proof}
We are now in a position to introduce an alternative \emph{non-constructive} proof for Proposition~\ref{modified_version_of_Theorem3}.

	Let $\ve{s}\in\Sigma$. In order to simplify notation, we assume without loss of generality that 
	$\sum_{i=d_0-2}^{n-1} s_i \leq q-1$. 
	Let $\ve{w}\in\Fq^n$. We wish to find $\ve{z}\in\Fq^l$ such that $w_i + (\ve{z}\H_0)_i\geq s_i$ for many indices $i$.\\ 
	As the $d_0-2$ leftmost columns of $\H_0$ are independent, there exists an invertible matrix $\ve{T}\in\Fq^{l\times l}$ such that
	\[ \ve{T}\H_0 = \begin{bmatrix} I_{d_0-2} & \A \\ \0 & \B \end{bmatrix}, \]
	where $I_{d_0-2}$ denotes the identity matrix of size $d_0-2$. \\
	For $i\in [d_0-2]$, we choose $\zeta_i=s_i-w_i$ and write
	\[ \ve{v}=\ve{w} + \ve{\zeta}\cdot (I_{d_0-2} \mid \A). \]
	By definition, $v_i=s_i$ for all $i\in[d_0-2]$. \\
	As any $d_0-1$ columns of $\ve{T}\H_0$ are independent, no column of $\B$ consists of only zeroes.
	Lemma~\ref{lem:basic0} implies that there is an $\ve{\eta}\in\Fq^{l-d_0+2}$ such that
	\begin{align*}
		&\Big| \; \big\{ i\in[d_0-2,n-1]\; \big| \; w_i + (\ve{\eta} \B)_{i+d_0-2} \geq s_i\big\} \Big| \;\;\geq\\
		& n-d_0+2- \left\lfloor\frac{1}{q} \sum_{i=d_0-2}^{n-1}s_i\right\rfloor.
	\end{align*}
	Combining this with the fact that $v_i=s_i$ for all $i\in [d_0-2]$, we infer that for all indices $i\in[n]$, 
	\begin{align*}
		w_i+ \left( (\ve{\zeta},\ve{\eta}) \ve{T}\H_0\right)_i \geq s_i. \hspace{4 cm} \qed
	\end{align*}
\subsection{Proof of Lemma~\ref{lm:bind}}
	Let $s\geq 1$. As $\mathbb{F}_q\subseteq \mathbb{F}_{q^s}$, it is clear that $d_1\geq d_s$. 

To show the converse, we use the trace function $T$ defined as $T(x) = \sum_{i=0}^{s-1} x^{q^i}$.
As is well-known, $T$ is a non-trivial mapping from $\mathbb{F}_{q^s}$ to $\mathbb{F}_q$, and 
\begin{equation} \label{eq:tracelinear}
	T(ax+by) = aT(x)+bT(y),
\end{equation}
for all $x,y\in \mathbb{F}_{q^s}$ and $a,b\in \mathbb{F}_q$.
We extend the trace function to vectors by applying it coordinate-wise. 

Let $\vmes\in\mathbb{F}_{q^s}^k\setminus \{\ve{0}\}$. We choose $\lambda\in\mathbb{F}_{q^s}$ such that 
$T(\lambda\cdot \vmes)\neq \ve{0}$. As $T(0)=0$, we infer that 
$ \wt(\vmes\G) = \wt(\lambda\cdot \vmes \G) \geq \wt(T(\lambda \cdot \vmes\G))$.
As all entries from $\G$\ are in $\mathbb{F}_q$, it follows from (\ref{eq:tracelinear}) that 
$T(\lambda\cdot \vmes\G) = T(\lambda \cdot \vmes)\G$. As a consequence,
\[ \wt(\vmes\G) \geq \wt(T(\lambda\cdot \vmes\G)) = \wt(T(\lambda\cdot \vmes)\G) \geq d_1, \]
\subsection{Proofs of analytical comparisons of the asymptotic GV bounds}\label{app:asGVcomp}
We will use the following lemma.
\begin{lem}\label{lem:entrop}
	Let $q\geq 2$ be an integer.
	If $0\leq x,y$ are such that $x+y\leq 1$, then 
	\[ h_q(x+y) \leq h_q(x) + h_q(y) . \]
\end{lem}
\begin{proof}
	Let $0\leq y < 1$, and consider the function $f_y(x)=h_q(x+y)-h_q(x)-h_q(y)$ on the interval $[0,1-y]$.
	Clearly, $f'_y(x)=h'_q(x+y)-h'_q(x) \leq 0$, where the inequality follows from the fact
	that the second derivative of $h_q$ is non-negative. Hence,
	$f_y(x)\leq f_y(0)=0$ for each $x\in[0,1-y]$.
\end{proof}
\subsubsection{Proof of Proposition~\ref{upartialvsaGVnew}}
Assume $\tau$ and $\upsilon$ are such that the conditions of Theorem~\ref{upartial_thm} and Theorem~\ref{Th:aGVnew} are
satisfied, that is, 
such that $2\tau+2\frac{\upsilon}{q} < 1-\frac{1}{q}$ and $h_q(\upsilon)+h_q(2\tau) < 1$. 
By invoking Lemma~\ref{lem:entrop}, we see that 
\[ h_q(2\tau+2\frac{v}{q}) \leq h_q(2\tau) + h_q(2\cdot \frac{\upsilon}{q}) \leq h_q(2\tau) + h_q(\upsilon), \]
where the final inequality holds as $q\geq 2$ and $h_q$ is monotonically increasing on $[0,1-\frac{1}{q}]$.
As a consequence, the code rate guaranteed by Theorem~\ref{upartial_thm}
is at least equal to the code rate guaranteed by Theorem~\ref{Th:aGVnew}.
\subsubsection{Proof of Proposition~\ref{aGVnewvsbGV-s}}
Assume that the conditions of Theorem~\ref{Th:aGVnew} and Theorem~\ref{th:bGV-s} are satisfied.
The difference between the rate of Theorem~\ref{th:bGV-s} and of Theorem~\ref{Th:aGVnew} equals
\begin{equation}\label{eq:diffaGVbGV} 
	h_{2^{\pow}}(\upsilon) - \frac{1}{\pow}h_2(\frac{\upsilon}{2^{\pow-1}}). 
\end{equation}
According to the conditions of Theorem~\ref{Th:aGVnew}, $\upsilon\leq 1-\frac{1}{2^{\pow}}$, and so
$h_{2^{\pow}}(\upsilon) \geq h_2(\frac{\upsilon}{2^{\pow-1}})$.
As $h_{2^{\pow}}(x)=\frac{1}{\pow} h_2(x) + x\log_{2^{\pow}}(2^{\pow}-1)$, 
the difference in (\ref{eq:diffaGVbGV}) is non-negative. 
That is, Theorem~\ref{th:bGV-s} is better than Theorem~\ref{Th:aGVnew}.
\subsubsection{Comparing Theorem~\ref{upartial_thm} and Theorem~\ref{th:bGV-s}}
Assume that $\tau$ and $\upsilon$ are such that the conditions of Theorem~\ref{upartial_thm} and of
Theorem~\ref{th:bGV-s} are satisfied, that is, $2\tau + 2\frac{\upsilon}{2^{\pow}}< 1-\frac{1}{2^{\pow}}$,
\[ 0\leq \upsilon \leq 2^{\pow-2},\; 0\leq 2\tau \leq \frac{1}{2} \mbox{ and } h_2(\frac{\upsilon}{2^{\pow-1}}) + 
h_2(2\tau) < 1 . \]
Let $f_{\pow}(\tau,\upsilon_0)$, where $\upsilon_0=\frac{\upsilon}{2^{\pow-1}}$, be the bound from Theorem~\ref{upartial_thm} 
minus the bound from Theorem~\ref{th:bGV-s}, that is
\[ f_{\pow}(\tau,\upsilon_0)= h_{2^{\pow}}(2\tau) + \frac{1}{\pow}h_2(\upsilon_0)-h_{2^{\pow}}(2\tau + \upsilon_0) .\]
The definition of the entropy function implies that for any $x\in[0,1]$
\begin{equation}\label{eq:entprop}
	h_{2^{\pow}}(x) = \frac{1}{\pow} \left( h_2(x) + x \log_2(2^{\pow}-1)\right. ).
\end{equation}
Applying (\ref{eq:entprop}) , we infer that
\begin{equation}\label{eq:falt}
	\pow f_{\pow}(\tau,\upsilon_0) = h_2(2\tau) + h_2(\upsilon_0) -h_2(2\tau + \upsilon_0) - \upsilon_0\log_2(2^{\pow}-1). 
\end{equation}
In particular, 
 $\pow f_{\pow}(0,\upsilon_0) = -\upsilon_0\log_2(2^{\pow}-1) \leq 0$.
 
So for $\tau=0$, Theorem~\ref{th:bGV-s} is better than Theorem~\ref{upartial_thm}. 
It follows from Lemma~\ref{lem:entrop} that the three leftmost terms in (\ref{eq:falt}) form a non-negative number.
The subtraction of the fourth term, however, can result in a negative function value, especially for large $\pow$.
\begin{examplex}[Numerical example]
	$\mu f_{\mu}(0.055,0.11) = 2\cdot h_2(0.11)- h_2(0.22) -0.11\log_2(2^{\pow}-1)\approx 
	0.23397-0.11 \log_2(2^{\pow}-1)$ is positive for $\mu\leq 2$ and negative otherwise. 
	
	\qquad \qquad\QEDA
\end{examplex}
We now prove Proposition~\ref{upartialvsbGV-s}. That is, we show that for $\tau>0$ and $\upsilon_0$ sufficiently small, $f_{\mu}(\tau,\upsilon_0) > 0$.
This follows from the Taylor expansion
of $\mu f_{\mu}(\tau,\upsilon_0)$ around $\upsilon_0=0$. Indeed, $f_{\mu}(\tau,0)=0$, and $h'_2(x)\rightarrow \infty$
if $x\downarrow 0$.

\end{document}

%% file: defs.tex
\usepackage{pgfplots}
\pgfplotsset{compat=newest}
\usepackage{tikz}
\usetikzlibrary{arrows,matrix,positioning,patterns}
\usepackage[utf8]{inputenc}
\usepackage[innermargin=0.545in,outermargin=0.65in,top=0.545in,bottom=.7in]{geometry}
\usepackage[english]{babel}
\usepackage[T1]{fontenc}
\usepackage{epsfig}
\usepackage{amsmath, amssymb, amsbsy}
\usepackage{mathdots}
\usepackage{xspace}
\usepackage[noend]{algpseudocode}
\usepackage[linesnumbered,ruled,vlined,titlenumbered]{algorithm2e}
\usepackage{algorithmicx}
\usepackage{color}
\usepackage{cite}
\usepackage{booktabs}
\usepackage{verbatim}
\usepackage{url}
\usepackage{lipsum}
\usepackage{enumitem}
\usepackage{colortbl}
\usepackage{amsthm}
\usepackage{dblfloatfix}
\usepackage{verbatim}
\makeatletter
\newcommand\footnoteref[1]{\protected@xdef\@thefnmark{\ref{#1}}\@footnotemark}
\makeatother

\usepackage{multicol}

\usepackage[final,tracking=true,kerning=true,spacing=true,factor=1100,stretch=10,shrink=20]{microtype}

\newenvironment{mymatrix}{\begin{bmatrix}} {\end{bmatrix} }

\def\ve#1{{\mathchoice{\mbox{\boldmath$\displaystyle #1$}}%
              {\mbox{\boldmath$\textstyle #1$}}%
              {\mbox{\boldmath$\scriptstyle #1$}}%
              {\mbox{\boldmath$\scriptscriptstyle #1$}}}}

\renewcommand{\vec}[1]{\ensuremath{\boldsymbol{#1}}}

\newcommand{\Fq}{\ensuremath{\mathbb{F}_q}}

\DeclareMathOperator{\wt}{wt}
\DeclareMathOperator{\Vol}{Vol}

\newcommand{\mycode}[1]{\ensuremath{\mathcal{#1}}}

\newcommand{\0}{\ve{0}}

\renewcommand{\H}{\ve{H}}
\newcommand{\y}{\ve{y}}

\renewcommand{\c}{\ve{c}}

\newcommand{\G}{\ve{G}}

\newcommand{\Y}{\ve{Y}}

\newcommand{\A}{\ve{A}}
\newcommand{\B}{\ve{B}}

\newcommand{\F}{\mathbb{F}}

%% file: figures/SPq3n121t3/k_cases_Only_PSMC.tex
\addplot 
table {
0.000000 121.000000
1.000000 120.630930
2.000000 120.261860
3.000000 119.892789
4.000000 119.523719
5.000000 119.154649
6.000000 118.785579
7.000000 118.416508
8.000000 118.047438
9.000000 117.678368
10.000000 117.309298
11.000000 116.940227
12.000000 116.571157
13.000000 116.202087
14.000000 115.833017
15.000000 115.463946
16.000000 115.094876
17.000000 114.725806
18.000000 114.356736
19.000000 113.987665
20.000000 113.618595
21.000000 113.249525
22.000000 112.880455
23.000000 112.511384
24.000000 112.142314
25.000000 111.773244
26.000000 111.404174
27.000000 111.035103
28.000000 110.666033
29.000000 110.296963
30.000000 109.927893
31.000000 109.558822
32.000000 109.189752
33.000000 108.820682
34.000000 108.451612
35.000000 108.082541
36.000000 107.713471
37.000000 107.344401
38.000000 106.975331
39.000000 106.606260
40.000000 106.237190
41.000000 105.868120
42.000000 105.499050
43.000000 105.129979
44.000000 104.760909
45.000000 104.391839
46.000000 104.022769
47.000000 103.653698
48.000000 103.284628
49.000000 102.915558
50.000000 102.546488
51.000000 102.177417
52.000000 101.808347
53.000000 101.439277
54.000000 101.070207
55.000000 100.701136
56.000000 100.332066
57.000000 99.962996
58.000000 99.593926
59.000000 99.224855
60.000000 98.855785
61.000000 98.486715
62.000000 98.117645
63.000000 97.748574
64.000000 97.379504
65.000000 97.010434
66.000000 96.641364
67.000000 96.272293
68.000000 95.903223
69.000000 95.534153
70.000000 95.165083
71.000000 94.796013
72.000000 94.426942
73.000000 94.057872
74.000000 93.688802
75.000000 93.319732
76.000000 92.950661
77.000000 92.581591
78.000000 92.212521
79.000000 91.843451
80.000000 91.474380
81.000000 91.105310
82.000000 90.736240
83.000000 90.367170
84.000000 89.998099
85.000000 89.629029
86.000000 89.259959
87.000000 88.890889
88.000000 88.521818
89.000000 88.152748
90.000000 87.783678
91.000000 87.414608
92.000000 87.045537
93.000000 86.676467
94.000000 86.307397
95.000000 85.938327
96.000000 85.569256
97.000000 85.200186
98.000000 84.831116
99.000000 84.462046
100.000000 84.092975
101.000000 83.723905
102.000000 83.354835
103.000000 82.985765
104.000000 82.616694
105.000000 82.247624
106.000000 81.878554
107.000000 81.509484
108.000000 81.140413
109.000000 80.771343
110.000000 80.402273
111.000000 80.033203
112.000000 79.664132
113.000000 79.295062
114.000000 78.925992
115.000000 78.556922
116.000000 78.187851
117.000000 77.818781
118.000000 77.449711
119.000000 77.080641
120.000000 76.711570
121.000000 76.342500
};

%% file: figures/SPq3n121t3/k_cases_overlapping.tex
\addplot 
table {
0.000000 107.653419
1.000000 107.295655
2.000000 106.937937
3.000000 106.580267
4.000000 106.222644
5.000000 105.865069
6.000000 105.507542
7.000000 105.150063
8.000000 104.792634
9.000000 104.435254
10.000000 104.077923
11.000000 103.720643
12.000000 103.363414
13.000000 103.006235
14.000000 102.649109
15.000000 102.292034
16.000000 101.935011
17.000000 101.578042
18.000000 101.221125
19.000000 100.864263
20.000000 100.507455
21.000000 100.150702
22.000000 99.794003
23.000000 99.437361
24.000000 99.080775
25.000000 98.724245
26.000000 98.367773
27.000000 98.011359
28.000000 97.655003
29.000000 97.298705
30.000000 96.942467
31.000000 96.586289
32.000000 96.230172
33.000000 95.874115
34.000000 95.518120
35.000000 95.162188
36.000000 94.806318
37.000000 94.450511
38.000000 94.094769
39.000000 93.739091
40.000000 93.383478
41.000000 93.027931
42.000000 92.672451
43.000000 92.317038
44.000000 91.961692
45.000000 91.606415
46.000000 91.251207
47.000000 90.896069
48.000000 90.541001
49.000000 90.186005
50.000000 89.831081
51.000000 89.476229
52.000000 89.121451
53.000000 88.766747
54.000000 88.412117
55.000000 88.057564
56.000000 87.703087
57.000000 87.348687
58.000000 86.994365
59.000000 86.640123
60.000000 86.285960
61.000000 85.931877
62.000000 85.577877
63.000000 85.223958
64.000000 84.870123
65.000000 84.516372
66.000000 84.162705
67.000000 83.809125
68.000000 83.455632
69.000000 83.102227
70.000000 82.748911
71.000000 82.395684
72.000000 82.042548
73.000000 81.689505
74.000000 81.336554
75.000000 80.983698
76.000000 80.630936
77.000000 80.278271
78.000000 79.925703
79.000000 79.573233
80.000000 79.220864
81.000000 78.868595
82.000000 78.516428
83.000000 78.164364
84.000000 77.812405
85.000000 77.460552
86.000000 77.108805
87.000000 76.757167
88.000000 76.405639
89.000000 76.054222
90.000000 75.702917
91.000000 75.351726
92.000000 75.000650
93.000000 74.649692
94.000000 74.298851
95.000000 73.948130
96.000000 73.597530
97.000000 73.247054
98.000000 72.896701
99.000000 72.546475
100.000000 72.196377
101.000000 71.846408
102.000000 71.496570
103.000000 71.146866
104.000000 70.797296
105.000000 70.447863
106.000000 70.098568
107.000000 69.749414
108.000000 69.400402
109.000000 69.051534
110.000000 68.702813
111.000000 68.354240
112.000000 68.005818
113.000000 67.657548
114.000000 67.309434
115.000000 66.961477
116.000000 66.613679
117.000000 66.266043
118.000000 65.918571
119.000000 65.571266
120.000000 65.224130
121.000000 64.877166
};

%% file: figures/SPq3n121t25/k_cases_overlapping.tex
\addplot 
table {
0.000000 51.193083
1.000000 50.922655
2.000000 50.652609
3.000000 50.382947
4.000000 50.113671
5.000000 49.844784
6.000000 49.576289
7.000000 49.308188
8.000000 49.040484
9.000000 48.773180
10.000000 48.506278
11.000000 48.239781
12.000000 47.973692
13.000000 47.708014
14.000000 47.442749
15.000000 47.177900
16.000000 46.913470
17.000000 46.649462
18.000000 46.385879
19.000000 46.122724
20.000000 45.860000
21.000000 45.597709
22.000000 45.335854
23.000000 45.074439
24.000000 44.813467
25.000000 44.552941
26.000000 44.292863
27.000000 44.033237
28.000000 43.774065
29.000000 43.515352
30.000000 43.257100
31.000000 42.999313
32.000000 42.741993
33.000000 42.485144
34.000000 42.228769
35.000000 41.972871
36.000000 41.717454
37.000000 41.462522
38.000000 41.208076
39.000000 40.954121
40.000000 40.700661
41.000000 40.447698
42.000000 40.195236
43.000000 39.943279
44.000000 39.691829
45.000000 39.440892
46.000000 39.190469
47.000000 38.940565
48.000000 38.691183
49.000000 38.442327
50.000000 38.194001
51.000000 37.946208
52.000000 37.698952
53.000000 37.452236
54.000000 37.206065
55.000000 36.960442
56.000000 36.715371
57.000000 36.470855
58.000000 36.226899
59.000000 35.983507
60.000000 35.740682
61.000000 35.498428
62.000000 35.256749
63.000000 35.015649
64.000000 34.775133
65.000000 34.535203
66.000000 34.295865
67.000000 34.057121
68.000000 33.818977
69.000000 33.581436
70.000000 33.344503
71.000000 33.108181
72.000000 32.872474
73.000000 32.637387
74.000000 32.402925
75.000000 32.169090
76.000000 31.935888
77.000000 31.703323
78.000000 31.471398
79.000000 31.240119
80.000000 31.009488
81.000000 30.779512
82.000000 30.550194
83.000000 30.321538
84.000000 30.093549
85.000000 29.866231
86.000000 29.639589
87.000000 29.413626
88.000000 29.188348
89.000000 28.963759
90.000000 28.739863
91.000000 28.516664
92.000000 28.294167
93.000000 28.072377
94.000000 27.851298
95.000000 27.630935
96.000000 27.411291
97.000000 27.192371
98.000000 26.974181
99.000000 26.756724
100.000000 26.540005
101.000000 26.324029
102.000000 26.108799
103.000000 25.894321
104.000000 25.680598
105.000000 25.467636
106.000000 25.255440
107.000000 25.044012
108.000000 24.833359
109.000000 24.623483
110.000000 24.414391
111.000000 24.206087
112.000000 23.998574
113.000000 23.791857
114.000000 23.585941
115.000000 23.380831
116.000000 23.176530
117.000000 22.973043
118.000000 22.770374
119.000000 22.568528
120.000000 22.367510
121.000000 22.167323
};

%% file: figures/comparisons/BCH_match_const1_GV_bound/plot_list_ball_blokhuis_bound.tex
\addplot table {
48.000000 0.008772
48.000000 0.017544
47.000000 0.026316
47.000000 0.035088
47.000000 0.043860
46.000000 0.052632
46.000000 0.061404
46.000000 0.070175
45.000000 0.078947
44.000000 0.087719
44.000000 0.096491
44.000000 0.105263
43.000000 0.114035
43.000000 0.122807
42.000000 0.131579
42.000000 0.140351
41.000000 0.149123
41.000000 0.157895
41.000000 0.166667
41.000000 0.175439
40.000000 0.184211
40.000000 0.192982
39.000000 0.201754
39.000000 0.210526
38.000000 0.219298
38.000000 0.228070
38.000000 0.236842
37.000000 0.245614
37.000000 0.254386
36.000000 0.263158
35.000000 0.271930
35.000000 0.280702
35.000000 0.289474
35.000000 0.298246
34.000000 0.307018
34.000000 0.315789
33.000000 0.324561
32.000000 0.333333
32.000000 0.342105
32.000000 0.350877
32.000000 0.359649
32.000000 0.368421
31.000000 0.377193
31.000000 0.385965
30.000000 0.394737
29.000000 0.403509
29.000000 0.412281
29.000000 0.421053
28.000000 0.429825
28.000000 0.438596
27.000000 0.447368
26.000000 0.456140
26.000000 0.464912
26.000000 0.473684
25.000000 0.482456
25.000000 0.491228
25.000000 0.500000
24.000000 0.508772
23.000000 0.517544
23.000000 0.526316
23.000000 0.535088
22.000000 0.543860
22.000000 0.552632
21.000000 0.561404
21.000000 0.570175
20.000000 0.578947
20.000000 0.587719
20.000000 0.596491
20.000000 0.605263
19.000000 0.614035
19.000000 0.622807
18.000000 0.631579
18.000000 0.640351
17.000000 0.649123
17.000000 0.657895
17.000000 0.666667
16.000000 0.675439
16.000000 0.684211
15.000000 0.692982
14.000000 0.701754
14.000000 0.710526
14.000000 0.719298
14.000000 0.728070
13.000000 0.736842
13.000000 0.745614
12.000000 0.754386
11.000000 0.763158
11.000000 0.771930
11.000000 0.780702
11.000000 0.789474
10.000000 0.798246
10.000000 0.807018
9.000000 0.815789
8.000000 0.824561
8.000000 0.833333
8.000000 0.842105
7.000000 0.850877
7.000000 0.859649
7.000000 0.868421
6.000000 0.877193
5.000000 0.885965
5.000000 0.894737
5.000000 0.903509
4.000000 0.912281
4.000000 0.921053
4.000000 0.929825
3.000000 0.938596
2.000000 0.947368
2.000000 0.956140
2.000000 0.964912
1.000000 0.973684
1.000000 0.982456
0.000000 0.991228
0.000000 1.000000
};
\addlegendentry{Ball--Blokhuis bound \cite{ball2013bound}}

%% file: figures/comparisons/BCH_match_const1_GV_bound/plot_list_griesmer_bound.tex
\addplot table {
56.000000 0.008772
49.000000 0.017544
48.000000 0.026316
48.000000 0.035088
47.000000 0.043860
47.000000 0.052632
46.000000 0.061404
46.000000 0.070175
45.000000 0.078947
45.000000 0.087719
45.000000 0.096491
44.000000 0.105263
44.000000 0.114035
43.000000 0.122807
43.000000 0.131579
42.000000 0.140351
42.000000 0.149123
41.000000 0.157895
41.000000 0.166667
41.000000 0.175439
40.000000 0.184211
40.000000 0.192982
39.000000 0.201754
39.000000 0.210526
38.000000 0.219298
38.000000 0.228070
38.000000 0.236842
37.000000 0.245614
37.000000 0.254386
36.000000 0.263158
36.000000 0.271930
35.000000 0.280702
35.000000 0.289474
34.000000 0.298246
34.000000 0.307018
34.000000 0.315789
33.000000 0.324561
33.000000 0.333333
32.000000 0.342105
32.000000 0.350877
31.000000 0.359649
31.000000 0.368421
31.000000 0.377193
30.000000 0.385965
30.000000 0.394737
29.000000 0.403509
29.000000 0.412281
28.000000 0.421053
28.000000 0.429825
27.000000 0.438596
27.000000 0.447368
27.000000 0.456140
26.000000 0.464912
26.000000 0.473684
25.000000 0.482456
25.000000 0.491228
24.000000 0.500000
24.000000 0.508772
24.000000 0.517544
24.000000 0.526316
23.000000 0.535088
23.000000 0.543860
22.000000 0.552632
22.000000 0.561404
21.000000 0.570175
21.000000 0.578947
20.000000 0.587719
20.000000 0.596491
20.000000 0.605263
19.000000 0.614035
19.000000 0.622807
18.000000 0.631579
18.000000 0.640351
17.000000 0.649123
17.000000 0.657895
17.000000 0.666667
16.000000 0.675439
16.000000 0.684211
15.000000 0.692982
15.000000 0.701754
14.000000 0.710526
14.000000 0.719298
13.000000 0.728070
13.000000 0.736842
13.000000 0.745614
12.000000 0.754386
12.000000 0.763158
11.000000 0.771930
11.000000 0.780702
10.000000 0.789474
10.000000 0.798246
10.000000 0.807018
9.000000 0.815789
9.000000 0.824561
8.000000 0.833333
8.000000 0.842105
7.000000 0.850877
7.000000 0.859649
6.000000 0.868421
6.000000 0.877193
6.000000 0.885965
5.000000 0.894737
5.000000 0.903509
4.000000 0.912281
4.000000 0.921053
3.000000 0.929825
3.000000 0.938596
3.000000 0.947368
2.000000 0.956140
2.000000 0.964912
1.000000 0.973684
1.000000 0.982456
0.000000 0.991228
0.000000 1.000000
};
\addlegendentry{Griesmer bound \cite{griesmer1960bound}}

%% file: figures/comparisons/BCH_match_const1_GV_bound/q_ary_usual_GV_n114.tex
\addplot table {
0.000000 0.991228
1.000000 0.938596
2.000000 0.894737
3.000000 0.850877
4.000000 0.807018
5.000000 0.771930
6.000000 0.736842
7.000000 0.701754
8.000000 0.666667
9.000000 0.640351
10.000000 0.605263
11.000000 0.578947
12.000000 0.543860
13.000000 0.517544
14.000000 0.491228
15.000000 0.464912
16.000000 0.447368
17.000000 0.421053
18.000000 0.394737
19.000000 0.368421
20.000000 0.350877
21.000000 0.333333
22.000000 0.307018
23.000000 0.289474
24.000000 0.271930
25.000000 0.254386
26.000000 0.228070
27.000000 0.219298
28.000000 0.201754
29.000000 0.184211
30.000000 0.166667
31.000000 0.149123
32.000000 0.140351
33.000000 0.122807
34.000000 0.114035
35.000000 0.096491
36.000000 0.087719
37.000000 0.078947
38.000000 0.061404
39.000000 0.052632
40.000000 0.043860
41.000000 0.035088
42.000000 0.026316
43.000000 0.026316
44.000000 0.017544
45.000000 0.008772
46.000000 0.008772
47.000000 0.000000
48.000000 0.000000
49.000000 0.000000
50.000000 0.000000
51.000000 0.000000
52.000000 0.000000
53.000000 0.000000
54.000000 0.000000
55.000000 0.000000
56.000000 0.000000
57.000000 -0.100000
};
\addlegendentry{$q$-ary GV ($u=0$)}

%% file: figures/comparisons/BCH_match_const1_GV_bound/plot_list_construction_1_n114.tex
\addplot table {
0.000000 0.982456
1.000000 0.929825
2.000000 0.885965
3.000000 0.842105
4.000000 0.798246
5.000000 0.763158
6.000000 0.728070
7.000000 0.692982
8.000000 0.657895
9.000000 0.631579
10.000000 0.596491
11.000000 0.570175
12.000000 0.535088
13.000000 0.508772
14.000000 0.482456
15.000000 0.456140
16.000000 0.438596
17.000000 0.412281
18.000000 0.385965
19.000000 0.359649
20.000000 0.342105
21.000000 0.324561
22.000000 0.298246
23.000000 0.280702
24.000000 0.263158
25.000000 0.245614
26.000000 0.219298
27.000000 0.210526
28.000000 0.192982
29.000000 0.175439
30.000000 0.157895
31.000000 0.140351
32.000000 0.131579
33.000000 0.114035
34.000000 0.105263
35.000000 0.087719
36.000000 0.078947
37.000000 0.070175
38.000000 0.052632
39.000000 0.043860
40.000000 0.035088
41.000000 0.026316
42.000000 0.017544
43.000000 0.017544
44.000000 0.008772
45.000000 0.000000
46.000000 0.000000
47.000000 -0.008772
48.000000 -0.008772
49.000000 -0.008772
50.000000 -0.008772
51.000000 -0.008772
52.000000 -0.008772
53.000000 -0.008772
54.000000 -0.008772
55.000000 -0.008772
56.000000 -0.008772
57.000000 -0.100000
};
\addlegendentry{Theorem~\ref{th:upartial}, $u \leq 6$}

%% file: figures/comparisons/BCH_match_const1_GV_bound/ordGVq7_n114.tex
\addplot table {
0.000000 0.912705
1.000000 0.864243
2.000000 0.815781
3.000000 0.775396
4.000000 0.735010
5.000000 0.702702
6.000000 0.662317
7.000000 0.630009
8.000000 0.605778
9.000000 0.573470
10.000000 0.541161
11.000000 0.516930
12.000000 0.492699
13.000000 0.460391
14.000000 0.436160
15.000000 0.411929
16.000000 0.387698
17.000000 0.371544
18.000000 0.347313
19.000000 0.323081
20.000000 0.306927
21.000000 0.282696
22.000000 0.266542
23.000000 0.250388
24.000000 0.226157
25.000000 0.210003
26.000000 0.193849
27.000000 0.177695
28.000000 0.161541
29.000000 0.153464
30.000000 0.137310
31.000000 0.121156
32.000000 0.113079
33.000000 0.096924
34.000000 0.088847
35.000000 0.072693
36.000000 0.064616
37.000000 0.056539
38.000000 0.048462
39.000000 0.040385
40.000000 0.032308
41.000000 0.024231
42.000000 0.016154
43.000000 0.016154
44.000000 0.008077
45.000000 0.008077
46.000000 0.000000
47.000000 0.000000
48.000000 0.000000
49.000000 0.000000
50.000000 0.000000
51.000000 0.000000
52.000000 0.000000
53.000000 0.000000
54.000000 0.000000
55.000000 0.000000
56.000000 0.000000
57.000000 -0.100000
};
\addlegendentry{$(q-1)$-ary GV ($u$ arbitrary)}

%% file: figures/comparisons/BCH_match_const1_GV_bound/plot_different_k_for_BCH_Codes.tex
\addplot table {
0.000000 0.991228
0.000000 0.964912
1.000000 0.938596
1.000000 0.912281
2.000000 0.885965
2.000000 0.859649
3.000000 0.833333
3.000000 0.833333
4.000000 0.807018
4.000000 0.780702
5.000000 0.754386
5.000000 0.728070
6.000000 0.701754
6.000000 0.675439
7.000000 0.675439
7.000000 0.649123
8.000000 0.622807
8.000000 0.622807
9.000000 0.622807
9.000000 0.614035
10.000000 0.587719
10.000000 0.587719
11.000000 0.561404
11.000000 0.535088
12.000000 0.508772
12.000000 0.482456
13.000000 0.482456
13.000000 0.456140
14.000000 0.456140
14.000000 0.429825
15.000000 0.403509
15.000000 0.377193
16.000000 0.350877
16.000000 0.350877
17.000000 0.350877
17.000000 0.350877
18.000000 0.350877
18.000000 0.350877
19.000000 0.342105
19.000000 0.315789
20.000000 0.315789
20.000000 0.289474
21.000000 0.289474
21.000000 0.263158
22.000000 0.236842
22.000000 0.236842
23.000000 0.210526
23.000000 0.210526
24.000000 0.184211
24.000000 0.184211
25.000000 0.184211
25.000000 0.184211
26.000000 0.184211
26.000000 0.184211
27.000000 0.184211
27.000000 0.184211
28.000000 0.184211
28.000000 0.175439
29.000000 0.149123
29.000000 0.149123
30.000000 0.122807
30.000000 0.122807
31.000000 0.096491
31.000000 0.096491
32.000000 0.096491
32.000000 0.070175
33.000000 0.070175
33.000000 0.070175
34.000000 0.070175
34.000000 0.070175
35.000000 0.070175
35.000000 0.070175
36.000000 0.070175
36.000000 0.070175
37.000000 0.070175
37.000000 0.070175
38.000000 0.061404
38.000000 0.061404
39.000000 0.061404
39.000000 0.035088
40.000000 0.035088
40.000000 0.008772
41.000000 0.008772
41.000000 0.008772
42.000000 0.008772
42.000000 0.008772
43.000000 0.008772
43.000000 0.008772
44.000000 0.008772
44.000000 0.008772
45.000000 0.008772
45.000000 0.008772
46.000000 0.008772
46.000000 0.008772
47.000000 0.008772
47.000000 0.000000
48.000000 0.000000
48.000000 0.000000
49.000000 0.000000
49.000000 0.000000
50.000000 0.000000
50.000000 0.000000
51.000000 0.000000
51.000000 0.000000
52.000000 0.000000
52.000000 0.000000
53.000000 0.000000
53.000000 0.000000
54.000000 0.000000
54.000000 0.000000
55.000000 0.000000
55.000000 0.000000
56.000000 0.000000
56.000000 0.000000
};
\addlegendentry{BCH Codes }

%% file: figures/comparisons/GV_const_2_3_q8/plot_list_const_2_q8.tex
\addplot table {
0.000000 0.995000
1.000000 0.965000
2.000000 0.935000
3.000000 0.905000
4.000000 0.880000
5.000000 0.860000
6.000000 0.835000
7.000000 0.815000
8.000000 0.790000
9.000000 0.770000
10.000000 0.750000
11.000000 0.730000
12.000000 0.715000
13.000000 0.695000
14.000000 0.675000
15.000000 0.660000
16.000000 0.640000
17.000000 0.625000
18.000000 0.605000
19.000000 0.590000
20.000000 0.575000
21.000000 0.560000
22.000000 0.545000
23.000000 0.530000
24.000000 0.515000
25.000000 0.500000
26.000000 0.485000
27.000000 0.470000
28.000000 0.455000
29.000000 0.440000
30.000000 0.430000
31.000000 0.415000
32.000000 0.400000
33.000000 0.390000
34.000000 0.375000
35.000000 0.365000
36.000000 0.350000
37.000000 0.340000
38.000000 0.330000
39.000000 0.315000
40.000000 0.305000
41.000000 0.295000
42.000000 0.280000
43.000000 0.270000
44.000000 0.260000
45.000000 0.250000
46.000000 0.240000
47.000000 0.230000
48.000000 0.220000
49.000000 0.210000
50.000000 0.200000
51.000000 0.190000
52.000000 0.185000
53.000000 0.175000
54.000000 0.165000
55.000000 0.155000
56.000000 0.150000
57.000000 0.140000
58.000000 0.130000
59.000000 0.125000
60.000000 0.115000
61.000000 0.110000
62.000000 0.105000
63.000000 0.095000
64.000000 0.090000
65.000000 0.080000
66.000000 0.075000
67.000000 0.070000
68.000000 0.065000
69.000000 0.060000
70.000000 0.055000
71.000000 0.050000
72.000000 0.045000
73.000000 0.040000
74.000000 0.035000
75.000000 0.030000
76.000000 0.025000
77.000000 0.020000
78.000000 0.015000
79.000000 0.015000
80.000000 0.010000
};
\addlegendentry{$u= 0$}
\addplot table {
0.000000 0.930000
1.000000 0.900000
2.000000 0.870000
3.000000 0.840000
4.000000 0.815000
5.000000 0.795000
6.000000 0.770000
7.000000 0.750000
8.000000 0.725000
9.000000 0.705000
10.000000 0.685000
11.000000 0.665000
12.000000 0.650000
13.000000 0.630000
14.000000 0.610000
15.000000 0.595000
16.000000 0.575000
17.000000 0.560000
18.000000 0.540000
19.000000 0.525000
20.000000 0.510000
21.000000 0.495000
22.000000 0.480000
23.000000 0.465000
24.000000 0.450000
25.000000 0.435000
26.000000 0.420000
27.000000 0.405000
28.000000 0.390000
29.000000 0.375000
30.000000 0.365000
31.000000 0.350000
32.000000 0.335000
33.000000 0.325000
34.000000 0.310000
35.000000 0.300000
36.000000 0.285000
37.000000 0.275000
38.000000 0.265000
39.000000 0.250000
40.000000 0.240000
41.000000 0.230000
42.000000 0.215000
43.000000 0.205000
44.000000 0.195000
45.000000 0.185000
46.000000 0.175000
47.000000 0.165000
48.000000 0.155000
49.000000 0.145000
50.000000 0.135000
51.000000 0.125000
52.000000 0.120000
53.000000 0.110000
54.000000 0.100000
55.000000 0.090000
56.000000 0.085000
57.000000 0.075000
58.000000 0.065000
59.000000 0.060000
60.000000 0.050000
61.000000 0.045000
62.000000 0.040000
63.000000 0.030000
64.000000 0.025000
65.000000 0.015000
66.000000 0.010000
67.000000 0.005000
68.000000 0.000000
69.000000 -0.100000
70.000000 -0.100000
71.000000 -0.100000
72.000000 -0.100000
73.000000 -0.100000
74.000000 -0.100000
75.000000 -0.100000
76.000000 -0.100000
77.000000 -0.100000
78.000000 -0.100000
79.000000 -0.100000
80.000000 -0.100000
};
\addlegendentry{$u= 10$}
\addplot table {
0.000000 0.810000
1.000000 0.780000
2.000000 0.750000
3.000000 0.720000
4.000000 0.695000
5.000000 0.675000
6.000000 0.650000
7.000000 0.630000
8.000000 0.605000
9.000000 0.585000
10.000000 0.565000
11.000000 0.545000
12.000000 0.530000
13.000000 0.510000
14.000000 0.490000
15.000000 0.475000
16.000000 0.455000
17.000000 0.440000
18.000000 0.420000
19.000000 0.405000
20.000000 0.390000
21.000000 0.375000
22.000000 0.360000
23.000000 0.345000
24.000000 0.330000
25.000000 0.315000
26.000000 0.300000
27.000000 0.285000
28.000000 0.270000
29.000000 0.255000
30.000000 0.245000
31.000000 0.230000
32.000000 0.215000
33.000000 0.205000
34.000000 0.190000
35.000000 0.180000
36.000000 0.165000
37.000000 0.155000
38.000000 0.145000
39.000000 0.130000
40.000000 0.120000
41.000000 0.110000
42.000000 0.095000
43.000000 0.085000
44.000000 0.075000
45.000000 0.065000
46.000000 0.055000
47.000000 0.045000
48.000000 0.035000
49.000000 0.025000
50.000000 0.015000
51.000000 0.005000
52.000000 0.000000
53.000000 -0.100000
54.000000 -0.100000
55.000000 -0.100000
56.000000 -0.100000
57.000000 -0.100000
58.000000 -0.100000
59.000000 -0.100000
60.000000 -0.100000
61.000000 -0.100000
62.000000 -0.100000
63.000000 -0.100000
64.000000 -0.100000
65.000000 -0.100000
66.000000 -0.100000
67.000000 -0.100000
68.000000 -0.100000
69.000000 -0.100000
70.000000 -0.100000
71.000000 -0.100000
72.000000 -0.100000
73.000000 -0.100000
74.000000 -0.100000
75.000000 -0.100000
76.000000 -0.100000
77.000000 -0.100000
78.000000 -0.100000
79.000000 -0.100000
80.000000 -0.100000
};
\addlegendentry{$u= 20$}
\addplot table {
0.000000 0.710000
1.000000 0.680000
2.000000 0.650000
3.000000 0.620000
4.000000 0.595000
5.000000 0.575000
6.000000 0.550000
7.000000 0.530000
8.000000 0.505000
9.000000 0.485000
10.000000 0.465000
11.000000 0.445000
12.000000 0.430000
13.000000 0.410000
14.000000 0.390000
15.000000 0.375000
16.000000 0.355000
17.000000 0.340000
18.000000 0.320000
19.000000 0.305000
20.000000 0.290000
21.000000 0.275000
22.000000 0.260000
23.000000 0.245000
24.000000 0.230000
25.000000 0.215000
26.000000 0.200000
27.000000 0.185000
28.000000 0.170000
29.000000 0.155000
30.000000 0.145000
31.000000 0.130000
32.000000 0.115000
33.000000 0.105000
34.000000 0.090000
35.000000 0.080000
36.000000 0.065000
37.000000 0.055000
38.000000 0.045000
39.000000 0.030000
40.000000 0.020000
41.000000 0.010000
42.000000 -0.100000
43.000000 -0.100000
44.000000 -0.100000
45.000000 -0.100000
46.000000 -0.100000
47.000000 -0.100000
48.000000 -0.100000
49.000000 -0.100000
50.000000 -0.100000
51.000000 -0.100000
52.000000 -0.100000
53.000000 -0.100000
54.000000 -0.100000
55.000000 -0.100000
56.000000 -0.100000
57.000000 -0.100000
58.000000 -0.100000
59.000000 -0.100000
60.000000 -0.100000
61.000000 -0.100000
62.000000 -0.100000
63.000000 -0.100000
64.000000 -0.100000
65.000000 -0.100000
66.000000 -0.100000
67.000000 -0.100000
68.000000 -0.100000
69.000000 -0.100000
70.000000 -0.100000
71.000000 -0.100000
72.000000 -0.100000
73.000000 -0.100000
74.000000 -0.100000
75.000000 -0.100000
76.000000 -0.100000
77.000000 -0.100000
78.000000 -0.100000
79.000000 -0.100000
80.000000 -0.100000
};
\addlegendentry{$u= 30$}
\addplot table {
0.000000 0.620000
1.000000 0.590000
2.000000 0.560000
3.000000 0.530000
4.000000 0.505000
5.000000 0.485000
6.000000 0.460000
7.000000 0.440000
8.000000 0.415000
9.000000 0.395000
10.000000 0.375000
11.000000 0.355000
12.000000 0.340000
13.000000 0.320000
14.000000 0.300000
15.000000 0.285000
16.000000 0.265000
17.000000 0.250000
18.000000 0.230000
19.000000 0.215000
20.000000 0.200000
21.000000 0.185000
22.000000 0.170000
23.000000 0.155000
24.000000 0.140000
25.000000 0.125000
26.000000 0.110000
27.000000 0.095000
28.000000 0.080000
29.000000 0.065000
30.000000 0.055000
31.000000 0.040000
32.000000 0.025000
33.000000 0.015000
34.000000 0.000000
35.000000 -0.100000
36.000000 -0.100000
37.000000 -0.100000
38.000000 -0.100000
39.000000 -0.100000
40.000000 -0.100000
41.000000 -0.100000
42.000000 -0.100000
43.000000 -0.100000
44.000000 -0.100000
45.000000 -0.100000
46.000000 -0.100000
47.000000 -0.100000
48.000000 -0.100000
49.000000 -0.100000
50.000000 -0.100000
51.000000 -0.100000
52.000000 -0.100000
53.000000 -0.100000
54.000000 -0.100000
55.000000 -0.100000
56.000000 -0.100000
57.000000 -0.100000
58.000000 -0.100000
59.000000 -0.100000
60.000000 -0.100000
61.000000 -0.100000
62.000000 -0.100000
63.000000 -0.100000
64.000000 -0.100000
65.000000 -0.100000
66.000000 -0.100000
67.000000 -0.100000
68.000000 -0.100000
69.000000 -0.100000
70.000000 -0.100000
71.000000 -0.100000
72.000000 -0.100000
73.000000 -0.100000
74.000000 -0.100000
75.000000 -0.100000
76.000000 -0.100000
77.000000 -0.100000
78.000000 -0.100000
79.000000 -0.100000
80.000000 -0.100000
};
\addlegendentry{$u= 40$}
\addplot table {
0.000000 0.540000
1.000000 0.510000
2.000000 0.480000
3.000000 0.450000
4.000000 0.425000
5.000000 0.405000
6.000000 0.380000
7.000000 0.360000
8.000000 0.335000
9.000000 0.315000
10.000000 0.295000
11.000000 0.275000
12.000000 0.260000
13.000000 0.240000
14.000000 0.220000
15.000000 0.205000
16.000000 0.185000
17.000000 0.170000
18.000000 0.150000
19.000000 0.135000
20.000000 0.120000
21.000000 0.105000
22.000000 0.090000
23.000000 0.075000
24.000000 0.060000
25.000000 0.045000
26.000000 0.030000
27.000000 0.015000
28.000000 0.000000
29.000000 -0.100000
30.000000 -0.100000
31.000000 -0.100000
32.000000 -0.100000
33.000000 -0.100000
34.000000 -0.100000
35.000000 -0.100000
36.000000 -0.100000
37.000000 -0.100000
38.000000 -0.100000
39.000000 -0.100000
40.000000 -0.100000
41.000000 -0.100000
42.000000 -0.100000
43.000000 -0.100000
44.000000 -0.100000
45.000000 -0.100000
46.000000 -0.100000
47.000000 -0.100000
48.000000 -0.100000
49.000000 -0.100000
50.000000 -0.100000
51.000000 -0.100000
52.000000 -0.100000
53.000000 -0.100000
54.000000 -0.100000
55.000000 -0.100000
56.000000 -0.100000
57.000000 -0.100000
58.000000 -0.100000
59.000000 -0.100000
60.000000 -0.100000
61.000000 -0.100000
62.000000 -0.100000
63.000000 -0.100000
64.000000 -0.100000
65.000000 -0.100000
66.000000 -0.100000
67.000000 -0.100000
68.000000 -0.100000
69.000000 -0.100000
70.000000 -0.100000
71.000000 -0.100000
72.000000 -0.100000
73.000000 -0.100000
74.000000 -0.100000
75.000000 -0.100000
76.000000 -0.100000
77.000000 -0.100000
78.000000 -0.100000
79.000000 -0.100000
80.000000 -0.100000
};
\addlegendentry{$u= 50$}
\addplot table {
0.000000 0.465000
1.000000 0.435000
2.000000 0.405000
3.000000 0.375000
4.000000 0.350000
5.000000 0.330000
6.000000 0.305000
7.000000 0.285000
8.000000 0.260000
9.000000 0.240000
10.000000 0.220000
11.000000 0.200000
12.000000 0.185000
13.000000 0.165000
14.000000 0.145000
15.000000 0.130000
16.000000 0.110000
17.000000 0.095000
18.000000 0.075000
19.000000 0.060000
20.000000 0.045000
21.000000 0.030000
22.000000 0.015000
23.000000 0.000000
24.000000 -0.100000
25.000000 -0.100000
26.000000 -0.100000
27.000000 -0.100000
28.000000 -0.100000
29.000000 -0.100000
30.000000 -0.100000
31.000000 -0.100000
32.000000 -0.100000
33.000000 -0.100000
34.000000 -0.100000
35.000000 -0.100000
36.000000 -0.100000
37.000000 -0.100000
38.000000 -0.100000
39.000000 -0.100000
40.000000 -0.100000
41.000000 -0.100000
42.000000 -0.100000
43.000000 -0.100000
44.000000 -0.100000
45.000000 -0.100000
46.000000 -0.100000
47.000000 -0.100000
48.000000 -0.100000
49.000000 -0.100000
50.000000 -0.100000
51.000000 -0.100000
52.000000 -0.100000
53.000000 -0.100000
54.000000 -0.100000
55.000000 -0.100000
56.000000 -0.100000
57.000000 -0.100000
58.000000 -0.100000
59.000000 -0.100000
60.000000 -0.100000
61.000000 -0.100000
62.000000 -0.100000
63.000000 -0.100000
64.000000 -0.100000
65.000000 -0.100000
66.000000 -0.100000
67.000000 -0.100000
68.000000 -0.100000
69.000000 -0.100000
70.000000 -0.100000
71.000000 -0.100000
72.000000 -0.100000
73.000000 -0.100000
74.000000 -0.100000
75.000000 -0.100000
76.000000 -0.100000
77.000000 -0.100000
78.000000 -0.100000
79.000000 -0.100000
80.000000 -0.100000
};
\addlegendentry{$u= 60$}
\addplot table {
0.000000 0.395000
1.000000 0.365000
2.000000 0.335000
3.000000 0.305000
4.000000 0.280000
5.000000 0.260000
6.000000 0.235000
7.000000 0.215000
8.000000 0.190000
9.000000 0.170000
10.000000 0.150000
11.000000 0.130000
12.000000 0.115000
13.000000 0.095000
14.000000 0.075000
15.000000 0.060000
16.000000 0.040000
17.000000 0.025000
18.000000 0.005000
19.000000 -0.100000
20.000000 -0.100000
21.000000 -0.100000
22.000000 -0.100000
23.000000 -0.100000
24.000000 -0.100000
25.000000 -0.100000
26.000000 -0.100000
27.000000 -0.100000
28.000000 -0.100000
29.000000 -0.100000
30.000000 -0.100000
31.000000 -0.100000
32.000000 -0.100000
33.000000 -0.100000
34.000000 -0.100000
35.000000 -0.100000
36.000000 -0.100000
37.000000 -0.100000
38.000000 -0.100000
39.000000 -0.100000
40.000000 -0.100000
41.000000 -0.100000
42.000000 -0.100000
43.000000 -0.100000
44.000000 -0.100000
45.000000 -0.100000
46.000000 -0.100000
47.000000 -0.100000
48.000000 -0.100000
49.000000 -0.100000
50.000000 -0.100000
51.000000 -0.100000
52.000000 -0.100000
53.000000 -0.100000
54.000000 -0.100000
55.000000 -0.100000
56.000000 -0.100000
57.000000 -0.100000
58.000000 -0.100000
59.000000 -0.100000
60.000000 -0.100000
61.000000 -0.100000
62.000000 -0.100000
63.000000 -0.100000
64.000000 -0.100000
65.000000 -0.100000
66.000000 -0.100000
67.000000 -0.100000
68.000000 -0.100000
69.000000 -0.100000
70.000000 -0.100000
71.000000 -0.100000
72.000000 -0.100000
73.000000 -0.100000
74.000000 -0.100000
75.000000 -0.100000
76.000000 -0.100000
77.000000 -0.100000
78.000000 -0.100000
79.000000 -0.100000
80.000000 -0.100000
};
\addlegendentry{$u= 70$}
\addplot table {
0.000000 0.335000
1.000000 0.305000
2.000000 0.275000
3.000000 0.245000
4.000000 0.220000
5.000000 0.200000
6.000000 0.175000
7.000000 0.155000
8.000000 0.130000
9.000000 0.110000
10.000000 0.090000
11.000000 0.070000
12.000000 0.055000
13.000000 0.035000
14.000000 0.015000
15.000000 0.000000
16.000000 -0.100000
17.000000 -0.100000
18.000000 -0.100000
19.000000 -0.100000
20.000000 -0.100000
21.000000 -0.100000
22.000000 -0.100000
23.000000 -0.100000
24.000000 -0.100000
25.000000 -0.100000
26.000000 -0.100000
27.000000 -0.100000
28.000000 -0.100000
29.000000 -0.100000
30.000000 -0.100000
31.000000 -0.100000
32.000000 -0.100000
33.000000 -0.100000
34.000000 -0.100000
35.000000 -0.100000
36.000000 -0.100000
37.000000 -0.100000
38.000000 -0.100000
39.000000 -0.100000
40.000000 -0.100000
41.000000 -0.100000
42.000000 -0.100000
43.000000 -0.100000
44.000000 -0.100000
45.000000 -0.100000
46.000000 -0.100000
47.000000 -0.100000
48.000000 -0.100000
49.000000 -0.100000
50.000000 -0.100000
51.000000 -0.100000
52.000000 -0.100000
53.000000 -0.100000
54.000000 -0.100000
55.000000 -0.100000
56.000000 -0.100000
57.000000 -0.100000
58.000000 -0.100000
59.000000 -0.100000
60.000000 -0.100000
61.000000 -0.100000
62.000000 -0.100000
63.000000 -0.100000
64.000000 -0.100000
65.000000 -0.100000
66.000000 -0.100000
67.000000 -0.100000
68.000000 -0.100000
69.000000 -0.100000
70.000000 -0.100000
71.000000 -0.100000
72.000000 -0.100000
73.000000 -0.100000
74.000000 -0.100000
75.000000 -0.100000
76.000000 -0.100000
77.000000 -0.100000
78.000000 -0.100000
79.000000 -0.100000
80.000000 -0.100000
};
\addlegendentry{$u= 80$}
\addplot table {
0.000000 0.275000
1.000000 0.245000
2.000000 0.215000
3.000000 0.185000
4.000000 0.160000
5.000000 0.140000
6.000000 0.115000
7.000000 0.095000
8.000000 0.070000
9.000000 0.050000
10.000000 0.030000
11.000000 0.010000
12.000000 -0.100000
13.000000 -0.100000
14.000000 -0.100000
15.000000 -0.100000
16.000000 -0.100000
17.000000 -0.100000
18.000000 -0.100000
19.000000 -0.100000
20.000000 -0.100000
21.000000 -0.100000
22.000000 -0.100000
23.000000 -0.100000
24.000000 -0.100000
25.000000 -0.100000
26.000000 -0.100000
27.000000 -0.100000
28.000000 -0.100000
29.000000 -0.100000
30.000000 -0.100000
31.000000 -0.100000
32.000000 -0.100000
33.000000 -0.100000
34.000000 -0.100000
35.000000 -0.100000
36.000000 -0.100000
37.000000 -0.100000
38.000000 -0.100000
39.000000 -0.100000
40.000000 -0.100000
41.000000 -0.100000
42.000000 -0.100000
43.000000 -0.100000
44.000000 -0.100000
45.000000 -0.100000
46.000000 -0.100000
47.000000 -0.100000
48.000000 -0.100000
49.000000 -0.100000
50.000000 -0.100000
51.000000 -0.100000
52.000000 -0.100000
53.000000 -0.100000
54.000000 -0.100000
55.000000 -0.100000
56.000000 -0.100000
57.000000 -0.100000
58.000000 -0.100000
59.000000 -0.100000
60.000000 -0.100000
61.000000 -0.100000
62.000000 -0.100000
63.000000 -0.100000
64.000000 -0.100000
65.000000 -0.100000
66.000000 -0.100000
67.000000 -0.100000
68.000000 -0.100000
69.000000 -0.100000
70.000000 -0.100000
71.000000 -0.100000
72.000000 -0.100000
73.000000 -0.100000
74.000000 -0.100000
75.000000 -0.100000
76.000000 -0.100000
77.000000 -0.100000
78.000000 -0.100000
79.000000 -0.100000
80.000000 -0.100000
};
\addlegendentry{$u= 90$}
\addplot table {
0.000000 0.225000
1.000000 0.195000
2.000000 0.165000
3.000000 0.135000
4.000000 0.110000
5.000000 0.090000
6.000000 0.065000
7.000000 0.045000
8.000000 0.020000
9.000000 0.000000
10.000000 -0.100000
11.000000 -0.100000
12.000000 -0.100000
13.000000 -0.100000
14.000000 -0.100000
15.000000 -0.100000
16.000000 -0.100000
17.000000 -0.100000
18.000000 -0.100000
19.000000 -0.100000
20.000000 -0.100000
21.000000 -0.100000
22.000000 -0.100000
23.000000 -0.100000
24.000000 -0.100000
25.000000 -0.100000
26.000000 -0.100000
27.000000 -0.100000
28.000000 -0.100000
29.000000 -0.100000
30.000000 -0.100000
31.000000 -0.100000
32.000000 -0.100000
33.000000 -0.100000
34.000000 -0.100000
35.000000 -0.100000
36.000000 -0.100000
37.000000 -0.100000
38.000000 -0.100000
39.000000 -0.100000
40.000000 -0.100000
41.000000 -0.100000
42.000000 -0.100000
43.000000 -0.100000
44.000000 -0.100000
45.000000 -0.100000
46.000000 -0.100000
47.000000 -0.100000
48.000000 -0.100000
49.000000 -0.100000
50.000000 -0.100000
51.000000 -0.100000
52.000000 -0.100000
53.000000 -0.100000
54.000000 -0.100000
55.000000 -0.100000
56.000000 -0.100000
57.000000 -0.100000
58.000000 -0.100000
59.000000 -0.100000
60.000000 -0.100000
61.000000 -0.100000
62.000000 -0.100000
63.000000 -0.100000
64.000000 -0.100000
65.000000 -0.100000
66.000000 -0.100000
67.000000 -0.100000
68.000000 -0.100000
69.000000 -0.100000
70.000000 -0.100000
71.000000 -0.100000
72.000000 -0.100000
73.000000 -0.100000
74.000000 -0.100000
75.000000 -0.100000
76.000000 -0.100000
77.000000 -0.100000
78.000000 -0.100000
79.000000 -0.100000
80.000000 -0.100000
};
\addlegendentry{$u= 100$}

%% file: figures/comparisons/Const1_to_ord_GV/ordGvq7.tex
\addplot table {
0.000000 0.931106
1.000000 0.903032
2.000000 0.874959
3.000000 0.846885
4.000000 0.823491
5.000000 0.800096
6.000000 0.776702
7.000000 0.757986
8.000000 0.734591
9.000000 0.715876
10.000000 0.697160
11.000000 0.678444
12.000000 0.659728
13.000000 0.641013
14.000000 0.622297
15.000000 0.608260
16.000000 0.589545
17.000000 0.575508
18.000000 0.556792
19.000000 0.542755
20.000000 0.528719
21.000000 0.510003
22.000000 0.495966
23.000000 0.481929
24.000000 0.467892
25.000000 0.453856
26.000000 0.439819
27.000000 0.425782
28.000000 0.411745
29.000000 0.402388
30.000000 0.388351
31.000000 0.374314
32.000000 0.364956
33.000000 0.350919
34.000000 0.336883
35.000000 0.327525
36.000000 0.313488
37.000000 0.304130
38.000000 0.294772
39.000000 0.280735
40.000000 0.271378
41.000000 0.262020
42.000000 0.252662
43.000000 0.238625
44.000000 0.229267
45.000000 0.219909
46.000000 0.210552
47.000000 0.201194
48.000000 0.191836
49.000000 0.182478
50.000000 0.177799
51.000000 0.168441
52.000000 0.159083
53.000000 0.149726
54.000000 0.145047
55.000000 0.135689
56.000000 0.126331
57.000000 0.121652
58.000000 0.112294
59.000000 0.107615
60.000000 0.098257
61.000000 0.093578
62.000000 0.084221
63.000000 0.079542
64.000000 0.074863
65.000000 0.070184
66.000000 0.060826
67.000000 0.056147
68.000000 0.051468
69.000000 0.046789
70.000000 0.042110
71.000000 0.037431
72.000000 0.032752
73.000000 0.028074
74.000000 0.023395
75.000000 0.023395
76.000000 0.018716
77.000000 0.014037
78.000000 0.014037
79.000000 0.009358
80.000000 0.009358
};
\addlegendentry{$(q-1)$-ary GV \\
		($u$ arbitrary)}

%% file: figures/comparisons/GV_const_2_3_q8/plot_list_constbin_q8.tex
\addplot table {
0.000000 0.991667
1.000000 0.961667
2.000000 0.931667
3.000000 0.906667
4.000000 0.881667
5.000000 0.856667
6.000000 0.831667
7.000000 0.811667
8.000000 0.791667
9.000000 0.771667
10.000000 0.751667
11.000000 0.731667
12.000000 0.711667
13.000000 0.691667
14.000000 0.676667
15.000000 0.656667
16.000000 0.636667
17.000000 0.621667
18.000000 0.606667
19.000000 0.586667
20.000000 0.571667
21.000000 0.556667
22.000000 0.541667
23.000000 0.526667
24.000000 0.511667
25.000000 0.496667
26.000000 0.481667
27.000000 0.466667
28.000000 0.451667
29.000000 0.441667
30.000000 0.426667
31.000000 0.411667
32.000000 0.401667
33.000000 0.386667
34.000000 0.376667
35.000000 0.361667
36.000000 0.351667
37.000000 0.336667
38.000000 0.326667
39.000000 0.316667
40.000000 0.301667
41.000000 0.291667
42.000000 0.281667
43.000000 0.271667
44.000000 0.261667
45.000000 0.246667
46.000000 0.236667
47.000000 0.226667
48.000000 -0.100000
49.000000 -0.100000
50.000000 -0.100000
51.000000 -0.100000
52.000000 -0.100000
53.000000 -0.100000
54.000000 -0.100000
55.000000 -0.100000
56.000000 -0.100000
57.000000 -0.100000
58.000000 -0.100000
59.000000 -0.100000
60.000000 -0.100000
61.000000 -0.100000
62.000000 -0.100000
63.000000 -0.100000
64.000000 -0.100000
65.000000 -0.100000
66.000000 -0.100000
67.000000 -0.100000
68.000000 -0.100000
69.000000 -0.100000
70.000000 -0.100000
71.000000 -0.100000
72.000000 -0.100000
73.000000 -0.100000
74.000000 -0.100000
75.000000 -0.100000
76.000000 -0.100000
77.000000 -0.100000
78.000000 -0.100000
79.000000 -0.100000
80.000000 -0.100000
};
\addplot table {
0.000000 0.975000
1.000000 0.945000
2.000000 0.915000
3.000000 0.890000
4.000000 0.865000
5.000000 0.840000
6.000000 0.815000
7.000000 0.795000
8.000000 0.775000
9.000000 0.755000
10.000000 0.735000
11.000000 0.715000
12.000000 0.695000
13.000000 0.675000
14.000000 0.660000
15.000000 0.640000
16.000000 0.620000
17.000000 0.605000
18.000000 0.590000
19.000000 0.570000
20.000000 0.555000
21.000000 0.540000
22.000000 0.525000
23.000000 0.510000
24.000000 0.495000
25.000000 0.480000
26.000000 0.465000
27.000000 0.450000
28.000000 0.435000
29.000000 0.425000
30.000000 0.410000
31.000000 0.395000
32.000000 0.385000
33.000000 0.370000
34.000000 0.360000
35.000000 0.345000
36.000000 0.335000
37.000000 0.320000
38.000000 -0.100000
39.000000 -0.100000
40.000000 -0.100000
41.000000 -0.100000
42.000000 -0.100000
43.000000 -0.100000
44.000000 -0.100000
45.000000 -0.100000
46.000000 -0.100000
47.000000 -0.100000
48.000000 -0.100000
49.000000 -0.100000
50.000000 -0.100000
51.000000 -0.100000
52.000000 -0.100000
53.000000 -0.100000
54.000000 -0.100000
55.000000 -0.100000
56.000000 -0.100000
57.000000 -0.100000
58.000000 -0.100000
59.000000 -0.100000
60.000000 -0.100000
61.000000 -0.100000
62.000000 -0.100000
63.000000 -0.100000
64.000000 -0.100000
65.000000 -0.100000
66.000000 -0.100000
67.000000 -0.100000
68.000000 -0.100000
69.000000 -0.100000
70.000000 -0.100000
71.000000 -0.100000
72.000000 -0.100000
73.000000 -0.100000
74.000000 -0.100000
75.000000 -0.100000
76.000000 -0.100000
77.000000 -0.100000
78.000000 -0.100000
79.000000 -0.100000
80.000000 -0.100000
};
\addplot table {
0.000000 0.945000
1.000000 0.915000
2.000000 0.885000
3.000000 0.860000
4.000000 0.835000
5.000000 0.810000
6.000000 0.785000
7.000000 0.765000
8.000000 0.745000
9.000000 0.725000
10.000000 0.705000
11.000000 0.685000
12.000000 0.665000
13.000000 0.645000
14.000000 0.630000
15.000000 0.610000
16.000000 0.590000
17.000000 0.575000
18.000000 0.560000
19.000000 0.540000
20.000000 0.525000
21.000000 0.510000
22.000000 0.495000
23.000000 0.480000
24.000000 0.465000
25.000000 0.450000
26.000000 0.435000
27.000000 0.420000
28.000000 0.405000
29.000000 -0.100000
30.000000 -0.100000
31.000000 -0.100000
32.000000 -0.100000
33.000000 -0.100000
34.000000 -0.100000
35.000000 -0.100000
36.000000 -0.100000
37.000000 -0.100000
38.000000 -0.100000
39.000000 -0.100000
40.000000 -0.100000
41.000000 -0.100000
42.000000 -0.100000
43.000000 -0.100000
44.000000 -0.100000
45.000000 -0.100000
46.000000 -0.100000
47.000000 -0.100000
48.000000 -0.100000
49.000000 -0.100000
50.000000 -0.100000
51.000000 -0.100000
52.000000 -0.100000
53.000000 -0.100000
54.000000 -0.100000
55.000000 -0.100000
56.000000 -0.100000
57.000000 -0.100000
58.000000 -0.100000
59.000000 -0.100000
60.000000 -0.100000
61.000000 -0.100000
62.000000 -0.100000
63.000000 -0.100000
64.000000 -0.100000
65.000000 -0.100000
66.000000 -0.100000
67.000000 -0.100000
68.000000 -0.100000
69.000000 -0.100000
70.000000 -0.100000
71.000000 -0.100000
72.000000 -0.100000
73.000000 -0.100000
74.000000 -0.100000
75.000000 -0.100000
76.000000 -0.100000
77.000000 -0.100000
78.000000 -0.100000
79.000000 -0.100000
80.000000 -0.100000
};
\addplot table {
0.000000 0.926667
1.000000 0.896667
2.000000 0.866667
3.000000 0.841667
4.000000 0.816667
5.000000 0.791667
6.000000 0.766667
7.000000 0.746667
8.000000 0.726667
9.000000 0.706667
10.000000 0.686667
11.000000 0.666667
12.000000 0.646667
13.000000 0.626667
14.000000 0.611667
15.000000 0.591667
16.000000 0.571667
17.000000 0.556667
18.000000 0.541667
19.000000 0.521667
20.000000 0.506667
21.000000 0.491667
22.000000 0.476667
23.000000 0.461667
24.000000 0.446667
25.000000 0.431667
26.000000 -0.100000
27.000000 -0.100000
28.000000 -0.100000
29.000000 -0.100000
30.000000 -0.100000
31.000000 -0.100000
32.000000 -0.100000
33.000000 -0.100000
34.000000 -0.100000
35.000000 -0.100000
36.000000 -0.100000
37.000000 -0.100000
38.000000 -0.100000
39.000000 -0.100000
40.000000 -0.100000
41.000000 -0.100000
42.000000 -0.100000
43.000000 -0.100000
44.000000 -0.100000
45.000000 -0.100000
46.000000 -0.100000
47.000000 -0.100000
48.000000 -0.100000
49.000000 -0.100000
50.000000 -0.100000
51.000000 -0.100000
52.000000 -0.100000
53.000000 -0.100000
54.000000 -0.100000
55.000000 -0.100000
56.000000 -0.100000
57.000000 -0.100000
58.000000 -0.100000
59.000000 -0.100000
60.000000 -0.100000
61.000000 -0.100000
62.000000 -0.100000
63.000000 -0.100000
64.000000 -0.100000
65.000000 -0.100000
66.000000 -0.100000
67.000000 -0.100000
68.000000 -0.100000
69.000000 -0.100000
70.000000 -0.100000
71.000000 -0.100000
72.000000 -0.100000
73.000000 -0.100000
74.000000 -0.100000
75.000000 -0.100000
76.000000 -0.100000
77.000000 -0.100000
78.000000 -0.100000
79.000000 -0.100000
80.000000 -0.100000
};
\addplot table {
0.000000 0.903333
1.000000 0.873333
2.000000 0.843333
3.000000 0.818333
4.000000 0.793333
5.000000 0.768333
6.000000 0.743333
7.000000 0.723333
8.000000 0.703333
9.000000 0.683333
10.000000 0.663333
11.000000 0.643333
12.000000 0.623333
13.000000 0.603333
14.000000 0.588333
15.000000 0.568333
16.000000 0.548333
17.000000 0.533333
18.000000 0.518333
19.000000 0.498333
20.000000 0.483333
21.000000 0.468333
22.000000 -0.100000
23.000000 -0.100000
24.000000 -0.100000
25.000000 -0.100000
26.000000 -0.100000
27.000000 -0.100000
28.000000 -0.100000
29.000000 -0.100000
30.000000 -0.100000
31.000000 -0.100000
32.000000 -0.100000
33.000000 -0.100000
34.000000 -0.100000
35.000000 -0.100000
36.000000 -0.100000
37.000000 -0.100000
38.000000 -0.100000
39.000000 -0.100000
40.000000 -0.100000
41.000000 -0.100000
42.000000 -0.100000
43.000000 -0.100000
44.000000 -0.100000
45.000000 -0.100000
46.000000 -0.100000
47.000000 -0.100000
48.000000 -0.100000
49.000000 -0.100000
50.000000 -0.100000
51.000000 -0.100000
52.000000 -0.100000
53.000000 -0.100000
54.000000 -0.100000
55.000000 -0.100000
56.000000 -0.100000
57.000000 -0.100000
58.000000 -0.100000
59.000000 -0.100000
60.000000 -0.100000
61.000000 -0.100000
62.000000 -0.100000
63.000000 -0.100000
64.000000 -0.100000
65.000000 -0.100000
66.000000 -0.100000
67.000000 -0.100000
68.000000 -0.100000
69.000000 -0.100000
70.000000 -0.100000
71.000000 -0.100000
72.000000 -0.100000
73.000000 -0.100000
74.000000 -0.100000
75.000000 -0.100000
76.000000 -0.100000
77.000000 -0.100000
78.000000 -0.100000
79.000000 -0.100000
80.000000 -0.100000
};
\addplot table {
0.000000 0.890000
1.000000 0.860000
2.000000 0.830000
3.000000 0.805000
4.000000 0.780000
5.000000 0.755000
6.000000 0.730000
7.000000 0.710000
8.000000 0.690000
9.000000 0.670000
10.000000 0.650000
11.000000 0.630000
12.000000 0.610000
13.000000 0.590000
14.000000 0.575000
15.000000 0.555000
16.000000 0.535000
17.000000 0.520000
18.000000 0.505000
19.000000 0.485000
20.000000 -0.100000
21.000000 -0.100000
22.000000 -0.100000
23.000000 -0.100000
24.000000 -0.100000
25.000000 -0.100000
26.000000 -0.100000
27.000000 -0.100000
28.000000 -0.100000
29.000000 -0.100000
30.000000 -0.100000
31.000000 -0.100000
32.000000 -0.100000
33.000000 -0.100000
34.000000 -0.100000
35.000000 -0.100000
36.000000 -0.100000
37.000000 -0.100000
38.000000 -0.100000
39.000000 -0.100000
40.000000 -0.100000
41.000000 -0.100000
42.000000 -0.100000
43.000000 -0.100000
44.000000 -0.100000
45.000000 -0.100000
46.000000 -0.100000
47.000000 -0.100000
48.000000 -0.100000
49.000000 -0.100000
50.000000 -0.100000
51.000000 -0.100000
52.000000 -0.100000
53.000000 -0.100000
54.000000 -0.100000
55.000000 -0.100000
56.000000 -0.100000
57.000000 -0.100000
58.000000 -0.100000
59.000000 -0.100000
60.000000 -0.100000
61.000000 -0.100000
62.000000 -0.100000
63.000000 -0.100000
64.000000 -0.100000
65.000000 -0.100000
66.000000 -0.100000
67.000000 -0.100000
68.000000 -0.100000
69.000000 -0.100000
70.000000 -0.100000
71.000000 -0.100000
72.000000 -0.100000
73.000000 -0.100000
74.000000 -0.100000
75.000000 -0.100000
76.000000 -0.100000
77.000000 -0.100000
78.000000 -0.100000
79.000000 -0.100000
80.000000 -0.100000
};
\addplot table {
0.000000 0.870000
1.000000 0.840000
2.000000 0.810000
3.000000 0.785000
4.000000 0.760000
5.000000 0.735000
6.000000 0.710000
7.000000 0.690000
8.000000 0.670000
9.000000 0.650000
10.000000 0.630000
11.000000 0.610000
12.000000 0.590000
13.000000 0.570000
14.000000 0.555000
15.000000 0.535000
16.000000 0.515000
17.000000 -0.100000
18.000000 -0.100000
19.000000 -0.100000
20.000000 -0.100000
21.000000 -0.100000
22.000000 -0.100000
23.000000 -0.100000
24.000000 -0.100000
25.000000 -0.100000
26.000000 -0.100000
27.000000 -0.100000
28.000000 -0.100000
29.000000 -0.100000
30.000000 -0.100000
31.000000 -0.100000
32.000000 -0.100000
33.000000 -0.100000
34.000000 -0.100000
35.000000 -0.100000
36.000000 -0.100000
37.000000 -0.100000
38.000000 -0.100000
39.000000 -0.100000
40.000000 -0.100000
41.000000 -0.100000
42.000000 -0.100000
43.000000 -0.100000
44.000000 -0.100000
45.000000 -0.100000
46.000000 -0.100000
47.000000 -0.100000
48.000000 -0.100000
49.000000 -0.100000
50.000000 -0.100000
51.000000 -0.100000
52.000000 -0.100000
53.000000 -0.100000
54.000000 -0.100000
55.000000 -0.100000
56.000000 -0.100000
57.000000 -0.100000
58.000000 -0.100000
59.000000 -0.100000
60.000000 -0.100000
61.000000 -0.100000
62.000000 -0.100000
63.000000 -0.100000
64.000000 -0.100000
65.000000 -0.100000
66.000000 -0.100000
67.000000 -0.100000
68.000000 -0.100000
69.000000 -0.100000
70.000000 -0.100000
71.000000 -0.100000
72.000000 -0.100000
73.000000 -0.100000
74.000000 -0.100000
75.000000 -0.100000
76.000000 -0.100000
77.000000 -0.100000
78.000000 -0.100000
79.000000 -0.100000
80.000000 -0.100000
};
\addplot table {
0.000000 0.858333
1.000000 0.828333
2.000000 0.798333
3.000000 0.773333
4.000000 0.748333
5.000000 0.723333
6.000000 0.698333
7.000000 0.678333
8.000000 0.658333
9.000000 0.638333
10.000000 0.618333
11.000000 0.598333
12.000000 0.578333
13.000000 0.558333
14.000000 0.543333
15.000000 -0.100000
16.000000 -0.100000
17.000000 -0.100000
18.000000 -0.100000
19.000000 -0.100000
20.000000 -0.100000
21.000000 -0.100000
22.000000 -0.100000
23.000000 -0.100000
24.000000 -0.100000
25.000000 -0.100000
26.000000 -0.100000
27.000000 -0.100000
28.000000 -0.100000
29.000000 -0.100000
30.000000 -0.100000
31.000000 -0.100000
32.000000 -0.100000
33.000000 -0.100000
34.000000 -0.100000
35.000000 -0.100000
36.000000 -0.100000
37.000000 -0.100000
38.000000 -0.100000
39.000000 -0.100000
40.000000 -0.100000
41.000000 -0.100000
42.000000 -0.100000
43.000000 -0.100000
44.000000 -0.100000
45.000000 -0.100000
46.000000 -0.100000
47.000000 -0.100000
48.000000 -0.100000
49.000000 -0.100000
50.000000 -0.100000
51.000000 -0.100000
52.000000 -0.100000
53.000000 -0.100000
54.000000 -0.100000
55.000000 -0.100000
56.000000 -0.100000
57.000000 -0.100000
58.000000 -0.100000
59.000000 -0.100000
60.000000 -0.100000
61.000000 -0.100000
62.000000 -0.100000
63.000000 -0.100000
64.000000 -0.100000
65.000000 -0.100000
66.000000 -0.100000
67.000000 -0.100000
68.000000 -0.100000
69.000000 -0.100000
70.000000 -0.100000
71.000000 -0.100000
72.000000 -0.100000
73.000000 -0.100000
74.000000 -0.100000
75.000000 -0.100000
76.000000 -0.100000
77.000000 -0.100000
78.000000 -0.100000
79.000000 -0.100000
80.000000 -0.100000
};
\addplot table {
0.000000 0.841667
1.000000 0.811667
2.000000 0.781667
3.000000 0.756667
4.000000 0.731667
5.000000 0.706667
6.000000 0.681667
7.000000 0.661667
8.000000 0.641667
9.000000 0.621667
10.000000 0.601667
11.000000 0.581667
12.000000 0.561667
13.000000 -0.100000
14.000000 -0.100000
15.000000 -0.100000
16.000000 -0.100000
17.000000 -0.100000
18.000000 -0.100000
19.000000 -0.100000
20.000000 -0.100000
21.000000 -0.100000
22.000000 -0.100000
23.000000 -0.100000
24.000000 -0.100000
25.000000 -0.100000
26.000000 -0.100000
27.000000 -0.100000
28.000000 -0.100000
29.000000 -0.100000
30.000000 -0.100000
31.000000 -0.100000
32.000000 -0.100000
33.000000 -0.100000
34.000000 -0.100000
35.000000 -0.100000
36.000000 -0.100000
37.000000 -0.100000
38.000000 -0.100000
39.000000 -0.100000
40.000000 -0.100000
41.000000 -0.100000
42.000000 -0.100000
43.000000 -0.100000
44.000000 -0.100000
45.000000 -0.100000
46.000000 -0.100000
47.000000 -0.100000
48.000000 -0.100000
49.000000 -0.100000
50.000000 -0.100000
51.000000 -0.100000
52.000000 -0.100000
53.000000 -0.100000
54.000000 -0.100000
55.000000 -0.100000
56.000000 -0.100000
57.000000 -0.100000
58.000000 -0.100000
59.000000 -0.100000
60.000000 -0.100000
61.000000 -0.100000
62.000000 -0.100000
63.000000 -0.100000
64.000000 -0.100000
65.000000 -0.100000
66.000000 -0.100000
67.000000 -0.100000
68.000000 -0.100000
69.000000 -0.100000
70.000000 -0.100000
71.000000 -0.100000
72.000000 -0.100000
73.000000 -0.100000
74.000000 -0.100000
75.000000 -0.100000
76.000000 -0.100000
77.000000 -0.100000
78.000000 -0.100000
79.000000 -0.100000
80.000000 -0.100000
};
\addplot table {
0.000000 0.831667
1.000000 0.801667
2.000000 0.771667
3.000000 0.746667
4.000000 0.721667
5.000000 0.696667
6.000000 0.671667
7.000000 0.651667
8.000000 0.631667
9.000000 0.611667
10.000000 0.591667
11.000000 0.571667
12.000000 -0.100000
13.000000 -0.100000
14.000000 -0.100000
15.000000 -0.100000
16.000000 -0.100000
17.000000 -0.100000
18.000000 -0.100000
19.000000 -0.100000
20.000000 -0.100000
21.000000 -0.100000
22.000000 -0.100000
23.000000 -0.100000
24.000000 -0.100000
25.000000 -0.100000
26.000000 -0.100000
27.000000 -0.100000
28.000000 -0.100000
29.000000 -0.100000
30.000000 -0.100000
31.000000 -0.100000
32.000000 -0.100000
33.000000 -0.100000
34.000000 -0.100000
35.000000 -0.100000
36.000000 -0.100000
37.000000 -0.100000
38.000000 -0.100000
39.000000 -0.100000
40.000000 -0.100000
41.000000 -0.100000
42.000000 -0.100000
43.000000 -0.100000
44.000000 -0.100000
45.000000 -0.100000
46.000000 -0.100000
47.000000 -0.100000
48.000000 -0.100000
49.000000 -0.100000
50.000000 -0.100000
51.000000 -0.100000
52.000000 -0.100000
53.000000 -0.100000
54.000000 -0.100000
55.000000 -0.100000
56.000000 -0.100000
57.000000 -0.100000
58.000000 -0.100000
59.000000 -0.100000
60.000000 -0.100000
61.000000 -0.100000
62.000000 -0.100000
63.000000 -0.100000
64.000000 -0.100000
65.000000 -0.100000
66.000000 -0.100000
67.000000 -0.100000
68.000000 -0.100000
69.000000 -0.100000
70.000000 -0.100000
71.000000 -0.100000
72.000000 -0.100000
73.000000 -0.100000
74.000000 -0.100000
75.000000 -0.100000
76.000000 -0.100000
77.000000 -0.100000
78.000000 -0.100000
79.000000 -0.100000
80.000000 -0.100000
};
\addplot table {
0.000000 0.816667
1.000000 0.786667
2.000000 0.756667
3.000000 0.731667
4.000000 0.706667
5.000000 0.681667
6.000000 0.656667
7.000000 0.636667
8.000000 0.616667
9.000000 0.596667
10.000000 0.576667
11.000000 -0.100000
12.000000 -0.100000
13.000000 -0.100000
14.000000 -0.100000
15.000000 -0.100000
16.000000 -0.100000
17.000000 -0.100000
18.000000 -0.100000
19.000000 -0.100000
20.000000 -0.100000
21.000000 -0.100000
22.000000 -0.100000
23.000000 -0.100000
24.000000 -0.100000
25.000000 -0.100000
26.000000 -0.100000
27.000000 -0.100000
28.000000 -0.100000
29.000000 -0.100000
30.000000 -0.100000
31.000000 -0.100000
32.000000 -0.100000
33.000000 -0.100000
34.000000 -0.100000
35.000000 -0.100000
36.000000 -0.100000
37.000000 -0.100000
38.000000 -0.100000
39.000000 -0.100000
40.000000 -0.100000
41.000000 -0.100000
42.000000 -0.100000
43.000000 -0.100000
44.000000 -0.100000
45.000000 -0.100000
46.000000 -0.100000
47.000000 -0.100000
48.000000 -0.100000
49.000000 -0.100000
50.000000 -0.100000
51.000000 -0.100000
52.000000 -0.100000
53.000000 -0.100000
54.000000 -0.100000
55.000000 -0.100000
56.000000 -0.100000
57.000000 -0.100000
58.000000 -0.100000
59.000000 -0.100000
60.000000 -0.100000
61.000000 -0.100000
62.000000 -0.100000
63.000000 -0.100000
64.000000 -0.100000
65.000000 -0.100000
66.000000 -0.100000
67.000000 -0.100000
68.000000 -0.100000
69.000000 -0.100000
70.000000 -0.100000
71.000000 -0.100000
72.000000 -0.100000
73.000000 -0.100000
74.000000 -0.100000
75.000000 -0.100000
76.000000 -0.100000
77.000000 -0.100000
78.000000 -0.100000
79.000000 -0.100000
80.000000 -0.100000
};

%% file: figures/comparisons/GV_const2_q4/plot_list_const_2_q4.tex
\addplot table {
0.000000 0.995000
1.000000 0.950000
2.000000 0.915000
3.000000 0.880000
4.000000 0.850000
5.000000 0.820000
6.000000 0.790000
7.000000 0.765000
8.000000 0.740000
9.000000 0.715000
10.000000 0.690000
11.000000 0.665000
12.000000 0.645000
13.000000 0.620000
14.000000 0.600000
15.000000 0.580000
16.000000 0.560000
17.000000 0.540000
18.000000 0.520000
19.000000 0.505000
20.000000 0.485000
21.000000 0.465000
22.000000 0.450000
23.000000 0.435000
24.000000 0.415000
25.000000 0.400000
26.000000 0.385000
27.000000 0.370000
28.000000 0.355000
29.000000 0.340000
30.000000 0.325000
31.000000 0.310000
32.000000 0.300000
33.000000 0.285000
34.000000 0.275000
35.000000 0.260000
36.000000 0.250000
37.000000 0.235000
38.000000 0.225000
39.000000 0.215000
40.000000 0.200000
41.000000 0.190000
42.000000 0.180000
43.000000 0.170000
44.000000 0.160000
45.000000 0.150000
46.000000 0.140000
47.000000 0.135000
48.000000 0.125000
49.000000 0.115000
50.000000 0.110000
51.000000 0.100000
52.000000 0.090000
53.000000 0.085000
54.000000 0.080000
55.000000 0.070000
56.000000 0.065000
57.000000 0.060000
58.000000 0.055000
59.000000 0.045000
60.000000 0.040000
61.000000 0.035000
62.000000 0.030000
63.000000 0.030000
64.000000 0.025000
65.000000 0.020000
66.000000 0.015000
67.000000 0.015000
68.000000 0.010000
69.000000 0.005000
70.000000 0.005000
71.000000 0.005000
72.000000 0.000000
73.000000 0.000000
74.000000 0.000000
75.000000 -0.100000
76.000000 -0.100000
77.000000 -0.100000
78.000000 -0.100000
79.000000 -0.100000
80.000000 -0.100000
};
\addlegendentry{$u= 0$}
\addplot table {
0.000000 0.845000
1.000000 0.800000
2.000000 0.765000
3.000000 0.730000
4.000000 0.700000
5.000000 0.670000
6.000000 0.640000
7.000000 0.615000
8.000000 0.590000
9.000000 0.565000
10.000000 0.540000
11.000000 0.515000
12.000000 0.495000
13.000000 0.470000
14.000000 0.450000
15.000000 0.430000
16.000000 0.410000
17.000000 0.390000
18.000000 0.370000
19.000000 0.355000
20.000000 0.335000
21.000000 0.315000
22.000000 0.300000
23.000000 0.285000
24.000000 0.265000
25.000000 0.250000
26.000000 0.235000
27.000000 0.220000
28.000000 0.205000
29.000000 0.190000
30.000000 0.175000
31.000000 0.160000
32.000000 0.150000
33.000000 0.135000
34.000000 0.125000
35.000000 0.110000
36.000000 0.100000
37.000000 0.085000
38.000000 0.075000
39.000000 0.065000
40.000000 0.050000
41.000000 0.040000
42.000000 0.030000
43.000000 0.020000
44.000000 0.010000
45.000000 0.000000
46.000000 -0.100000
47.000000 -0.100000
48.000000 -0.100000
49.000000 -0.100000
50.000000 -0.100000
51.000000 -0.100000
52.000000 -0.100000
53.000000 -0.100000
54.000000 -0.100000
55.000000 -0.100000
56.000000 -0.100000
57.000000 -0.100000
58.000000 -0.100000
59.000000 -0.100000
60.000000 -0.100000
61.000000 -0.100000
62.000000 -0.100000
63.000000 -0.100000
64.000000 -0.100000
65.000000 -0.100000
66.000000 -0.100000
67.000000 -0.100000
68.000000 -0.100000
69.000000 -0.100000
70.000000 -0.100000
71.000000 -0.100000
72.000000 -0.100000
73.000000 -0.100000
74.000000 -0.100000
75.000000 -0.100000
76.000000 -0.100000
77.000000 -0.100000
78.000000 -0.100000
79.000000 -0.100000
80.000000 -0.100000
};
\addlegendentry{$u= 10$}
\addplot table {
0.000000 0.710000
1.000000 0.665000
2.000000 0.630000
3.000000 0.595000
4.000000 0.565000
5.000000 0.535000
6.000000 0.505000
7.000000 0.480000
8.000000 0.455000
9.000000 0.430000
10.000000 0.405000
11.000000 0.380000
12.000000 0.360000
13.000000 0.335000
14.000000 0.315000
15.000000 0.295000
16.000000 0.275000
17.000000 0.255000
18.000000 0.235000
19.000000 0.220000
20.000000 0.200000
21.000000 0.180000
22.000000 0.165000
23.000000 0.150000
24.000000 0.130000
25.000000 0.115000
26.000000 0.100000
27.000000 0.085000
28.000000 0.070000
29.000000 0.055000
30.000000 0.040000
31.000000 0.025000
32.000000 0.015000
33.000000 0.000000
34.000000 -0.100000
35.000000 -0.100000
36.000000 -0.100000
37.000000 -0.100000
38.000000 -0.100000
39.000000 -0.100000
40.000000 -0.100000
41.000000 -0.100000
42.000000 -0.100000
43.000000 -0.100000
44.000000 -0.100000
45.000000 -0.100000
46.000000 -0.100000
47.000000 -0.100000
48.000000 -0.100000
49.000000 -0.100000
50.000000 -0.100000
51.000000 -0.100000
52.000000 -0.100000
53.000000 -0.100000
54.000000 -0.100000
55.000000 -0.100000
56.000000 -0.100000
57.000000 -0.100000
58.000000 -0.100000
59.000000 -0.100000
60.000000 -0.100000
61.000000 -0.100000
62.000000 -0.100000
63.000000 -0.100000
64.000000 -0.100000
65.000000 -0.100000
66.000000 -0.100000
67.000000 -0.100000
68.000000 -0.100000
69.000000 -0.100000
70.000000 -0.100000
71.000000 -0.100000
72.000000 -0.100000
73.000000 -0.100000
74.000000 -0.100000
75.000000 -0.100000
76.000000 -0.100000
77.000000 -0.100000
78.000000 -0.100000
79.000000 -0.100000
80.000000 -0.100000
};
\addlegendentry{$u= 20$}
\addplot table {
0.000000 0.595000
1.000000 0.550000
2.000000 0.515000
3.000000 0.480000
4.000000 0.450000
5.000000 0.420000
6.000000 0.390000
7.000000 0.365000
8.000000 0.340000
9.000000 0.315000
10.000000 0.290000
11.000000 0.265000
12.000000 0.245000
13.000000 0.220000
14.000000 0.200000
15.000000 0.180000
16.000000 0.160000
17.000000 0.140000
18.000000 0.120000
19.000000 0.105000
20.000000 0.085000
21.000000 0.065000
22.000000 0.050000
23.000000 0.035000
24.000000 0.015000
25.000000 0.000000
26.000000 -0.100000
27.000000 -0.100000
28.000000 -0.100000
29.000000 -0.100000
30.000000 -0.100000
31.000000 -0.100000
32.000000 -0.100000
33.000000 -0.100000
34.000000 -0.100000
35.000000 -0.100000
36.000000 -0.100000
37.000000 -0.100000
38.000000 -0.100000
39.000000 -0.100000
40.000000 -0.100000
41.000000 -0.100000
42.000000 -0.100000
43.000000 -0.100000
44.000000 -0.100000
45.000000 -0.100000
46.000000 -0.100000
47.000000 -0.100000
48.000000 -0.100000
49.000000 -0.100000
50.000000 -0.100000
51.000000 -0.100000
52.000000 -0.100000
53.000000 -0.100000
54.000000 -0.100000
55.000000 -0.100000
56.000000 -0.100000
57.000000 -0.100000
58.000000 -0.100000
59.000000 -0.100000
60.000000 -0.100000
61.000000 -0.100000
62.000000 -0.100000
63.000000 -0.100000
64.000000 -0.100000
65.000000 -0.100000
66.000000 -0.100000
67.000000 -0.100000
68.000000 -0.100000
69.000000 -0.100000
70.000000 -0.100000
71.000000 -0.100000
72.000000 -0.100000
73.000000 -0.100000
74.000000 -0.100000
75.000000 -0.100000
76.000000 -0.100000
77.000000 -0.100000
78.000000 -0.100000
79.000000 -0.100000
80.000000 -0.100000
};
\addlegendentry{$u= 30$}
\addplot table {
0.000000 0.500000
1.000000 0.455000
2.000000 0.420000
3.000000 0.385000
4.000000 0.355000
5.000000 0.325000
6.000000 0.295000
7.000000 0.270000
8.000000 0.245000
9.000000 0.220000
10.000000 0.195000
11.000000 0.170000
12.000000 0.150000
13.000000 0.125000
14.000000 0.105000
15.000000 0.085000
16.000000 0.065000
17.000000 0.045000
18.000000 0.025000
19.000000 0.010000
20.000000 -0.100000
21.000000 -0.100000
22.000000 -0.100000
23.000000 -0.100000
24.000000 -0.100000
25.000000 -0.100000
26.000000 -0.100000
27.000000 -0.100000
28.000000 -0.100000
29.000000 -0.100000
30.000000 -0.100000
31.000000 -0.100000
32.000000 -0.100000
33.000000 -0.100000
34.000000 -0.100000
35.000000 -0.100000
36.000000 -0.100000
37.000000 -0.100000
38.000000 -0.100000
39.000000 -0.100000
40.000000 -0.100000
41.000000 -0.100000
42.000000 -0.100000
43.000000 -0.100000
44.000000 -0.100000
45.000000 -0.100000
46.000000 -0.100000
47.000000 -0.100000
48.000000 -0.100000
49.000000 -0.100000
50.000000 -0.100000
51.000000 -0.100000
52.000000 -0.100000
53.000000 -0.100000
54.000000 -0.100000
55.000000 -0.100000
56.000000 -0.100000
57.000000 -0.100000
58.000000 -0.100000
59.000000 -0.100000
60.000000 -0.100000
61.000000 -0.100000
62.000000 -0.100000
63.000000 -0.100000
64.000000 -0.100000
65.000000 -0.100000
66.000000 -0.100000
67.000000 -0.100000
68.000000 -0.100000
69.000000 -0.100000
70.000000 -0.100000
71.000000 -0.100000
72.000000 -0.100000
73.000000 -0.100000
74.000000 -0.100000
75.000000 -0.100000
76.000000 -0.100000
77.000000 -0.100000
78.000000 -0.100000
79.000000 -0.100000
80.000000 -0.100000
};
\addlegendentry{$u= 40$}
\addplot table {
0.000000 0.410000
1.000000 0.365000
2.000000 0.330000
3.000000 0.295000
4.000000 0.265000
5.000000 0.235000
6.000000 0.205000
7.000000 0.180000
8.000000 0.155000
9.000000 0.130000
10.000000 0.105000
11.000000 0.080000
12.000000 0.060000
13.000000 0.035000
14.000000 0.015000
15.000000 -0.100000
16.000000 -0.100000
17.000000 -0.100000
18.000000 -0.100000
19.000000 -0.100000
20.000000 -0.100000
21.000000 -0.100000
22.000000 -0.100000
23.000000 -0.100000
24.000000 -0.100000
25.000000 -0.100000
26.000000 -0.100000
27.000000 -0.100000
28.000000 -0.100000
29.000000 -0.100000
30.000000 -0.100000
31.000000 -0.100000
32.000000 -0.100000
33.000000 -0.100000
34.000000 -0.100000
35.000000 -0.100000
36.000000 -0.100000
37.000000 -0.100000
38.000000 -0.100000
39.000000 -0.100000
40.000000 -0.100000
41.000000 -0.100000
42.000000 -0.100000
43.000000 -0.100000
44.000000 -0.100000
45.000000 -0.100000
46.000000 -0.100000
47.000000 -0.100000
48.000000 -0.100000
49.000000 -0.100000
50.000000 -0.100000
51.000000 -0.100000
52.000000 -0.100000
53.000000 -0.100000
54.000000 -0.100000
55.000000 -0.100000
56.000000 -0.100000
57.000000 -0.100000
58.000000 -0.100000
59.000000 -0.100000
60.000000 -0.100000
61.000000 -0.100000
62.000000 -0.100000
63.000000 -0.100000
64.000000 -0.100000
65.000000 -0.100000
66.000000 -0.100000
67.000000 -0.100000
68.000000 -0.100000
69.000000 -0.100000
70.000000 -0.100000
71.000000 -0.100000
72.000000 -0.100000
73.000000 -0.100000
74.000000 -0.100000
75.000000 -0.100000
76.000000 -0.100000
77.000000 -0.100000
78.000000 -0.100000
79.000000 -0.100000
80.000000 -0.100000
};
\addlegendentry{$u= 50$}
\addplot table {
0.000000 0.335000
1.000000 0.290000
2.000000 0.255000
3.000000 0.220000
4.000000 0.190000
5.000000 0.160000
6.000000 0.130000
7.000000 0.105000
8.000000 0.080000
9.000000 0.055000
10.000000 0.030000
11.000000 0.005000
12.000000 -0.100000
13.000000 -0.100000
14.000000 -0.100000
15.000000 -0.100000
16.000000 -0.100000
17.000000 -0.100000
18.000000 -0.100000
19.000000 -0.100000
20.000000 -0.100000
21.000000 -0.100000
22.000000 -0.100000
23.000000 -0.100000
24.000000 -0.100000
25.000000 -0.100000
26.000000 -0.100000
27.000000 -0.100000
28.000000 -0.100000
29.000000 -0.100000
30.000000 -0.100000
31.000000 -0.100000
32.000000 -0.100000
33.000000 -0.100000
34.000000 -0.100000
35.000000 -0.100000
36.000000 -0.100000
37.000000 -0.100000
38.000000 -0.100000
39.000000 -0.100000
40.000000 -0.100000
41.000000 -0.100000
42.000000 -0.100000
43.000000 -0.100000
44.000000 -0.100000
45.000000 -0.100000
46.000000 -0.100000
47.000000 -0.100000
48.000000 -0.100000
49.000000 -0.100000
50.000000 -0.100000
51.000000 -0.100000
52.000000 -0.100000
53.000000 -0.100000
54.000000 -0.100000
55.000000 -0.100000
56.000000 -0.100000
57.000000 -0.100000
58.000000 -0.100000
59.000000 -0.100000
60.000000 -0.100000
61.000000 -0.100000
62.000000 -0.100000
63.000000 -0.100000
64.000000 -0.100000
65.000000 -0.100000
66.000000 -0.100000
67.000000 -0.100000
68.000000 -0.100000
69.000000 -0.100000
70.000000 -0.100000
71.000000 -0.100000
72.000000 -0.100000
73.000000 -0.100000
74.000000 -0.100000
75.000000 -0.100000
76.000000 -0.100000
77.000000 -0.100000
78.000000 -0.100000
79.000000 -0.100000
80.000000 -0.100000
};
\addlegendentry{$u= 60$}
\addplot table {
0.000000 0.270000
1.000000 0.225000
2.000000 0.190000
3.000000 0.155000
4.000000 0.125000
5.000000 0.095000
6.000000 0.065000
7.000000 0.040000
8.000000 0.015000
9.000000 -0.100000
10.000000 -0.100000
11.000000 -0.100000
12.000000 -0.100000
13.000000 -0.100000
14.000000 -0.100000
15.000000 -0.100000
16.000000 -0.100000
17.000000 -0.100000
18.000000 -0.100000
19.000000 -0.100000
20.000000 -0.100000
21.000000 -0.100000
22.000000 -0.100000
23.000000 -0.100000
24.000000 -0.100000
25.000000 -0.100000
26.000000 -0.100000
27.000000 -0.100000
28.000000 -0.100000
29.000000 -0.100000
30.000000 -0.100000
31.000000 -0.100000
32.000000 -0.100000
33.000000 -0.100000
34.000000 -0.100000
35.000000 -0.100000
36.000000 -0.100000
37.000000 -0.100000
38.000000 -0.100000
39.000000 -0.100000
40.000000 -0.100000
41.000000 -0.100000
42.000000 -0.100000
43.000000 -0.100000
44.000000 -0.100000
45.000000 -0.100000
46.000000 -0.100000
47.000000 -0.100000
48.000000 -0.100000
49.000000 -0.100000
50.000000 -0.100000
51.000000 -0.100000
52.000000 -0.100000
53.000000 -0.100000
54.000000 -0.100000
55.000000 -0.100000
56.000000 -0.100000
57.000000 -0.100000
58.000000 -0.100000
59.000000 -0.100000
60.000000 -0.100000
61.000000 -0.100000
62.000000 -0.100000
63.000000 -0.100000
64.000000 -0.100000
65.000000 -0.100000
66.000000 -0.100000
67.000000 -0.100000
68.000000 -0.100000
69.000000 -0.100000
70.000000 -0.100000
71.000000 -0.100000
72.000000 -0.100000
73.000000 -0.100000
74.000000 -0.100000
75.000000 -0.100000
76.000000 -0.100000
77.000000 -0.100000
78.000000 -0.100000
79.000000 -0.100000
80.000000 -0.100000
};
\addlegendentry{$u= 70$}
\addplot table {
0.000000 0.210000
1.000000 0.165000
2.000000 0.130000
3.000000 0.095000
4.000000 0.065000
5.000000 0.035000
6.000000 0.005000
7.000000 -0.100000
8.000000 -0.100000
9.000000 -0.100000
10.000000 -0.100000
11.000000 -0.100000
12.000000 -0.100000
13.000000 -0.100000
14.000000 -0.100000
15.000000 -0.100000
16.000000 -0.100000
17.000000 -0.100000
18.000000 -0.100000
19.000000 -0.100000
20.000000 -0.100000
21.000000 -0.100000
22.000000 -0.100000
23.000000 -0.100000
24.000000 -0.100000
25.000000 -0.100000
26.000000 -0.100000
27.000000 -0.100000
28.000000 -0.100000
29.000000 -0.100000
30.000000 -0.100000
31.000000 -0.100000
32.000000 -0.100000
33.000000 -0.100000
34.000000 -0.100000
35.000000 -0.100000
36.000000 -0.100000
37.000000 -0.100000
38.000000 -0.100000
39.000000 -0.100000
40.000000 -0.100000
41.000000 -0.100000
42.000000 -0.100000
43.000000 -0.100000
44.000000 -0.100000
45.000000 -0.100000
46.000000 -0.100000
47.000000 -0.100000
48.000000 -0.100000
49.000000 -0.100000
50.000000 -0.100000
51.000000 -0.100000
52.000000 -0.100000
53.000000 -0.100000
54.000000 -0.100000
55.000000 -0.100000
56.000000 -0.100000
57.000000 -0.100000
58.000000 -0.100000
59.000000 -0.100000
60.000000 -0.100000
61.000000 -0.100000
62.000000 -0.100000
63.000000 -0.100000
64.000000 -0.100000
65.000000 -0.100000
66.000000 -0.100000
67.000000 -0.100000
68.000000 -0.100000
69.000000 -0.100000
70.000000 -0.100000
71.000000 -0.100000
72.000000 -0.100000
73.000000 -0.100000
74.000000 -0.100000
75.000000 -0.100000
76.000000 -0.100000
77.000000 -0.100000
78.000000 -0.100000
79.000000 -0.100000
80.000000 -0.100000
};
\addlegendentry{$u= 80$}
\addplot table {
0.000000 0.155000
1.000000 0.110000
2.000000 0.075000
3.000000 0.040000
4.000000 0.010000
5.000000 -0.100000
6.000000 -0.100000
7.000000 -0.100000
8.000000 -0.100000
9.000000 -0.100000
10.000000 -0.100000
11.000000 -0.100000
12.000000 -0.100000
13.000000 -0.100000
14.000000 -0.100000
15.000000 -0.100000
16.000000 -0.100000
17.000000 -0.100000
18.000000 -0.100000
19.000000 -0.100000
20.000000 -0.100000
21.000000 -0.100000
22.000000 -0.100000
23.000000 -0.100000
24.000000 -0.100000
25.000000 -0.100000
26.000000 -0.100000
27.000000 -0.100000
28.000000 -0.100000
29.000000 -0.100000
30.000000 -0.100000
31.000000 -0.100000
32.000000 -0.100000
33.000000 -0.100000
34.000000 -0.100000
35.000000 -0.100000
36.000000 -0.100000
37.000000 -0.100000
38.000000 -0.100000
39.000000 -0.100000
40.000000 -0.100000
41.000000 -0.100000
42.000000 -0.100000
43.000000 -0.100000
44.000000 -0.100000
45.000000 -0.100000
46.000000 -0.100000
47.000000 -0.100000
48.000000 -0.100000
49.000000 -0.100000
50.000000 -0.100000
51.000000 -0.100000
52.000000 -0.100000
53.000000 -0.100000
54.000000 -0.100000
55.000000 -0.100000
56.000000 -0.100000
57.000000 -0.100000
58.000000 -0.100000
59.000000 -0.100000
60.000000 -0.100000
61.000000 -0.100000
62.000000 -0.100000
63.000000 -0.100000
64.000000 -0.100000
65.000000 -0.100000
66.000000 -0.100000
67.000000 -0.100000
68.000000 -0.100000
69.000000 -0.100000
70.000000 -0.100000
71.000000 -0.100000
72.000000 -0.100000
73.000000 -0.100000
74.000000 -0.100000
75.000000 -0.100000
76.000000 -0.100000
77.000000 -0.100000
78.000000 -0.100000
79.000000 -0.100000
80.000000 -0.100000
};
\addlegendentry{$u= 90$}
\addplot table {
0.000000 0.110000
1.000000 0.065000
2.000000 0.030000
3.000000 -0.100000
4.000000 -0.100000
5.000000 -0.100000
6.000000 -0.100000
7.000000 -0.100000
8.000000 -0.100000
9.000000 -0.100000
10.000000 -0.100000
11.000000 -0.100000
12.000000 -0.100000
13.000000 -0.100000
14.000000 -0.100000
15.000000 -0.100000
16.000000 -0.100000
17.000000 -0.100000
18.000000 -0.100000
19.000000 -0.100000
20.000000 -0.100000
21.000000 -0.100000
22.000000 -0.100000
23.000000 -0.100000
24.000000 -0.100000
25.000000 -0.100000
26.000000 -0.100000
27.000000 -0.100000
28.000000 -0.100000
29.000000 -0.100000
30.000000 -0.100000
31.000000 -0.100000
32.000000 -0.100000
33.000000 -0.100000
34.000000 -0.100000
35.000000 -0.100000
36.000000 -0.100000
37.000000 -0.100000
38.000000 -0.100000
39.000000 -0.100000
40.000000 -0.100000
41.000000 -0.100000
42.000000 -0.100000
43.000000 -0.100000
44.000000 -0.100000
45.000000 -0.100000
46.000000 -0.100000
47.000000 -0.100000
48.000000 -0.100000
49.000000 -0.100000
50.000000 -0.100000
51.000000 -0.100000
52.000000 -0.100000
53.000000 -0.100000
54.000000 -0.100000
55.000000 -0.100000
56.000000 -0.100000
57.000000 -0.100000
58.000000 -0.100000
59.000000 -0.100000
60.000000 -0.100000
61.000000 -0.100000
62.000000 -0.100000
63.000000 -0.100000
64.000000 -0.100000
65.000000 -0.100000
66.000000 -0.100000
67.000000 -0.100000
68.000000 -0.100000
69.000000 -0.100000
70.000000 -0.100000
71.000000 -0.100000
72.000000 -0.100000
73.000000 -0.100000
74.000000 -0.100000
75.000000 -0.100000
76.000000 -0.100000
77.000000 -0.100000
78.000000 -0.100000
79.000000 -0.100000
80.000000 -0.100000
};
\addlegendentry{$u= 100$}

%% file: figures/comparisons/GV_const2_q4/ordGVq3.tex
\addplot table {
	0.000000 0.788519
	1.000000 0.748895
	2.000000 0.717196
	3.000000 0.685496
	4.000000 0.657759
	5.000000 0.630023
	6.000000 0.606248
	7.000000 0.578511
	8.000000 0.558699
	9.000000 0.534925
	10.000000 0.515113
	11.000000 0.495301
	12.000000 0.475489
	13.000000 0.455677
	14.000000 0.435865
	15.000000 0.420015
	16.000000 0.404165
	17.000000 0.384353
	18.000000 0.368504
	19.000000 0.352654
	20.000000 0.336805
	21.000000 0.324917
	22.000000 0.309068
	23.000000 0.297180
	24.000000 0.281331
	25.000000 0.269444
	26.000000 0.257556
	27.000000 0.245669
	28.000000 0.233782
	29.000000 0.221895
	30.000000 0.210008
	31.000000 0.198120
	32.000000 0.186233
	33.000000 0.178308
	34.000000 0.166421
	35.000000 0.158496
	36.000000 0.146609
	37.000000 0.138684
	38.000000 0.130759
	39.000000 0.122835
	40.000000 0.114910
	41.000000 0.106985
	42.000000 0.099060
	43.000000 0.091135
	44.000000 0.083211
	45.000000 0.079248
	46.000000 0.071323
	47.000000 0.063399
	48.000000 0.059436
	49.000000 0.055474
	50.000000 0.047549
	51.000000 0.043586
	52.000000 0.039624
	53.000000 0.035662
	54.000000 0.031699
	55.000000 0.027737
	56.000000 0.023774
	57.000000 0.019812
	58.000000 0.015850
	59.000000 0.011887
	60.000000 0.011887
	61.000000 0.007925
	62.000000 0.007925
	63.000000 0.003962
	64.000000 0.003962
	65.000000 0.003962
	66.000000 0.000000
	67.000000 0.000000
	68.000000 0.000000
	69.000000 0.000000
	70.000000 0.000000
	71.000000 0.000000
	72.000000 0.000000
	73.000000 0.000000
	74.000000 0.000000
	75.000000 0.000000
	76.000000 0.000000
	77.000000 0.000000
	78.000000 0.000000
	79.000000 0.000000
	80.000000 0.000000
};
\addlegendentry{$(q-1)$-ary GV \\ 	($u$ arbitrary)}

%% file: figures/comparisons/GV_const2_q4/plot_list_constbin_q4.tex
\addplot table {
0.000000 0.992500
1.000000 0.952500
2.000000 0.912500
3.000000 0.882500
4.000000 0.847500
5.000000 0.817500
6.000000 0.792500
7.000000 0.762500
8.000000 0.737500
9.000000 0.712500
10.000000 0.687500
11.000000 0.667500
12.000000 0.642500
13.000000 0.622500
14.000000 0.602500
15.000000 0.582500
16.000000 0.562500
17.000000 0.542500
18.000000 0.522500
19.000000 0.502500
20.000000 0.482500
21.000000 0.467500
22.000000 0.447500
23.000000 0.432500
24.000000 0.417500
25.000000 0.402500
26.000000 0.387500
27.000000 0.367500
28.000000 0.352500
29.000000 0.342500
30.000000 0.327500
31.000000 0.312500
32.000000 0.297500
33.000000 0.287500
34.000000 0.272500
35.000000 0.262500
36.000000 0.247500
37.000000 0.237500
38.000000 0.222500
39.000000 0.212500
40.000000 0.202500
41.000000 0.192500
42.000000 0.182500
43.000000 0.172500
44.000000 0.162500
45.000000 0.152500
46.000000 0.142500
47.000000 0.132500
48.000000 -0.100000
49.000000 -0.100000
50.000000 -0.100000
51.000000 -0.100000
52.000000 -0.100000
53.000000 -0.100000
54.000000 -0.100000
55.000000 -0.100000
56.000000 -0.100000
57.000000 -0.100000
58.000000 -0.100000
59.000000 -0.100000
60.000000 -0.100000
61.000000 -0.100000
62.000000 -0.100000
63.000000 -0.100000
64.000000 -0.100000
65.000000 -0.100000
66.000000 -0.100000
67.000000 -0.100000
68.000000 -0.100000
69.000000 -0.100000
70.000000 -0.100000
71.000000 -0.100000
72.000000 -0.100000
73.000000 -0.100000
74.000000 -0.100000
75.000000 -0.100000
76.000000 -0.100000
77.000000 -0.100000
78.000000 -0.100000
79.000000 -0.100000
80.000000 -0.100000
};
\addplot table {
0.000000 0.922500
1.000000 0.882500
2.000000 0.842500
3.000000 0.812500
4.000000 0.777500
5.000000 0.747500
6.000000 0.722500
7.000000 0.692500
8.000000 0.667500
9.000000 0.642500
10.000000 0.617500
11.000000 0.597500
12.000000 0.572500
13.000000 0.552500
14.000000 0.532500
15.000000 0.512500
16.000000 0.492500
17.000000 0.472500
18.000000 0.452500
19.000000 0.432500
20.000000 0.412500
21.000000 0.397500
22.000000 0.377500
23.000000 0.362500
24.000000 0.347500
25.000000 0.332500
26.000000 0.317500
27.000000 0.297500
28.000000 0.282500
29.000000 -0.100000
30.000000 -0.100000
31.000000 -0.100000
32.000000 -0.100000
33.000000 -0.100000
34.000000 -0.100000
35.000000 -0.100000
36.000000 -0.100000
37.000000 -0.100000
38.000000 -0.100000
39.000000 -0.100000
40.000000 -0.100000
41.000000 -0.100000
42.000000 -0.100000
43.000000 -0.100000
44.000000 -0.100000
45.000000 -0.100000
46.000000 -0.100000
47.000000 -0.100000
48.000000 -0.100000
49.000000 -0.100000
50.000000 -0.100000
51.000000 -0.100000
52.000000 -0.100000
53.000000 -0.100000
54.000000 -0.100000
55.000000 -0.100000
56.000000 -0.100000
57.000000 -0.100000
58.000000 -0.100000
59.000000 -0.100000
60.000000 -0.100000
61.000000 -0.100000
62.000000 -0.100000
63.000000 -0.100000
64.000000 -0.100000
65.000000 -0.100000
66.000000 -0.100000
67.000000 -0.100000
68.000000 -0.100000
69.000000 -0.100000
70.000000 -0.100000
71.000000 -0.100000
72.000000 -0.100000
73.000000 -0.100000
74.000000 -0.100000
75.000000 -0.100000
76.000000 -0.100000
77.000000 -0.100000
78.000000 -0.100000
79.000000 -0.100000
80.000000 -0.100000
};
\addplot table {
0.000000 0.860000
1.000000 0.820000
2.000000 0.780000
3.000000 0.750000
4.000000 0.715000
5.000000 0.685000
6.000000 0.660000
7.000000 0.630000
8.000000 0.605000
9.000000 0.580000
10.000000 0.555000
11.000000 0.535000
12.000000 0.510000
13.000000 0.490000
14.000000 0.470000
15.000000 0.450000
16.000000 0.430000
17.000000 0.410000
18.000000 0.390000
19.000000 0.370000
20.000000 0.350000
21.000000 0.335000
22.000000 -0.100000
23.000000 -0.100000
24.000000 -0.100000
25.000000 -0.100000
26.000000 -0.100000
27.000000 -0.100000
28.000000 -0.100000
29.000000 -0.100000
30.000000 -0.100000
31.000000 -0.100000
32.000000 -0.100000
33.000000 -0.100000
34.000000 -0.100000
35.000000 -0.100000
36.000000 -0.100000
37.000000 -0.100000
38.000000 -0.100000
39.000000 -0.100000
40.000000 -0.100000
41.000000 -0.100000
42.000000 -0.100000
43.000000 -0.100000
44.000000 -0.100000
45.000000 -0.100000
46.000000 -0.100000
47.000000 -0.100000
48.000000 -0.100000
49.000000 -0.100000
50.000000 -0.100000
51.000000 -0.100000
52.000000 -0.100000
53.000000 -0.100000
54.000000 -0.100000
55.000000 -0.100000
56.000000 -0.100000
57.000000 -0.100000
58.000000 -0.100000
59.000000 -0.100000
60.000000 -0.100000
61.000000 -0.100000
62.000000 -0.100000
63.000000 -0.100000
64.000000 -0.100000
65.000000 -0.100000
66.000000 -0.100000
67.000000 -0.100000
68.000000 -0.100000
69.000000 -0.100000
70.000000 -0.100000
71.000000 -0.100000
72.000000 -0.100000
73.000000 -0.100000
74.000000 -0.100000
75.000000 -0.100000
76.000000 -0.100000
77.000000 -0.100000
78.000000 -0.100000
79.000000 -0.100000
80.000000 -0.100000
};
\addplot table {
0.000000 0.810000
1.000000 0.770000
2.000000 0.730000
3.000000 0.700000
4.000000 0.665000
5.000000 0.635000
6.000000 0.610000
7.000000 0.580000
8.000000 0.555000
9.000000 0.530000
10.000000 0.505000
11.000000 0.485000
12.000000 0.460000
13.000000 0.440000
14.000000 0.420000
15.000000 0.400000
16.000000 0.380000
17.000000 -0.100000
18.000000 -0.100000
19.000000 -0.100000
20.000000 -0.100000
21.000000 -0.100000
22.000000 -0.100000
23.000000 -0.100000
24.000000 -0.100000
25.000000 -0.100000
26.000000 -0.100000
27.000000 -0.100000
28.000000 -0.100000
29.000000 -0.100000
30.000000 -0.100000
31.000000 -0.100000
32.000000 -0.100000
33.000000 -0.100000
34.000000 -0.100000
35.000000 -0.100000
36.000000 -0.100000
37.000000 -0.100000
38.000000 -0.100000
39.000000 -0.100000
40.000000 -0.100000
41.000000 -0.100000
42.000000 -0.100000
43.000000 -0.100000
44.000000 -0.100000
45.000000 -0.100000
46.000000 -0.100000
47.000000 -0.100000
48.000000 -0.100000
49.000000 -0.100000
50.000000 -0.100000
51.000000 -0.100000
52.000000 -0.100000
53.000000 -0.100000
54.000000 -0.100000
55.000000 -0.100000
56.000000 -0.100000
57.000000 -0.100000
58.000000 -0.100000
59.000000 -0.100000
60.000000 -0.100000
61.000000 -0.100000
62.000000 -0.100000
63.000000 -0.100000
64.000000 -0.100000
65.000000 -0.100000
66.000000 -0.100000
67.000000 -0.100000
68.000000 -0.100000
69.000000 -0.100000
70.000000 -0.100000
71.000000 -0.100000
72.000000 -0.100000
73.000000 -0.100000
74.000000 -0.100000
75.000000 -0.100000
76.000000 -0.100000
77.000000 -0.100000
78.000000 -0.100000
79.000000 -0.100000
80.000000 -0.100000
};
\addplot table {
0.000000 0.767500
1.000000 0.727500
2.000000 0.687500
3.000000 0.657500
4.000000 0.622500
5.000000 0.592500
6.000000 0.567500
7.000000 0.537500
8.000000 0.512500
9.000000 0.487500
10.000000 0.462500
11.000000 0.442500
12.000000 0.417500
13.000000 -0.100000
14.000000 -0.100000
15.000000 -0.100000
16.000000 -0.100000
17.000000 -0.100000
18.000000 -0.100000
19.000000 -0.100000
20.000000 -0.100000
21.000000 -0.100000
22.000000 -0.100000
23.000000 -0.100000
24.000000 -0.100000
25.000000 -0.100000
26.000000 -0.100000
27.000000 -0.100000
28.000000 -0.100000
29.000000 -0.100000
30.000000 -0.100000
31.000000 -0.100000
32.000000 -0.100000
33.000000 -0.100000
34.000000 -0.100000
35.000000 -0.100000
36.000000 -0.100000
37.000000 -0.100000
38.000000 -0.100000
39.000000 -0.100000
40.000000 -0.100000
41.000000 -0.100000
42.000000 -0.100000
43.000000 -0.100000
44.000000 -0.100000
45.000000 -0.100000
46.000000 -0.100000
47.000000 -0.100000
48.000000 -0.100000
49.000000 -0.100000
50.000000 -0.100000
51.000000 -0.100000
52.000000 -0.100000
53.000000 -0.100000
54.000000 -0.100000
55.000000 -0.100000
56.000000 -0.100000
57.000000 -0.100000
58.000000 -0.100000
59.000000 -0.100000
60.000000 -0.100000
61.000000 -0.100000
62.000000 -0.100000
63.000000 -0.100000
64.000000 -0.100000
65.000000 -0.100000
66.000000 -0.100000
67.000000 -0.100000
68.000000 -0.100000
69.000000 -0.100000
70.000000 -0.100000
71.000000 -0.100000
72.000000 -0.100000
73.000000 -0.100000
74.000000 -0.100000
75.000000 -0.100000
76.000000 -0.100000
77.000000 -0.100000
78.000000 -0.100000
79.000000 -0.100000
80.000000 -0.100000
};
\addplot table {
0.000000 0.730000
1.000000 0.690000
2.000000 0.650000
3.000000 0.620000
4.000000 0.585000
5.000000 0.555000
6.000000 0.530000
7.000000 0.500000
8.000000 0.475000
9.000000 0.450000
10.000000 0.425000
11.000000 -0.100000
12.000000 -0.100000
13.000000 -0.100000
14.000000 -0.100000
15.000000 -0.100000
16.000000 -0.100000
17.000000 -0.100000
18.000000 -0.100000
19.000000 -0.100000
20.000000 -0.100000
21.000000 -0.100000
22.000000 -0.100000
23.000000 -0.100000
24.000000 -0.100000
25.000000 -0.100000
26.000000 -0.100000
27.000000 -0.100000
28.000000 -0.100000
29.000000 -0.100000
30.000000 -0.100000
31.000000 -0.100000
32.000000 -0.100000
33.000000 -0.100000
34.000000 -0.100000
35.000000 -0.100000
36.000000 -0.100000
37.000000 -0.100000
38.000000 -0.100000
39.000000 -0.100000
40.000000 -0.100000
41.000000 -0.100000
42.000000 -0.100000
43.000000 -0.100000
44.000000 -0.100000
45.000000 -0.100000
46.000000 -0.100000
47.000000 -0.100000
48.000000 -0.100000
49.000000 -0.100000
50.000000 -0.100000
51.000000 -0.100000
52.000000 -0.100000
53.000000 -0.100000
54.000000 -0.100000
55.000000 -0.100000
56.000000 -0.100000
57.000000 -0.100000
58.000000 -0.100000
59.000000 -0.100000
60.000000 -0.100000
61.000000 -0.100000
62.000000 -0.100000
63.000000 -0.100000
64.000000 -0.100000
65.000000 -0.100000
66.000000 -0.100000
67.000000 -0.100000
68.000000 -0.100000
69.000000 -0.100000
70.000000 -0.100000
71.000000 -0.100000
72.000000 -0.100000
73.000000 -0.100000
74.000000 -0.100000
75.000000 -0.100000
76.000000 -0.100000
77.000000 -0.100000
78.000000 -0.100000
79.000000 -0.100000
80.000000 -0.100000
};
\addplot table {
0.000000 0.695000
1.000000 0.655000
2.000000 0.615000
3.000000 0.585000
4.000000 0.550000
5.000000 0.520000
6.000000 0.495000
7.000000 0.465000
8.000000 0.440000
9.000000 -0.100000
10.000000 -0.100000
11.000000 -0.100000
12.000000 -0.100000
13.000000 -0.100000
14.000000 -0.100000
15.000000 -0.100000
16.000000 -0.100000
17.000000 -0.100000
18.000000 -0.100000
19.000000 -0.100000
20.000000 -0.100000
21.000000 -0.100000
22.000000 -0.100000
23.000000 -0.100000
24.000000 -0.100000
25.000000 -0.100000
26.000000 -0.100000
27.000000 -0.100000
28.000000 -0.100000
29.000000 -0.100000
30.000000 -0.100000
31.000000 -0.100000
32.000000 -0.100000
33.000000 -0.100000
34.000000 -0.100000
35.000000 -0.100000
36.000000 -0.100000
37.000000 -0.100000
38.000000 -0.100000
39.000000 -0.100000
40.000000 -0.100000
41.000000 -0.100000
42.000000 -0.100000
43.000000 -0.100000
44.000000 -0.100000
45.000000 -0.100000
46.000000 -0.100000
47.000000 -0.100000
48.000000 -0.100000
49.000000 -0.100000
50.000000 -0.100000
51.000000 -0.100000
52.000000 -0.100000
53.000000 -0.100000
54.000000 -0.100000
55.000000 -0.100000
56.000000 -0.100000
57.000000 -0.100000
58.000000 -0.100000
59.000000 -0.100000
60.000000 -0.100000
61.000000 -0.100000
62.000000 -0.100000
63.000000 -0.100000
64.000000 -0.100000
65.000000 -0.100000
66.000000 -0.100000
67.000000 -0.100000
68.000000 -0.100000
69.000000 -0.100000
70.000000 -0.100000
71.000000 -0.100000
72.000000 -0.100000
73.000000 -0.100000
74.000000 -0.100000
75.000000 -0.100000
76.000000 -0.100000
77.000000 -0.100000
78.000000 -0.100000
79.000000 -0.100000
80.000000 -0.100000
};
\addplot table {
0.000000 0.665000
1.000000 0.625000
2.000000 0.585000
3.000000 0.555000
4.000000 0.520000
5.000000 0.490000
6.000000 0.465000
7.000000 -0.100000
8.000000 -0.100000
9.000000 -0.100000
10.000000 -0.100000
11.000000 -0.100000
12.000000 -0.100000
13.000000 -0.100000
14.000000 -0.100000
15.000000 -0.100000
16.000000 -0.100000
17.000000 -0.100000
18.000000 -0.100000
19.000000 -0.100000
20.000000 -0.100000
21.000000 -0.100000
22.000000 -0.100000
23.000000 -0.100000
24.000000 -0.100000
25.000000 -0.100000
26.000000 -0.100000
27.000000 -0.100000
28.000000 -0.100000
29.000000 -0.100000
30.000000 -0.100000
31.000000 -0.100000
32.000000 -0.100000
33.000000 -0.100000
34.000000 -0.100000
35.000000 -0.100000
36.000000 -0.100000
37.000000 -0.100000
38.000000 -0.100000
39.000000 -0.100000
40.000000 -0.100000
41.000000 -0.100000
42.000000 -0.100000
43.000000 -0.100000
44.000000 -0.100000
45.000000 -0.100000
46.000000 -0.100000
47.000000 -0.100000
48.000000 -0.100000
49.000000 -0.100000
50.000000 -0.100000
51.000000 -0.100000
52.000000 -0.100000
53.000000 -0.100000
54.000000 -0.100000
55.000000 -0.100000
56.000000 -0.100000
57.000000 -0.100000
58.000000 -0.100000
59.000000 -0.100000
60.000000 -0.100000
61.000000 -0.100000
62.000000 -0.100000
63.000000 -0.100000
64.000000 -0.100000
65.000000 -0.100000
66.000000 -0.100000
67.000000 -0.100000
68.000000 -0.100000
69.000000 -0.100000
70.000000 -0.100000
71.000000 -0.100000
72.000000 -0.100000
73.000000 -0.100000
74.000000 -0.100000
75.000000 -0.100000
76.000000 -0.100000
77.000000 -0.100000
78.000000 -0.100000
79.000000 -0.100000
80.000000 -0.100000
};
\addplot table {
0.000000 0.640000
1.000000 0.600000
2.000000 0.560000
3.000000 0.530000
4.000000 0.495000
5.000000 0.465000
6.000000 -0.100000
7.000000 -0.100000
8.000000 -0.100000
9.000000 -0.100000
10.000000 -0.100000
11.000000 -0.100000
12.000000 -0.100000
13.000000 -0.100000
14.000000 -0.100000
15.000000 -0.100000
16.000000 -0.100000
17.000000 -0.100000
18.000000 -0.100000
19.000000 -0.100000
20.000000 -0.100000
21.000000 -0.100000
22.000000 -0.100000
23.000000 -0.100000
24.000000 -0.100000
25.000000 -0.100000
26.000000 -0.100000
27.000000 -0.100000
28.000000 -0.100000
29.000000 -0.100000
30.000000 -0.100000
31.000000 -0.100000
32.000000 -0.100000
33.000000 -0.100000
34.000000 -0.100000
35.000000 -0.100000
36.000000 -0.100000
37.000000 -0.100000
38.000000 -0.100000
39.000000 -0.100000
40.000000 -0.100000
41.000000 -0.100000
42.000000 -0.100000
43.000000 -0.100000
44.000000 -0.100000
45.000000 -0.100000
46.000000 -0.100000
47.000000 -0.100000
48.000000 -0.100000
49.000000 -0.100000
50.000000 -0.100000
51.000000 -0.100000
52.000000 -0.100000
53.000000 -0.100000
54.000000 -0.100000
55.000000 -0.100000
56.000000 -0.100000
57.000000 -0.100000
58.000000 -0.100000
59.000000 -0.100000
60.000000 -0.100000
61.000000 -0.100000
62.000000 -0.100000
63.000000 -0.100000
64.000000 -0.100000
65.000000 -0.100000
66.000000 -0.100000
67.000000 -0.100000
68.000000 -0.100000
69.000000 -0.100000
70.000000 -0.100000
71.000000 -0.100000
72.000000 -0.100000
73.000000 -0.100000
74.000000 -0.100000
75.000000 -0.100000
76.000000 -0.100000
77.000000 -0.100000
78.000000 -0.100000
79.000000 -0.100000
80.000000 -0.100000
};
\addplot table {
0.000000 0.615000
1.000000 0.575000
2.000000 0.535000
3.000000 0.505000
4.000000 0.470000
5.000000 -0.100000
6.000000 -0.100000
7.000000 -0.100000
8.000000 -0.100000
9.000000 -0.100000
10.000000 -0.100000
11.000000 -0.100000
12.000000 -0.100000
13.000000 -0.100000
14.000000 -0.100000
15.000000 -0.100000
16.000000 -0.100000
17.000000 -0.100000
18.000000 -0.100000
19.000000 -0.100000
20.000000 -0.100000
21.000000 -0.100000
22.000000 -0.100000
23.000000 -0.100000
24.000000 -0.100000
25.000000 -0.100000
26.000000 -0.100000
27.000000 -0.100000
28.000000 -0.100000
29.000000 -0.100000
30.000000 -0.100000
31.000000 -0.100000
32.000000 -0.100000
33.000000 -0.100000
34.000000 -0.100000
35.000000 -0.100000
36.000000 -0.100000
37.000000 -0.100000
38.000000 -0.100000
39.000000 -0.100000
40.000000 -0.100000
41.000000 -0.100000
42.000000 -0.100000
43.000000 -0.100000
44.000000 -0.100000
45.000000 -0.100000
46.000000 -0.100000
47.000000 -0.100000
48.000000 -0.100000
49.000000 -0.100000
50.000000 -0.100000
51.000000 -0.100000
52.000000 -0.100000
53.000000 -0.100000
54.000000 -0.100000
55.000000 -0.100000
56.000000 -0.100000
57.000000 -0.100000
58.000000 -0.100000
59.000000 -0.100000
60.000000 -0.100000
61.000000 -0.100000
62.000000 -0.100000
63.000000 -0.100000
64.000000 -0.100000
65.000000 -0.100000
66.000000 -0.100000
67.000000 -0.100000
68.000000 -0.100000
69.000000 -0.100000
70.000000 -0.100000
71.000000 -0.100000
72.000000 -0.100000
73.000000 -0.100000
74.000000 -0.100000
75.000000 -0.100000
76.000000 -0.100000
77.000000 -0.100000
78.000000 -0.100000
79.000000 -0.100000
80.000000 -0.100000
};
\addplot table {
0.000000 0.592500
1.000000 0.552500
2.000000 0.512500
3.000000 0.482500
4.000000 -0.100000
5.000000 -0.100000
6.000000 -0.100000
7.000000 -0.100000
8.000000 -0.100000
9.000000 -0.100000
10.000000 -0.100000
11.000000 -0.100000
12.000000 -0.100000
13.000000 -0.100000
14.000000 -0.100000
15.000000 -0.100000
16.000000 -0.100000
17.000000 -0.100000
18.000000 -0.100000
19.000000 -0.100000
20.000000 -0.100000
21.000000 -0.100000
22.000000 -0.100000
23.000000 -0.100000
24.000000 -0.100000
25.000000 -0.100000
26.000000 -0.100000
27.000000 -0.100000
28.000000 -0.100000
29.000000 -0.100000
30.000000 -0.100000
31.000000 -0.100000
32.000000 -0.100000
33.000000 -0.100000
34.000000 -0.100000
35.000000 -0.100000
36.000000 -0.100000
37.000000 -0.100000
38.000000 -0.100000
39.000000 -0.100000
40.000000 -0.100000
41.000000 -0.100000
42.000000 -0.100000
43.000000 -0.100000
44.000000 -0.100000
45.000000 -0.100000
46.000000 -0.100000
47.000000 -0.100000
48.000000 -0.100000
49.000000 -0.100000
50.000000 -0.100000
51.000000 -0.100000
52.000000 -0.100000
53.000000 -0.100000
54.000000 -0.100000
55.000000 -0.100000
56.000000 -0.100000
57.000000 -0.100000
58.000000 -0.100000
59.000000 -0.100000
60.000000 -0.100000
61.000000 -0.100000
62.000000 -0.100000
63.000000 -0.100000
64.000000 -0.100000
65.000000 -0.100000
66.000000 -0.100000
67.000000 -0.100000
68.000000 -0.100000
69.000000 -0.100000
70.000000 -0.100000
71.000000 -0.100000
72.000000 -0.100000
73.000000 -0.100000
74.000000 -0.100000
75.000000 -0.100000
76.000000 -0.100000
77.000000 -0.100000
78.000000 -0.100000
79.000000 -0.100000
80.000000 -0.100000
};

%% file: figures/comparisons/Const1_to_ord_GV/plot_list_construction_1.tex
\addplot table {
0.000000 0.990000
1.000000 0.960000
2.000000 0.930000
3.000000 0.905000
4.000000 0.880000
5.000000 0.855000
6.000000 0.830000
7.000000 0.810000
8.000000 0.790000
9.000000 0.770000
10.000000 0.750000
11.000000 0.730000
12.000000 0.710000
13.000000 0.690000
14.000000 0.675000
15.000000 0.655000
16.000000 0.635000
17.000000 0.620000
18.000000 0.605000
19.000000 0.585000
20.000000 0.570000
21.000000 0.555000
22.000000 0.540000
23.000000 0.525000
24.000000 0.510000
25.000000 0.495000
26.000000 0.480000
27.000000 0.465000
28.000000 0.450000
29.000000 0.440000
30.000000 0.425000
31.000000 0.410000
32.000000 0.400000
33.000000 0.385000
34.000000 0.375000
35.000000 0.360000
36.000000 0.350000
37.000000 0.335000
38.000000 0.325000
39.000000 0.315000
40.000000 0.300000
41.000000 0.290000
42.000000 0.280000
43.000000 0.270000
44.000000 0.260000
45.000000 0.245000
46.000000 0.235000
47.000000 0.225000
48.000000 0.215000
49.000000 0.205000
50.000000 0.200000
51.000000 0.190000
52.000000 0.180000
53.000000 0.170000
54.000000 0.160000
55.000000 0.155000
56.000000 0.145000
57.000000 0.135000
58.000000 0.130000
59.000000 0.120000
60.000000 0.115000
61.000000 0.105000
62.000000 0.100000
63.000000 0.090000
64.000000 0.085000
65.000000 0.080000
66.000000 0.075000
67.000000 0.065000
68.000000 0.060000
69.000000 0.055000
70.000000 0.050000
71.000000 0.045000
72.000000 0.040000
73.000000 0.035000
74.000000 0.030000
75.000000 0.025000
76.000000 0.020000
77.000000 0.020000
78.000000 0.015000
79.000000 0.010000
80.000000 0.010000
};
\addlegendentry{Theorem~\ref{th:upartial}: $u \leq 7$}

%% file: figures/comparisons/Const1_to_ord_GV/GV_treat_u_as_t.tex
\addplot table {
0.000000 0.960000
1.000000 0.930000
2.000000 0.905000
3.000000 0.880000
4.000000 0.855000
5.000000 0.830000
6.000000 0.810000
7.000000 0.790000
8.000000 0.770000
9.000000 0.750000
10.000000 0.730000
11.000000 0.710000
12.000000 0.690000
13.000000 0.675000
14.000000 0.655000
15.000000 0.635000
16.000000 0.620000
17.000000 0.605000
18.000000 0.585000
19.000000 0.570000
20.000000 0.555000
21.000000 0.540000
22.000000 0.525000
23.000000 0.510000
24.000000 0.495000
25.000000 0.480000
26.000000 0.465000
27.000000 0.450000
28.000000 0.440000
29.000000 0.425000
30.000000 0.410000
31.000000 0.400000
32.000000 0.385000
33.000000 0.375000
34.000000 0.360000
35.000000 0.350000
36.000000 0.335000
37.000000 0.325000
38.000000 0.315000
39.000000 0.300000
40.000000 0.290000
41.000000 0.280000
42.000000 0.270000
43.000000 0.260000
44.000000 0.245000
45.000000 0.235000
46.000000 0.225000
47.000000 0.215000
48.000000 0.205000
49.000000 0.200000
50.000000 0.190000
51.000000 0.180000
52.000000 0.170000
53.000000 0.160000
54.000000 0.155000
55.000000 0.145000
56.000000 0.135000
57.000000 0.130000
58.000000 0.120000
59.000000 0.115000
60.000000 0.105000
61.000000 0.100000
62.000000 0.090000
63.000000 0.085000
64.000000 0.080000
65.000000 0.075000
66.000000 0.065000
67.000000 0.060000
68.000000 0.055000
69.000000 0.050000
70.000000 0.045000
71.000000 0.040000
72.000000 0.035000
73.000000 0.030000
74.000000 0.025000
75.000000 0.020000
76.000000 0.020000
77.000000 0.015000
78.000000 0.010000
79.000000 0.010000
80.000000 0.005000
};
\addlegendentry{Theorem~\ref{th:upartial}: $u= 10$}
\addplot table {
0.000000 0.930000
1.000000 0.905000
2.000000 0.880000
3.000000 0.855000
4.000000 0.830000
5.000000 0.810000
6.000000 0.790000
7.000000 0.770000
8.000000 0.750000
9.000000 0.730000
10.000000 0.710000
11.000000 0.690000
12.000000 0.675000
13.000000 0.655000
14.000000 0.635000
15.000000 0.620000
16.000000 0.605000
17.000000 0.585000
18.000000 0.570000
19.000000 0.555000
20.000000 0.540000
21.000000 0.525000
22.000000 0.510000
23.000000 0.495000
24.000000 0.480000
25.000000 0.465000
26.000000 0.450000
27.000000 0.440000
28.000000 0.425000
29.000000 0.410000
30.000000 0.400000
31.000000 0.385000
32.000000 0.375000
33.000000 0.360000
34.000000 0.350000
35.000000 0.335000
36.000000 0.325000
37.000000 0.315000
38.000000 0.300000
39.000000 0.290000
40.000000 0.280000
41.000000 0.270000
42.000000 0.260000
43.000000 0.245000
44.000000 0.235000
45.000000 0.225000
46.000000 0.215000
47.000000 0.205000
48.000000 0.200000
49.000000 0.190000
50.000000 0.180000
51.000000 0.170000
52.000000 0.160000
53.000000 0.155000
54.000000 0.145000
55.000000 0.135000
56.000000 0.130000
57.000000 0.120000
58.000000 0.115000
59.000000 0.105000
60.000000 0.100000
61.000000 0.090000
62.000000 0.085000
63.000000 0.080000
64.000000 0.075000
65.000000 0.065000
66.000000 0.060000
67.000000 0.055000
68.000000 0.050000
69.000000 0.045000
70.000000 0.040000
71.000000 0.035000
72.000000 0.030000
73.000000 0.025000
74.000000 0.020000
75.000000 0.020000
76.000000 0.015000
77.000000 0.010000
78.000000 0.010000
79.000000 0.005000
80.000000 0.005000
};
\addlegendentry{Theorem~\ref{th:upartial}: $u= 20$}
\addplot table {
0.000000 0.905000
1.000000 0.880000
2.000000 0.855000
3.000000 0.830000
4.000000 0.810000
5.000000 0.790000
6.000000 0.770000
7.000000 0.750000
8.000000 0.730000
9.000000 0.710000
10.000000 0.690000
11.000000 0.675000
12.000000 0.655000
13.000000 0.635000
14.000000 0.620000
15.000000 0.605000
16.000000 0.585000
17.000000 0.570000
18.000000 0.555000
19.000000 0.540000
20.000000 0.525000
21.000000 0.510000
22.000000 0.495000
23.000000 0.480000
24.000000 0.465000
25.000000 0.450000
26.000000 0.440000
27.000000 0.425000
28.000000 0.410000
29.000000 0.400000
30.000000 0.385000
31.000000 0.375000
32.000000 0.360000
33.000000 0.350000
34.000000 0.335000
35.000000 0.325000
36.000000 0.315000
37.000000 0.300000
38.000000 0.290000
39.000000 0.280000
40.000000 0.270000
41.000000 0.260000
42.000000 0.245000
43.000000 0.235000
44.000000 0.225000
45.000000 0.215000
46.000000 0.205000
47.000000 0.200000
48.000000 0.190000
49.000000 0.180000
50.000000 0.170000
51.000000 0.160000
52.000000 0.155000
53.000000 0.145000
54.000000 0.135000
55.000000 0.130000
56.000000 0.120000
57.000000 0.115000
58.000000 0.105000
59.000000 0.100000
60.000000 0.090000
61.000000 0.085000
62.000000 0.080000
63.000000 0.075000
64.000000 0.065000
65.000000 0.060000
66.000000 0.055000
67.000000 0.050000
68.000000 0.045000
69.000000 0.040000
70.000000 0.035000
71.000000 0.030000
72.000000 0.025000
73.000000 0.020000
74.000000 0.020000
75.000000 0.015000
76.000000 0.010000
77.000000 0.010000
78.000000 0.005000
79.000000 0.005000
80.000000 0.000000
};
\addlegendentry{Theorem~\ref{th:upartial}: $u= 30$}
\addplot table {
0.000000 0.855000
1.000000 0.830000
2.000000 0.810000
3.000000 0.790000
4.000000 0.770000
5.000000 0.750000
6.000000 0.730000
7.000000 0.710000
8.000000 0.690000
9.000000 0.675000
10.000000 0.655000
11.000000 0.635000
12.000000 0.620000
13.000000 0.605000
14.000000 0.585000
15.000000 0.570000
16.000000 0.555000
17.000000 0.540000
18.000000 0.525000
19.000000 0.510000
20.000000 0.495000
21.000000 0.480000
22.000000 0.465000
23.000000 0.450000
24.000000 0.440000
25.000000 0.425000
26.000000 0.410000
27.000000 0.400000
28.000000 0.385000
29.000000 0.375000
30.000000 0.360000
31.000000 0.350000
32.000000 0.335000
33.000000 0.325000
34.000000 0.315000
35.000000 0.300000
36.000000 0.290000
37.000000 0.280000
38.000000 0.270000
39.000000 0.260000
40.000000 0.245000
41.000000 0.235000
42.000000 0.225000
43.000000 0.215000
44.000000 0.205000
45.000000 0.200000
46.000000 0.190000
47.000000 0.180000
48.000000 0.170000
49.000000 0.160000
50.000000 0.155000
51.000000 0.145000
52.000000 0.135000
53.000000 0.130000
54.000000 0.120000
55.000000 0.115000
56.000000 0.105000
57.000000 0.100000
58.000000 0.090000
59.000000 0.085000
60.000000 0.080000
61.000000 0.075000
62.000000 0.065000
63.000000 0.060000
64.000000 0.055000
65.000000 0.050000
66.000000 0.045000
67.000000 0.040000
68.000000 0.035000
69.000000 0.030000
70.000000 0.025000
71.000000 0.020000
72.000000 0.020000
73.000000 0.015000
74.000000 0.010000
75.000000 0.010000
76.000000 0.005000
77.000000 0.005000
78.000000 0.000000
79.000000 0.000000
80.000000 -0.005000
};
\addlegendentry{Theorem~\ref{th:upartial}: $u= 40$}
\addplot table {
0.000000 0.830000
1.000000 0.810000
2.000000 0.790000
3.000000 0.770000
4.000000 0.750000
5.000000 0.730000
6.000000 0.710000
7.000000 0.690000
8.000000 0.675000
9.000000 0.655000
10.000000 0.635000
11.000000 0.620000
12.000000 0.605000
13.000000 0.585000
14.000000 0.570000
15.000000 0.555000
16.000000 0.540000
17.000000 0.525000
18.000000 0.510000
19.000000 0.495000
20.000000 0.480000
21.000000 0.465000
22.000000 0.450000
23.000000 0.440000
24.000000 0.425000
25.000000 0.410000
26.000000 0.400000
27.000000 0.385000
28.000000 0.375000
29.000000 0.360000
30.000000 0.350000
31.000000 0.335000
32.000000 0.325000
33.000000 0.315000
34.000000 0.300000
35.000000 0.290000
36.000000 0.280000
37.000000 0.270000
38.000000 0.260000
39.000000 0.245000
40.000000 0.235000
41.000000 0.225000
42.000000 0.215000
43.000000 0.205000
44.000000 0.200000
45.000000 0.190000
46.000000 0.180000
47.000000 0.170000
48.000000 0.160000
49.000000 0.155000
50.000000 0.145000
51.000000 0.135000
52.000000 0.130000
53.000000 0.120000
54.000000 0.115000
55.000000 0.105000
56.000000 0.100000
57.000000 0.090000
58.000000 0.085000
59.000000 0.080000
60.000000 0.075000
61.000000 0.065000
62.000000 0.060000
63.000000 0.055000
64.000000 0.050000
65.000000 0.045000
66.000000 0.040000
67.000000 0.035000
68.000000 0.030000
69.000000 0.025000
70.000000 0.020000
71.000000 0.020000
72.000000 0.015000
73.000000 0.010000
74.000000 0.010000
75.000000 0.005000
76.000000 0.005000
77.000000 0.000000
78.000000 0.000000
79.000000 -0.005000
80.000000 -0.005000
};
\addlegendentry{Theorem~\ref{th:upartial}: $u= 50$}
\addplot table {
0.000000 0.810000
1.000000 0.790000
2.000000 0.770000
3.000000 0.750000
4.000000 0.730000
5.000000 0.710000
6.000000 0.690000
7.000000 0.675000
8.000000 0.655000
9.000000 0.635000
10.000000 0.620000
11.000000 0.605000
12.000000 0.585000
13.000000 0.570000
14.000000 0.555000
15.000000 0.540000
16.000000 0.525000
17.000000 0.510000
18.000000 0.495000
19.000000 0.480000
20.000000 0.465000
21.000000 0.450000
22.000000 0.440000
23.000000 0.425000
24.000000 0.410000
25.000000 0.400000
26.000000 0.385000
27.000000 0.375000
28.000000 0.360000
29.000000 0.350000
30.000000 0.335000
31.000000 0.325000
32.000000 0.315000
33.000000 0.300000
34.000000 0.290000
35.000000 0.280000
36.000000 0.270000
37.000000 0.260000
38.000000 0.245000
39.000000 0.235000
40.000000 0.225000
41.000000 0.215000
42.000000 0.205000
43.000000 0.200000
44.000000 0.190000
45.000000 0.180000
46.000000 0.170000
47.000000 0.160000
48.000000 0.155000
49.000000 0.145000
50.000000 0.135000
51.000000 0.130000
52.000000 0.120000
53.000000 0.115000
54.000000 0.105000
55.000000 0.100000
56.000000 0.090000
57.000000 0.085000
58.000000 0.080000
59.000000 0.075000
60.000000 0.065000
61.000000 0.060000
62.000000 0.055000
63.000000 0.050000
64.000000 0.045000
65.000000 0.040000
66.000000 0.035000
67.000000 0.030000
68.000000 0.025000
69.000000 0.020000
70.000000 0.020000
71.000000 0.015000
72.000000 0.010000
73.000000 0.010000
74.000000 0.005000
75.000000 0.005000
76.000000 0.000000
77.000000 0.000000
78.000000 -0.005000
79.000000 -0.005000
80.000000 -0.005000
};
\addlegendentry{Theorem~\ref{th:upartial}: $u= 60$}
\addplot table {
0.000000 0.790000
1.000000 0.770000
2.000000 0.750000
3.000000 0.730000
4.000000 0.710000
5.000000 0.690000
6.000000 0.675000
7.000000 0.655000
8.000000 0.635000
9.000000 0.620000
10.000000 0.605000
11.000000 0.585000
12.000000 0.570000
13.000000 0.555000
14.000000 0.540000
15.000000 0.525000
16.000000 0.510000
17.000000 0.495000
18.000000 0.480000
19.000000 0.465000
20.000000 0.450000
21.000000 0.440000
22.000000 0.425000
23.000000 0.410000
24.000000 0.400000
25.000000 0.385000
26.000000 0.375000
27.000000 0.360000
28.000000 0.350000
29.000000 0.335000
30.000000 0.325000
31.000000 0.315000
32.000000 0.300000
33.000000 0.290000
34.000000 0.280000
35.000000 0.270000
36.000000 0.260000
37.000000 0.245000
38.000000 0.235000
39.000000 0.225000
40.000000 0.215000
41.000000 0.205000
42.000000 0.200000
43.000000 0.190000
44.000000 0.180000
45.000000 0.170000
46.000000 0.160000
47.000000 0.155000
48.000000 0.145000
49.000000 0.135000
50.000000 0.130000
51.000000 0.120000
52.000000 0.115000
53.000000 0.105000
54.000000 0.100000
55.000000 0.090000
56.000000 0.085000
57.000000 0.080000
58.000000 0.075000
59.000000 0.065000
60.000000 0.060000
61.000000 0.055000
62.000000 0.050000
63.000000 0.045000
64.000000 0.040000
65.000000 0.035000
66.000000 0.030000
67.000000 0.025000
68.000000 0.020000
69.000000 0.020000
70.000000 0.015000
71.000000 0.010000
72.000000 0.010000
73.000000 0.005000
74.000000 0.005000
75.000000 0.000000
76.000000 0.000000
77.000000 -0.005000
78.000000 -0.005000
79.000000 -0.005000
80.000000 -0.005000
};
\addlegendentry{Theorem~\ref{th:upartial}: $u= 70$}
\addplot table {
0.000000 0.750000
1.000000 0.730000
2.000000 0.710000
3.000000 0.690000
4.000000 0.675000
5.000000 0.655000
6.000000 0.635000
7.000000 0.620000
8.000000 0.605000
9.000000 0.585000
10.000000 0.570000
11.000000 0.555000
12.000000 0.540000
13.000000 0.525000
14.000000 0.510000
15.000000 0.495000
16.000000 0.480000
17.000000 0.465000
18.000000 0.450000
19.000000 0.440000
20.000000 0.425000
21.000000 0.410000
22.000000 0.400000
23.000000 0.385000
24.000000 0.375000
25.000000 0.360000
26.000000 0.350000
27.000000 0.335000
28.000000 0.325000
29.000000 0.315000
30.000000 0.300000
31.000000 0.290000
32.000000 0.280000
33.000000 0.270000
34.000000 0.260000
35.000000 0.245000
36.000000 0.235000
37.000000 0.225000
38.000000 0.215000
39.000000 0.205000
40.000000 0.200000
41.000000 0.190000
42.000000 0.180000
43.000000 0.170000
44.000000 0.160000
45.000000 0.155000
46.000000 0.145000
47.000000 0.135000
48.000000 0.130000
49.000000 0.120000
50.000000 0.115000
51.000000 0.105000
52.000000 0.100000
53.000000 0.090000
54.000000 0.085000
55.000000 0.080000
56.000000 0.075000
57.000000 0.065000
58.000000 0.060000
59.000000 0.055000
60.000000 0.050000
61.000000 0.045000
62.000000 0.040000
63.000000 0.035000
64.000000 0.030000
65.000000 0.025000
66.000000 0.020000
67.000000 0.020000
68.000000 0.015000
69.000000 0.010000
70.000000 0.010000
71.000000 0.005000
72.000000 0.005000
73.000000 0.000000
74.000000 0.000000
75.000000 -0.005000
76.000000 -0.005000
77.000000 -0.005000
78.000000 -0.005000
79.000000 -0.005000
80.000000 -0.005000
};
\addlegendentry{Theorem~\ref{th:upartial}: $u= 80$}
\addplot table {
0.000000 0.730000
1.000000 0.710000
2.000000 0.690000
3.000000 0.675000
4.000000 0.655000
5.000000 0.635000
6.000000 0.620000
7.000000 0.605000
8.000000 0.585000
9.000000 0.570000
10.000000 0.555000
11.000000 0.540000
12.000000 0.525000
13.000000 0.510000
14.000000 0.495000
15.000000 0.480000
16.000000 0.465000
17.000000 0.450000
18.000000 0.440000
19.000000 0.425000
20.000000 0.410000
21.000000 0.400000
22.000000 0.385000
23.000000 0.375000
24.000000 0.360000
25.000000 0.350000
26.000000 0.335000
27.000000 0.325000
28.000000 0.315000
29.000000 0.300000
30.000000 0.290000
31.000000 0.280000
32.000000 0.270000
33.000000 0.260000
34.000000 0.245000
35.000000 0.235000
36.000000 0.225000
37.000000 0.215000
38.000000 0.205000
39.000000 0.200000
40.000000 0.190000
41.000000 0.180000
42.000000 0.170000
43.000000 0.160000
44.000000 0.155000
45.000000 0.145000
46.000000 0.135000
47.000000 0.130000
48.000000 0.120000
49.000000 0.115000
50.000000 0.105000
51.000000 0.100000
52.000000 0.090000
53.000000 0.085000
54.000000 0.080000
55.000000 0.075000
56.000000 0.065000
57.000000 0.060000
58.000000 0.055000
59.000000 0.050000
60.000000 0.045000
61.000000 0.040000
62.000000 0.035000
63.000000 0.030000
64.000000 0.025000
65.000000 0.020000
66.000000 0.020000
67.000000 0.015000
68.000000 0.010000
69.000000 0.010000
70.000000 0.005000
71.000000 0.005000
72.000000 0.000000
73.000000 0.000000
74.000000 -0.005000
75.000000 -0.005000
76.000000 -0.005000
77.000000 -0.005000
78.000000 -0.005000
79.000000 -0.005000
80.000000 -0.005000
};
\addlegendentry{Theorem~\ref{th:upartial}: $u= 90$}
\addplot table {
0.000000 0.710000
1.000000 0.690000
2.000000 0.675000
3.000000 0.655000
4.000000 0.635000
5.000000 0.620000
6.000000 0.605000
7.000000 0.585000
8.000000 0.570000
9.000000 0.555000
10.000000 0.540000
11.000000 0.525000
12.000000 0.510000
13.000000 0.495000
14.000000 0.480000
15.000000 0.465000
16.000000 0.450000
17.000000 0.440000
18.000000 0.425000
19.000000 0.410000
20.000000 0.400000
21.000000 0.385000
22.000000 0.375000
23.000000 0.360000
24.000000 0.350000
25.000000 0.335000
26.000000 0.325000
27.000000 0.315000
28.000000 0.300000
29.000000 0.290000
30.000000 0.280000
31.000000 0.270000
32.000000 0.260000
33.000000 0.245000
34.000000 0.235000
35.000000 0.225000
36.000000 0.215000
37.000000 0.205000
38.000000 0.200000
39.000000 0.190000
40.000000 0.180000
41.000000 0.170000
42.000000 0.160000
43.000000 0.155000
44.000000 0.145000
45.000000 0.135000
46.000000 0.130000
47.000000 0.120000
48.000000 0.115000
49.000000 0.105000
50.000000 0.100000
51.000000 0.090000
52.000000 0.085000
53.000000 0.080000
54.000000 0.075000
55.000000 0.065000
56.000000 0.060000
57.000000 0.055000
58.000000 0.050000
59.000000 0.045000
60.000000 0.040000
61.000000 0.035000
62.000000 0.030000
63.000000 0.025000
64.000000 0.020000
65.000000 0.020000
66.000000 0.015000
67.000000 0.010000
68.000000 0.010000
69.000000 0.005000
70.000000 0.005000
71.000000 0.000000
72.000000 0.000000
73.000000 -0.005000
74.000000 -0.005000
75.000000 -0.005000
76.000000 -0.005000
77.000000 -0.005000
78.000000 -0.005000
79.000000 -0.005000
80.000000 -0.005000
};
\addlegendentry{Theorem~\ref{th:upartial}: $u= 100$}

%% file: Coding and Bounds for Partially Defective Memory Cells_arxiv/figures/comparisons/Const1_to_ord_GV/ordGvq7.tex
\addplot table {
0.000000 0.931106
1.000000 0.903032
2.000000 0.874959
3.000000 0.846885
4.000000 0.823491
5.000000 0.800096
6.000000 0.776702
7.000000 0.757986
8.000000 0.734591
9.000000 0.715876
10.000000 0.697160
11.000000 0.678444
12.000000 0.659728
13.000000 0.641013
14.000000 0.622297
15.000000 0.608260
16.000000 0.589545
17.000000 0.575508
18.000000 0.556792
19.000000 0.542755
20.000000 0.528719
21.000000 0.510003
22.000000 0.495966
23.000000 0.481929
24.000000 0.467892
25.000000 0.453856
26.000000 0.439819
27.000000 0.425782
28.000000 0.411745
29.000000 0.402388
30.000000 0.388351
31.000000 0.374314
32.000000 0.364956
33.000000 0.350919
34.000000 0.336883
35.000000 0.327525
36.000000 0.313488
37.000000 0.304130
38.000000 0.294772
39.000000 0.280735
40.000000 0.271378
41.000000 0.262020
42.000000 0.252662
43.000000 0.238625
44.000000 0.229267
45.000000 0.219909
46.000000 0.210552
47.000000 0.201194
48.000000 0.191836
49.000000 0.182478
50.000000 0.177799
51.000000 0.168441
52.000000 0.159083
53.000000 0.149726
54.000000 0.145047
55.000000 0.135689
56.000000 0.126331
57.000000 0.121652
58.000000 0.112294
59.000000 0.107615
60.000000 0.098257
61.000000 0.093578
62.000000 0.084221
63.000000 0.079542
64.000000 0.074863
65.000000 0.070184
66.000000 0.060826
67.000000 0.056147
68.000000 0.051468
69.000000 0.046789
70.000000 0.042110
71.000000 0.037431
72.000000 0.032752
73.000000 0.028074
74.000000 0.023395
75.000000 0.023395
76.000000 0.018716
77.000000 0.014037
78.000000 0.014037
79.000000 0.009358
80.000000 0.009358
};
\addlegendentry{$(q-1)$-ary GV \\
		($u$ arbitrary)}

%% file: figures/comparisons/Trading_Com3_Th11/q_minus_1_ord_GVq7_few.tex
\addplot table {
0.000000 0.931106
1.000000 0.903032
2.000000 0.874959
4.000000 0.823491
5.000000 0.800096
10.000000 0.697160
20.000000 0.528719
};
\addlegendentry{$(q-1)$-ary GV \\ ($u$ arbitrary)}

%% file: figures/comparisons/Trading_Com3_Th11/plot_list_constbin_few.tex
\addplot table {
0.000000 0.975000
1.000000 0.945000
2.000000 0.915000
4.000000 0.865000
5.000000 0.840000
10.000000 0.735000
20.000000 0.555000
};
\addlegendentry{Theorem~\ref{Th:GV-s2}, $u= 10$}
\addplot table {
0.000000 0.926667
1.000000 0.896667
2.000000 0.866667
4.000000 0.816667
5.000000 0.791667
10.000000 0.686667
20.000000 0.506667
};
\addlegendentry{Theorem~\ref{Th:GV-s2}, $u= 30$}
\addplot table {
0.000000 0.903333	
1.000000 0.873333
2.000000 0.843333
4.000000 0.793333
5.000000 0.768333
10.000000 0.663333
20.000000 0.483333
};
\addlegendentry{Theorem~\ref{Th:GV-s2}, $u= 40$}

%% file: figures/comparisons/Trading_Com3_Th11/GV_treat_u_as_t.tex
\addplot table {
0.000000 0.960000
1.000000 0.930000
2.000000 0.905000
4.000000 0.855000
5.000000 0.830000
10.000000 0.730000
20.000000 0.555000
};
\addlegendentry{Theorem~\ref{th:upartial}, $u= 10$}
\addplot table {
0.000000 0.905000
1.000000 0.880000
2.000000 0.855000
4.000000 0.810000
5.000000 0.790000
10.000000 0.690000
20.000000 0.525000
};
\addlegendentry{Theorem~\ref{th:upartial}, $u= 30$}
\addplot table {
0.000000 0.855000
1.000000 0.830000
2.000000 0.810000
4.000000 0.770000
5.000000 0.750000
10.000000 0.655000
20.000000 0.495000
};
\addlegendentry{Theorem~\ref{th:upartial}, $u= 40$}

%% file: figures/comparisons/Trading_Theorem_12/plot_list_const_2.tex
\addplot table {
0.000000 0.845000
1.000000 0.815000
2.000000 0.785000
3.000000 0.760000
4.000000 0.735000
5.000000 0.710000
6.000000 0.685000
7.000000 0.665000
8.000000 0.645000
9.000000 0.625000
10.000000 0.605000
11.000000 0.585000
12.000000 0.565000
13.000000 0.545000
14.000000 0.530000
15.000000 0.510000
16.000000 0.490000
17.000000 0.475000
18.000000 0.460000
19.000000 0.440000
20.000000 0.425000
21.000000 0.410000
22.000000 0.395000
23.000000 0.380000
24.000000 0.365000
25.000000 0.350000
26.000000 0.335000
27.000000 0.320000
28.000000 0.305000
29.000000 0.295000
30.000000 0.280000
31.000000 0.265000
32.000000 0.255000
33.000000 0.240000
34.000000 0.230000
35.000000 0.215000
36.000000 0.205000
37.000000 0.190000
38.000000 0.180000
39.000000 0.170000
40.000000 0.155000
41.000000 0.145000
42.000000 0.135000
43.000000 0.125000
44.000000 0.115000
45.000000 0.100000
46.000000 0.090000
47.000000 0.080000
48.000000 0.070000
49.000000 0.060000
50.000000 0.055000
51.000000 0.045000
};
\addlegendentry{Theorem~\ref{th12_new} $(u, t): u= 17$}
\addplot table {
0.000000 0.790000
1.000000 0.760000
2.000000 0.730000
3.000000 0.705000
4.000000 0.680000
5.000000 0.655000
6.000000 0.630000
7.000000 0.610000
8.000000 0.590000
9.000000 0.570000
10.000000 0.550000
11.000000 0.530000
12.000000 0.510000
13.000000 0.490000
14.000000 0.475000
15.000000 0.455000
16.000000 0.435000
17.000000 0.420000
18.000000 0.405000
19.000000 0.385000
20.000000 0.370000
21.000000 0.355000
22.000000 0.340000
23.000000 0.325000
24.000000 0.310000
25.000000 0.295000
26.000000 0.280000
27.000000 0.265000
28.000000 0.250000
29.000000 0.240000
30.000000 0.225000
31.000000 0.210000
32.000000 0.200000
33.000000 0.185000
34.000000 0.175000
35.000000 0.160000
36.000000 0.150000
37.000000 0.135000
38.000000 0.125000
39.000000 0.115000
40.000000 0.100000
41.000000 0.090000
42.000000 0.080000
43.000000 0.070000
44.000000 0.060000
45.000000 0.045000
46.000000 0.035000
47.000000 0.025000
48.000000 0.015000
49.000000 0.005000
50.000000 -0.100000
51.000000 -0.100000
};
\addlegendentry{Theorem~\ref{th12_new} $(u, t): u= 22$}
\addplot table {
0.000000 0.750000
1.000000 0.720000
2.000000 0.690000
3.000000 0.665000
4.000000 0.640000
5.000000 0.615000
6.000000 0.590000
7.000000 0.570000
8.000000 0.550000
9.000000 0.530000
10.000000 0.510000
11.000000 0.490000
12.000000 0.470000
13.000000 0.450000
14.000000 0.435000
15.000000 0.415000
16.000000 0.395000
17.000000 0.380000
18.000000 0.365000
19.000000 0.345000
20.000000 0.330000
21.000000 0.315000
22.000000 0.300000
23.000000 0.285000
24.000000 0.270000
25.000000 0.255000
26.000000 0.240000
27.000000 0.225000
28.000000 0.210000
29.000000 0.200000
30.000000 0.185000
31.000000 0.170000
32.000000 0.160000
33.000000 0.145000
34.000000 0.135000
35.000000 0.120000
36.000000 0.110000
37.000000 0.095000
38.000000 0.085000
39.000000 0.075000
40.000000 0.060000
41.000000 0.050000
42.000000 0.040000
43.000000 0.030000
44.000000 0.020000
45.000000 0.005000
46.000000 -0.100000
47.000000 -0.100000
48.000000 -0.100000
49.000000 -0.100000
50.000000 -0.100000
51.000000 -0.100000
};
\addlegendentry{Theorem~\ref{th12_new} $(u, t): u= 26$}

%% file: figures/comparisons/Trading_Theorem_12/plot_list_after_const2_trading.tex
\addplot table {
0.000000 0.830000
1.000000 0.800000
2.000000 0.775000
3.000000 0.750000
4.000000 0.725000
5.000000 0.700000
6.000000 0.680000
7.000000 0.660000
8.000000 0.640000
9.000000 0.620000
10.000000 0.600000
11.000000 0.580000
12.000000 0.560000
13.000000 0.545000
14.000000 0.525000
15.000000 0.505000
16.000000 0.490000
17.000000 0.475000
18.000000 0.455000
19.000000 0.440000
20.000000 0.425000
21.000000 0.410000
22.000000 0.395000
23.000000 0.380000
24.000000 0.365000
25.000000 0.350000
26.000000 0.335000
27.000000 0.320000
28.000000 0.315000
29.000000 0.295000
30.000000 0.280000
31.000000 0.275000
32.000000 0.255000
33.000000 0.250000
34.000000 0.230000
35.000000 0.225000
36.000000 0.205000
37.000000 0.200000
38.000000 0.190000
39.000000 0.170000
40.000000 0.165000
41.000000 0.155000
42.000000 0.145000
43.000000 0.135000
44.000000 0.115000
45.000000 0.110000
46.000000 0.100000
47.000000 0.090000
48.000000 0.080000
49.000000 0.080000
50.000000 0.065000
};
\addlegendentry{Theorem~\ref{tradingu_t} $(u+1, t-1): u= 16$}
\addplot table {
0.000000 0.775000
1.000000 0.745000
2.000000 0.720000
3.000000 0.695000
4.000000 0.670000
5.000000 0.645000
6.000000 0.625000
7.000000 0.605000
8.000000 0.585000
9.000000 0.565000
10.000000 0.545000
11.000000 0.525000
12.000000 0.505000
13.000000 0.490000
14.000000 0.470000
15.000000 0.450000
16.000000 0.435000
17.000000 0.420000
18.000000 0.400000
19.000000 0.385000
20.000000 0.370000
21.000000 0.355000
22.000000 0.340000
23.000000 0.325000
24.000000 0.310000
25.000000 0.295000
26.000000 0.280000
27.000000 0.265000
28.000000 0.260000
29.000000 0.240000
30.000000 0.225000
31.000000 0.220000
32.000000 0.200000
33.000000 0.195000
34.000000 0.175000
35.000000 0.170000
36.000000 0.150000
37.000000 0.145000
38.000000 0.135000
39.000000 0.115000
40.000000 0.110000
41.000000 0.100000
42.000000 0.090000
43.000000 0.080000
44.000000 0.060000
45.000000 0.055000
46.000000 0.045000
47.000000 0.035000
48.000000 0.025000
49.000000 0.025000
50.000000 -0.100000
};
\addlegendentry{Theorem~\ref{tradingu_t} $(u+1, t-1): u= 21$}
\addplot table {
0.000000 0.730000
1.000000 0.700000
2.000000 0.675000
3.000000 0.650000
4.000000 0.625000
5.000000 0.600000
6.000000 0.580000
7.000000 0.560000
8.000000 0.540000
9.000000 0.520000
10.000000 0.500000
11.000000 0.480000
12.000000 0.460000
13.000000 0.445000
14.000000 0.425000
15.000000 0.405000
16.000000 0.390000
17.000000 0.375000
18.000000 0.355000
19.000000 0.340000
20.000000 0.325000
21.000000 0.310000
22.000000 0.295000
23.000000 0.280000
24.000000 0.265000
25.000000 0.250000
26.000000 0.235000
27.000000 0.220000
28.000000 0.210000
29.000000 0.195000
30.000000 0.180000
31.000000 0.170000
32.000000 0.155000
33.000000 0.145000
34.000000 0.130000
35.000000 0.120000
36.000000 0.105000
37.000000 0.095000
38.000000 0.085000
39.000000 0.070000
40.000000 0.060000
41.000000 0.050000
42.000000 0.040000
43.000000 0.030000
44.000000 0.015000
45.000000 0.005000
46.000000 -0.100000
47.000000 -0.100000
48.000000 -0.100000
49.000000 -0.100000
50.000000 -0.100000
};
\addlegendentry{Theorem~\ref{tradingu_t} $(u+1, t-1): u= 25$}

%% file: figures/comparisons/Trading_Theorem_13/plot_list_constbin_rates.tex
\addplot table {
0.000000 0.981667
1.000000 0.951667
2.000000 0.921667
3.000000 0.896667
4.000000 0.871667
5.000000 0.846667
6.000000 0.821667
7.000000 0.801667
8.000000 0.781667
9.000000 0.761667
10.000000 0.741667
11.000000 0.721667
12.000000 0.701667
13.000000 0.681667
14.000000 0.666667
15.000000 0.646667
16.000000 0.626667
17.000000 0.611667
18.000000 0.596667
19.000000 0.576667
20.000000 0.561667
21.000000 0.546667
22.000000 0.531667
23.000000 0.516667
24.000000 0.501667
25.000000 0.486667
26.000000 0.471667
27.000000 0.456667
28.000000 0.441667
29.000000 0.431667
30.000000 0.416667
31.000000 0.401667
32.000000 0.391667
33.000000 0.376667
34.000000 0.366667
35.000000 0.351667
36.000000 0.341667
37.000000 0.326667
38.000000 0.316667
39.000000 0.000000
40.000000 0.000000
41.000000 0.000000
42.000000 0.000000
43.000000 0.000000
44.000000 0.000000
45.000000 0.000000
46.000000 0.000000
47.000000 0.000000
48.000000 0.000000
49.000000 0.000000
50.000000 0.000000
};
\addlegendentry{Theorem~\ref{Th:GV-s2} $(u, t): u= 8,\dots,11$}	
\addplot table {
0.000000 0.970000
1.000000 0.940000
2.000000 0.910000
3.000000 0.885000
4.000000 0.860000
5.000000 0.835000
6.000000 0.810000
7.000000 0.790000
8.000000 0.770000
9.000000 0.750000
10.000000 0.730000
11.000000 0.710000
12.000000 0.690000
13.000000 0.670000
14.000000 0.655000
15.000000 0.635000
16.000000 0.615000
17.000000 0.600000
18.000000 0.585000
19.000000 0.565000
20.000000 0.550000
21.000000 0.535000
22.000000 0.520000
23.000000 0.505000
24.000000 0.490000
25.000000 0.475000
26.000000 0.460000
27.000000 0.445000
28.000000 0.430000
29.000000 0.420000
30.000000 0.405000
31.000000 0.390000
32.000000 0.380000
33.000000 0.365000
34.000000 0.355000
35.000000 0.000000
36.000000 0.000000
37.000000 0.000000
38.000000 0.000000
39.000000 0.000000
40.000000 0.000000
41.000000 0.000000
42.000000 0.000000
43.000000 0.000000
44.000000 0.000000
45.000000 0.000000
46.000000 0.000000
47.000000 0.000000
48.000000 0.000000
49.000000 0.000000
50.000000 0.000000
};
\addlegendentry{Theorem~\ref{Th:GV-s2} $(u, t): u= 12,\dots,15$}
\addplot table {
0.000000 0.960000
1.000000 0.930000
2.000000 0.900000
3.000000 0.875000
4.000000 0.850000
5.000000 0.825000
6.000000 0.800000
7.000000 0.780000
8.000000 0.760000
9.000000 0.740000
10.000000 0.720000
11.000000 0.700000
12.000000 0.680000
13.000000 0.660000
14.000000 0.645000
15.000000 0.625000
16.000000 0.605000
17.000000 0.590000
18.000000 0.575000
19.000000 0.555000
20.000000 0.540000
21.000000 0.525000
22.000000 0.510000
23.000000 0.495000
24.000000 0.480000
25.000000 0.465000
26.000000 0.450000
27.000000 0.435000
28.000000 0.420000
29.000000 0.410000
30.000000 0.395000
31.000000 0.380000
32.000000 0.000000
33.000000 0.000000
34.000000 0.000000
35.000000 0.000000
36.000000 0.000000
37.000000 0.000000
38.000000 0.000000
39.000000 0.000000
40.000000 0.000000
41.000000 0.000000
42.000000 0.000000
43.000000 0.000000
44.000000 0.000000
45.000000 0.000000
46.000000 0.000000
47.000000 0.000000
48.000000 0.000000
49.000000 0.000000
50.000000 0.000000
};
\addlegendentry{Theorem~\ref{Th:GV-s2} $(u, t): u= 16,\dots,19$}
\addplot table {
0.000000 0.951667
1.000000 0.921667
2.000000 0.891667
3.000000 0.866667
4.000000 0.841667
5.000000 0.816667
6.000000 0.791667
7.000000 0.771667
8.000000 0.751667
9.000000 0.731667
10.000000 0.711667
11.000000 0.691667
12.000000 0.671667
13.000000 0.651667
14.000000 0.636667
15.000000 0.616667
16.000000 0.596667
17.000000 0.581667
18.000000 0.566667
19.000000 0.546667
20.000000 0.531667
21.000000 0.516667
22.000000 0.501667
23.000000 0.486667
24.000000 0.471667
25.000000 0.456667
26.000000 0.441667
27.000000 0.426667
28.000000 0.411667
29.000000 0.401667
30.000000 0.000000
31.000000 0.000000
32.000000 0.000000
33.000000 0.000000
34.000000 0.000000
35.000000 0.000000
36.000000 0.000000
37.000000 0.000000
38.000000 0.000000
39.000000 0.000000
40.000000 0.000000
41.000000 0.000000
42.000000 0.000000
43.000000 0.000000
44.000000 0.000000
45.000000 0.000000
46.000000 0.000000
47.000000 0.000000
48.000000 0.000000
49.000000 0.000000
50.000000 0.000000
};
\addlegendentry{Theorem~\ref{Th:GV-s2} $(u, t): u= 20,\dots,23$}
\addplot table {
0.000000 0.941667
1.000000 0.911667
2.000000 0.881667
3.000000 0.856667
4.000000 0.831667
5.000000 0.806667
6.000000 0.781667
7.000000 0.761667
8.000000 0.741667
9.000000 0.721667
10.000000 0.701667
11.000000 0.681667
12.000000 0.661667
13.000000 0.641667
14.000000 0.626667
15.000000 0.606667
16.000000 0.586667
17.000000 0.571667
18.000000 0.556667
19.000000 0.536667
20.000000 0.521667
21.000000 0.506667
22.000000 0.491667
23.000000 0.476667
24.000000 0.461667
25.000000 0.446667
26.000000 0.431667
27.000000 0.416667
28.000000 0.000000
29.000000 0.000000
30.000000 0.000000
31.000000 0.000000
32.000000 0.000000
33.000000 0.000000
34.000000 0.000000
35.000000 0.000000
36.000000 0.000000
37.000000 0.000000
38.000000 0.000000
39.000000 0.000000
40.000000 0.000000
41.000000 0.000000
42.000000 0.000000
43.000000 0.000000
44.000000 0.000000
45.000000 0.000000
46.000000 0.000000
47.000000 0.000000
48.000000 0.000000
49.000000 0.000000
50.000000 0.000000
};
\addlegendentry{Theorem~\ref{Th:GV-s2} $(u, t): u= 24,\dots,27$}

%% file: figures/comparisons/Trading_Theorem_13/plot_list_after_constbin_trading_rates.tex
\addplot table {
0.000000 0.930000
1.000000 0.900000
2.000000 0.875000
3.000000 0.850000
4.000000 0.825000
5.000000 0.800000
6.000000 0.780000
7.000000 0.760000
8.000000 0.740000
9.000000 0.720000
10.000000 0.700000
11.000000 0.680000
12.000000 0.660000
13.000000 0.645000
14.000000 0.625000
15.000000 0.605000
16.000000 0.590000
17.000000 0.575000
18.000000 0.555000
19.000000 0.540000
20.000000 0.525000
21.000000 0.510000
22.000000 0.495000
23.000000 0.480000
24.000000 0.465000
25.000000 0.450000
26.000000 0.435000
27.000000 0.420000
28.000000 0.410000
29.000000 0.395000
30.000000 0.380000
31.000000 0.000000
32.000000 0.000000
33.000000 0.000000
34.000000 0.000000
35.000000 0.000000
36.000000 0.000000
37.000000 0.000000
38.000000 0.000000
39.000000 0.000000
40.000000 0.000000
41.000000 0.000000
42.000000 0.000000
43.000000 0.000000
44.000000 0.000000
45.000000 0.000000
46.000000 0.000000
47.000000 0.000000
48.000000 0.000000
49.000000 0.000000
50.000000 0.000000
};
\addlegendentry{Theorem~\ref{tradingu_t} $(u+1, t-1): u= 19$}
\addplot table {
0.000000 0.921667
1.000000 0.891667
2.000000 0.866667
3.000000 0.841667
4.000000 0.816667
5.000000 0.791667
6.000000 0.771667
7.000000 0.751667
8.000000 0.731667
9.000000 0.711667
10.000000 0.691667
11.000000 0.671667
12.000000 0.651667
13.000000 0.636667
14.000000 0.616667
15.000000 0.596667
16.000000 0.581667
17.000000 0.566667
18.000000 0.546667
19.000000 0.531667
20.000000 0.516667
21.000000 0.501667
22.000000 0.486667
23.000000 0.471667
24.000000 0.456667
25.000000 0.441667
26.000000 0.426667
27.000000 0.411667
28.000000 0.401667
29.000000 0.000000
30.000000 0.000000
31.000000 0.000000
32.000000 0.000000
33.000000 0.000000
34.000000 0.000000
35.000000 0.000000
36.000000 0.000000
37.000000 0.000000
38.000000 0.000000
39.000000 0.000000
40.000000 0.000000
41.000000 0.000000
42.000000 0.000000
43.000000 0.000000
44.000000 0.000000
45.000000 0.000000
46.000000 0.000000
47.000000 0.000000
48.000000 0.000000
49.000000 0.000000
50.000000 0.000000
};
\addlegendentry{Theorem~\ref{tradingu_t} $(u+1, t-1): u= 23$}
\addplot table {
0.000000 0.911667
1.000000 0.881667
2.000000 0.856667
3.000000 0.831667
4.000000 0.806667
5.000000 0.781667
6.000000 0.761667
7.000000 0.741667
8.000000 0.721667
9.000000 0.701667
10.000000 0.681667
11.000000 0.661667
12.000000 0.641667
13.000000 0.626667
14.000000 0.606667
15.000000 0.586667
16.000000 0.571667
17.000000 0.556667
18.000000 0.536667
19.000000 0.521667
20.000000 0.506667
21.000000 0.491667
22.000000 0.476667
23.000000 0.461667
24.000000 0.446667
25.000000 0.431667
26.000000 0.416667
27.000000 0.000000
28.000000 0.000000
29.000000 0.000000
30.000000 0.000000
31.000000 0.000000
32.000000 0.000000
33.000000 0.000000
34.000000 0.000000
35.000000 0.000000
36.000000 0.000000
37.000000 0.000000
38.000000 0.000000
39.000000 0.000000
40.000000 0.000000
41.000000 0.000000
42.000000 0.000000
43.000000 0.000000
44.000000 0.000000
45.000000 0.000000
46.000000 0.000000
47.000000 0.000000
48.000000 0.000000
49.000000 0.000000
50.000000 0.000000
};
\addlegendentry{Theorem~\ref{tradingu_t} $(u+1, t-1): u= 27$}

%% file: figures/comparisons/Trading_Theorem_13/plot_list_after_constbin_trading_lemma_rates.tex
\addplot table {
	0.000000 0.930000
	1.000000 0.900000
	2.000000 0.875000
	3.000000 0.850000
	4.000000 0.825000
	5.000000 0.800000
	6.000000 0.780000
	7.000000 0.760000
	8.000000 0.740000
	9.000000 0.720000
	10.000000 0.700000
	11.000000 0.680000
	12.000000 0.660000
	13.000000 0.645000
	14.000000 0.625000
	15.000000 0.605000
	16.000000 0.590000
	17.000000 0.575000
	18.000000 0.555000
	19.000000 0.540000
	20.000000 0.525000
	21.000000 0.510000
	22.000000 0.495000
	23.000000 0.480000
	24.000000 0.465000
	25.000000 0.450000
	26.000000 0.435000
	27.000000 0.420000
	28.000000 0.410000
	29.000000 0.395000
	30.000000 0.380000
	31.000000 0.000000
	32.000000 0.000000
	33.000000 0.000000
	34.000000 0.000000
	35.000000 0.000000
	36.000000 0.000000
	37.000000 0.000000
	38.000000 0.000000
	39.000000 0.000000
	40.000000 0.000000
	41.000000 0.000000
	42.000000 0.000000
	43.000000 0.000000
	44.000000 0.000000
	45.000000 0.000000
	46.000000 0.000000
	47.000000 0.000000
	48.000000 0.000000
	49.000000 0.000000
	50.000000 0.000000
};
\addlegendentry{Lemma~\ref{lem03} $(u+4, t-1): u= 19$}

%% file: figures/comparisons/Trading_Theorem_13/plot_list_after_constbin_trading_lemma2_rates.tex
\addplot table {
0.000000 0.990000
1.000000 0.960000
2.000000 0.930000
3.000000 0.905000
4.000000 0.880000
5.000000 0.855000
6.000000 0.830000
7.000000 0.810000
8.000000 0.790000
9.000000 0.770000
10.000000 0.750000
11.000000 0.730000
12.000000 0.710000
13.000000 0.690000
14.000000 0.675000
15.000000 0.655000
16.000000 0.635000
17.000000 0.620000
18.000000 0.605000
19.000000 0.585000
20.000000 0.570000
21.000000 0.555000
22.000000 0.540000
23.000000 0.525000
24.000000 0.510000
25.000000 0.495000
26.000000 0.480000
27.000000 0.465000
28.000000 0.450000
29.000000 0.440000
30.000000 0.425000
31.000000 0.410000
32.000000 0.400000
33.000000 0.385000
34.000000 0.375000
35.000000 0.360000
36.000000 0.350000
37.000000 0.335000
38.000000 0.325000
39.000000 0.315000
40.000000 0.300000
41.000000 0.290000
42.000000 0.280000
43.000000 0.270000
44.000000 0.260000
45.000000 0.245000
46.000000 0.235000
47.000000 0.225000
48.000000 0.215000
49.000000 0.205000
50.000000 0.200000
};
\addlegendentry{
Lemma~\ref{lem01} $(u+8, t-1): u= q-1$}